\documentclass[a4paper]{article}
\usepackage{cbiffp}
\renewbibmacro{in:}{}
\bibliography{cbiffp}

\title{Comparison-Based Indexing\\ From First Principles}
\author{Magnus Lie Hetland}
\date{\mbox{}}

\begin{document}
\maketitle

\begin{abstract}
\noindent
    Basic assumptions about comparison-based indexing are laid down and a
    general design space is derived from these. An index structure spanning
    this design space (the \emph{sprawl}) is described, along with an
    associated family of partitioning predicates, or regions (the
    \emph{ambits}), as well as algorithms for search and, to some extent,
    construction. The \emph{sprawl of ambits} forms a unification and
    generalization of current indexing methods, and a jumping-off point for
    future designs.
\end{abstract}

\section{Introduction}
To speed up search operations, it is common to preprocess data by constructing
an index structure.
In many cases, access to the data is limited to some form of comparison, such
as orders or similarity measures.
The index then represents possible incremental explorations, with comparisons
performed along the way.
Similarity search has been a particularly fertile ground for such structures,
with a great variety appearing over the last few decades.
These indices are based on a range of disparate ideas, with few overarching
organizing principles, possibly leading to missed opportunities for
intermediate, hybrid or variant structures.
My aim is to map out a design space that springs from a handful of basic
assumptions, covering all existing methods, as well as countless new
variations, combinations and extensions.
To describe this design space, I introduce the \emph{sprawl}, a fully general
yet fully implementable comparison-based index in the form of a region-labeled
hyperdigraph, along with a new family of regions, the \emph{ambits}, a common
generalization of currently used comparison-based regions, and of weighted
polyellipses, polytopes and metaballs (see \cref{fig:toc}).%
\endnote{The etymology of the new terms is straightforward. In existing usage,
\emph{sprawl} can mean
\begin{inl}
\item an aggregation or
network of regions, as in an \emph{urban sprawl}, or
\item to spread in a rambling and irregular way,
\end{inl}
while \emph{ambit} is another word for \emph{extent}, \emph{reach}, or
\emph{sphere of influence}.}
This work opens up two avenues of further exploration in the search for new
and better methods: Either one may accept my explicitly stated assumptions and
examines the resulting design space, or one may explicitly depart from them,
and thus have a jumping-off point for wholly new designs.
{
\def\xyscale{1.7}%
\def\D{.5}%
\def\outline{(-\D - .1, -\D - .2) rectangle +(3*\D + .2, 2*\D + .4)}

\def\P[#1]{\csuse{P[#1]}}%
\csdef{P[1,1]}{-2.6492}%
\csdef{P[2,1]}{-1.875}%
\csdef{P[1,2]}{-2.6492}%
\csdef{P[2,2]}{0.0}%
\csdef{P[1,3]}{-2.6492}%
\csdef{P[2,3]}{1.875}%
\csdef{P[1,4]}{0.0}%
\csdef{P[2,4]}{0.0}%
\csdef{P[1,5]}{0.0}%
\csdef{P[2,5]}{-1.875}%
\csdef{P[1,6]}{0.0}%
\csdef{P[2,6]}{1.875}%
\csdef{P[1,7]}{2.6492}%
\csdef{P[2,7]}{-2.25}%
\csdef{P[1,8]}{2.6492}%
\csdef{P[2,8]}{-0.75}%
\csdef{P[1,9]}{2.6492}%
\csdef{P[2,9]}{0.75}%
\csdef{P[1,10]}{2.6492}%
\csdef{P[2,10]}{2.25}%
\csdef{n}{10}%
\csdef{figcount}{4}%

\providetoggle{isfocus[1]}\settoggle{isfocus[1]}{true}%
\providetoggle{isfocus[2]}\settoggle{isfocus[2]}{true}%
\providetoggle{isfocus[3]}\settoggle{isfocus[3]}{true}%
\providetoggle{isfocus[4]}\settoggle{isfocus[4]}{true}%
\providetoggle{isfocus[5]}\settoggle{isfocus[5]}{true}%
\providetoggle{isfocus[6]}\settoggle{isfocus[6]}{true}%
\providetoggle{isfocus[7]}\settoggle{isfocus[7]}{false}%
\providetoggle{isfocus[8]}\settoggle{isfocus[8]}{false}%
\providetoggle{isfocus[9]}\settoggle{isfocus[9]}{false}%
\providetoggle{isfocus[10]}\settoggle{isfocus[10]}{false}%

\tikzset{
    larger/.style = {
        scale = 1.3
    },
}
\tikzmath{
    \capwidth = 1.3*\xyscale*(3*\D + .2);
}
\begin{figure}%
\begin{minipage}[t]{\capwidth cm}
\vspace{0pt}%
\caption{A sprawl of ambits. The regions act as
hyperedges---gatekeepers for their targets, defined in part by their sources.%
\endnote{The figure simplifies things, by having a one-to-one
correspondence between regions and hyperedges. In general, each hyperedge
could have multiple regions, though in many cases this would not be
necessary (cf.\ \cref{rem:singleregions}).} In a realistic setting, the
regions would typically contain their target points}%
\label{fig:toc}%
\end{minipage}%
\hfill
\begin{tikzpicture}[baseline=1.3*\xyscale*(\D cm + .2 cm),
    larger, xscale=\xyscale, yscale=\xyscale]

    \clip \outline;
    \fill[use as bounding box, lightshade] \outline;

    \colorlet{outline}{lightshade}
    \tikzset{
        region/.style = {
            draw, thick, fill=darkshade,
            preaction={overlay, draw=lightshade, line width=3pt}
        },
        edge/.style = {
            semithick,
            -Stealthy
        },
        point/.style = {
            circle,
            minimum size=2.5pt,
            inner sep=0pt,
            draw=outline,
            line width=3pt,
            overlay,
            postaction={
                overlay=false,
                thin,
                fill=black,
                draw=black,
            }
        }
    }

    \draw (.5*\D,0) node {\tikz[baseline=0, larger,
        xscale=.2*\xyscale, yscale=.2*\xyscale
        ]{%
    \begin{pgfonlayer}{foreground}
        \foreach \i in {1,...,\n} {
            \iftoggle{isfocus[\i]}{
                \colorlet{outline}{darkshade}
            }{
            }
            \draw ({\P[1,\i]},{\P[2,\i]}) node[point] (P\i) {};
        }
    \end{pgfonlayer}

    \draw[edge] (P5) to (P7);
    \draw[edge] ($(P5)!.5!(P4)$) to (P8);
    \draw[edge] (P4) to (P9);
    \draw[edge] (P6) edge (P9) to (P10);

    \draw[region] plot file {tocfig2.dat};
    \draw[region] plot file {tocfig3.dat};

    \def\arw#1#2{

        \begin{scope}
            \path[invclip]
                plot file {tocfig2.dat}
                plot file {tocfig3.dat}
                ;
            \draw[line width=3pt, lightshade] (#1) to (#2);
        \end{scope}

        \begin{scope}
            \path[clip]
                plot file {tocfig2.dat}
                plot file {tocfig3.dat}
                ;
            \draw[line width=3pt, darkshade] (#1) to (#2);
        \end{scope}

        \draw[edge] (#1) to (#2);

    }

    \arw{$(P1)!.2!(P2)$}{P5}
    \arw{P2}{P4}
    \arw{$(P3)!.2!(P2)$}{P6}

    \draw[region]
        plot file {tocfig1.dat};

    \draw ($(P4)!.5!(P5)$) +(.15pt, .2pt) node[font=\footnotesize] {$R$};

    }};

\end{tikzpicture}%
\hfill
\begin{tikzpicture}[baseline=1.3*\xyscale*(\D cm + .2 cm),
    larger, xscale=\xyscale, yscale=\xyscale]

    \tikzset{
        region/.style = {
            circle, draw,
            semithick,
            fill=darkshade, inner sep=0, minimum size=8pt
        },
        point/.style = {
            point node
        },
        edge/.style = {
            semithick,
            line cap=round, -Stealthy,
            shorten >=1.5pt
        },
        LR/.style={-, shorten <=2pt, shorten >=0pt}
    }

    \path[use as bounding box] \outline;
    \draw[thin]
        (-\D cm - .1 cm + .5\pgflinewidth,
         -\D cm - .2 cm + .5\pgflinewidth)
        rectangle
        +(3*\D cm + .2 cm - \pgflinewidth,
          2*\D cm + .4 cm - \pgflinewidth)
    ;

    \begin{scope}[xshift=.5*\D cm]
    \draw[scale=.2, x=1.5cm, y=1.5cm]

        ( 2,  1.5) node[point] (leaf1) {}
        ( 2,  0.5) node[point] (leaf2) {}
        ( 2, -0.5) node[point] (leaf3) {}
        ( 2, -1.5) node[point] (leaf4) {}

        ( 1,  1.0) node[region] (region1) {}
        ( 1, -0.5) node[region] (region2) {}

        ( 0,  1.0) node[point] (internal1) {}
        ( 0,  0.0) node[point] (internal2) {}
        ( 0, -1.0) node[point] (internal3) {}

        (-1,  0.0) node[region] (region3) {}

        (-2,  1.25) node[point] (root1) {}
        (-2,  0.0) node[point] (root2) {}
        (-2, -1.25) node[point] (root3) {}

        ;
    \end{scope}

    \begin{scope}[every edge/.append style={edge}]

        \draw

            (region1) edge (leaf1) edge (leaf2)
            (region2) edge (leaf2) edge (leaf3) edge (leaf4)

            (internal1) edge[LR] (region1)
            (internal2) edge[LR] (region2)
            (internal3) edge[LR] (region2)

            (region3) edge (internal1) edge (internal2) edge (internal3)

            (root1) edge[LR] (region3)
            (root2) edge[LR] (region3)
            (root3) edge[LR] (region3)

            ;

    \end{scope}

    \draw
    (region2 |- leaf4) +(.5pt, 0)
    node[font=\footnotesize, inner sep=1.5pt] (R) {$R$};
    \draw[thin] (region2) to[in=80, out=-105] (R);

\end{tikzpicture}%
\end{figure}
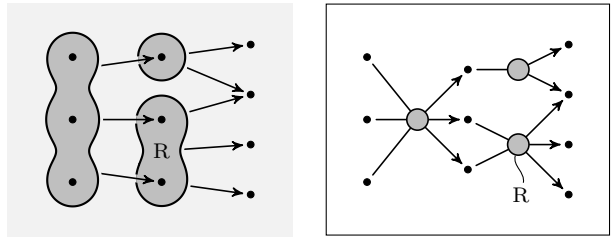

}
As a side-effect, it also introduces a new query modality for existing
similarity-based search, where one may search for \emph{weighted combinations}
of sample objects, possibly including negative weights for contrasting
objects~(cf.\ \ref{rem:ambitqueries}).

\subsection{Scope, Organization and Contributions}

The scope of this paper is defined formally in \cref{sec:defs}, but
intuitively my focus is on indexing methods that rely on comparing pairs of
objects using some black-box oracle. That is, performing computations directly
on the objects themselves, as in most hashing-based methods, is out of scope.
Also, I only consider exact search, where false positives are permitted, but
false negatives are not. Metric indexing~\citep{Zezula:2006} is a prototypical
example, and one of the main motivations behind this work.\footnote{I include
here also the less common \emph{quasi}-metric indexing \citep[see,
e.g.,][]{Pestov:2006}.}

The basic definitions and axioms are laid out in \cref{sec:defs}.
Following that, there are two main parts to the paper. The first part,
\cref{sec:sprawl}, defines the sprawl data structure (\cref{def:sprawl}), and
shows how it follows from a bare-bones conception of traversal and exploration
(\cref{thm:schemes}). A full search algorithm is presented
(\cref{alg:sprawltraversal}), along with several results on sprawl correctness
(\cref{sec:correctness}). Finally, a stepwise procedure is presented for
transforming an arbitrary comparison-based index structure into a sprawl with
the same behavior (\cref{sec:emulatingidx}).

The sprawl is a directed hypergraph at heart, and its hyperedges are each
labeled with one or more \emph{regions} of some sort. The second part of the
paper, \cref{sec:ambitregiontype}, defines the ambit region type, which may be
used for this purpose. Metaphorically, the ambit extends some ways from the
hyperedge sources, also known as the region \emph{foci}. Focal comparisons are
combined in some manner to become a measure of \emph{remoteness}, which is
capped at a \emph{radius}.

The manner in which comparisons are combined into remoteness is of central
importance; to be useful, this remoteness map must preserve some properties of
the comparisons. The idea is to take axioms and theories originally designed
for the comparison function, and to apply them to some other domain, which
contains the image of the remoteness map. This permits us to reuse existing
region definitions, and to introduce the choice of remoteness map as an
additional parameter, greatly increasing the flexibility of our regions:
Rather than restricting ourselves to the formal language of the original
domain, we broaden our horizon to include the meta-language of
structure-preserving functions.

Linear remoteness in particular is examined in some detail, including
conditions for overlap (\cref{sec:linearmetric}), ways of emulating the
regions of existing methods (\cref{sec:emulatingregions}), optimizing a
region's shape (\cref{sec:optcoeff}) and the choice of its foci
(\cref{sec:focusselection}). I also sketch out ways to work with non-linear
remoteness (\cref{sec:nonl}) as well as highly general families of comparisons
and their properties, along with possible remoteness maps
(\cref{sec:preserve}).

\subsection{Basic Definitions}
\label{sec:defs}

During the exploration of an index structure, we incrementally accumulate
information of some sort. Specifically, in a comparison-based index, this
comes in the form of comparisons between objects. The following definitions
clarify what this means.

\begin{definition}
    \label{def:workload}
    Given a \emph{universe $U$}, a \emph{workload} in $U$ is a family
    $\fml{Q_i}{i}$ of sets $Q\subseteq U$ known as \emph{queries}.
    For a given \emph{ground set} $V\subseteq U$, the workload is
    \emph{atomistic} if, for every query $Q$, the elements $v\in Q\cap V$
    form valid queries $\set{v}$.
\end{definition}

\noindent
This definition of a workload similar to that of \citet{Pestov:2006}, who
define it as consisting of the universe (the \emph{domain}), the ground set
(\emph{dataset} or \emph{instance}) and the queries.

\begin{definition}
    \label{def:compfunc}
    Given a \emph{universe $U$} and a set $K$ of \emph{comparison values}, a
    \emph{comparison function on $U$} is any binary function of the form
    \(
        \delta:U\times U\to K
    \).
\end{definition}
\begin{examples}
    \label{ex:cmpfuncs}
    \begin{paras}
    \item Partial orders may be represented by binary Boolean
        functions, but we might also have a more informative codomain
        $K\deq\set{0,1,2,3}$, with the values signifying, in
        order, that $u$ is equal to, less than, greater than or incomparable
        to $v$, yielding a form of \emph{asymmetric
        distance}~\citep{Mennucci:2013}.
    \item More generally, distances and similarity functions correspond to
        real-valued comparisons, usually non-negative, i.e., $\delta:U\times
        U\to [0,\infty]$. Distances ordinarily reserve zero for identical
        objects. If they also satisfy the triangle inequality, for example,
        they are called \emph{quasimetrics} (see \cref{sec:linearmetric}).
    \item
        \label{ex:cmpfuncs:generalized}
        Generalizing further, we may have an \emph{$L$-fuzzy relation}
        $\delta:U\times U\to L$, for some poset $L$, usually a complete
        lattice~\citep{Goguen:1967}. If, in addition, $L$ admits a binary
        operation, $\delta$ may be one of several forms of \emph{generalized
        metrics}~\citep{Schweizer:1983,Filip:2010,Deza:2013,Conant:2016}.
    \end{paras}
\end{examples}

\begin{adhoc}{On symmetry}
    \label{adhoc:skewsymmetry}
    In metric indexing, symmetry is assumed, and this means that one
    comparison suffices between any two objects. For quasimetrics and more
    general comparisons, we do not necessarily have symmetry, but we may have
    a kind of \emph{generalized skew symmetry}, where for any distinct objects
    $u$ and $v$, there exists a function $\delta(u,v)\mapsto\delta(v,u)$. This
    is true of linear and partial orders, for example. If are able to
    determine object equality by some separate mechanism, this again means
    that we need only compare in one direction.
\end{adhoc}

\begin{adhoc}{On totality}
    The comparison function need not be total; an undefined comparison may
    simply be treated as one that has not been performed during a given
    search, thereby blocking any decisions depending on it.
\end{adhoc}

\noindent
Resolving a query against an index structure involves incrementally exploring
some of its objects, or nodes, with the goal being to examine all relevant
ones and avoiding most of those that are irrelevant.\endnote{In some rare
cases, one may be able to deduce relevance without explicit examination, but
the performance gains this provides are generally quite modest.}
The behavior of such a search procedure could be characterized by listing the
explored objects. The index structure itself might permit several such
behaviors, with any non-determinism being resolved by some heuristic. This
repertoire of behaviors could then be described by a set of sequences, or
\emph{language}.

\begin{definition}
    For a finite, non-empty \emph{ground set} $V$, let $V^\hiast$ be the set
    of finite sequences $\seq{x_1,\dots,x_k}$ of elements $x_i\in V$,
    including the empty sequence $\emptyseq$.
    A \emph{language} is a pair \tup{V,\L} where $\L\subseteq V^\hiast$.
    For any $\alpha\deq\seq{x_1,\dots,x_k}$, let
    $\alpha x=\seq{x_1,\dots,x_k,x}$
    and
    $\tilde{\alpha}=\set{x_1,\dots,x_k}$.
    If $\alpha x\in\L$ then $x$ is a \emph{feasible continuation} of $\alpha$
    in \tup{V,\L}.
    Sequences without feasible continuations are \emph{maximal}.
\end{definition}

\noindent
Regardless of whether we are performing a comparison-based search or some
other, more general exploration, there are certain properties common to any
kind of such traversal---properties I intend to capture in the axioms of
\cref{def:traversal}.
Most fundamentally, nodes are visited one by one, and no node is explored more
than once, as captured by \cref{axiom:nonempty,axiom:simple,axiom:hereditary}.
\Cref{axiom:interval} states that nodes are discovered and eliminated at most
once, and that current availability depends on which nodes have been
traversed so far, regardless of order.

\begin{definition}
    \label{def:traversal}
    A \emph{traversal repertoire} is a language
    \tup{V,\L} satisfying the following \emph{traversal axioms} for all
    $\alpha,\tau,\omega\in V^\hiast$
    and
    $x\in V$, where
    $\tilde\alpha\subseteq\tilde\tau\subseteq\tilde\omega$:
    \begin{axioms}[series=traversal]{T}
    \item
        $\emptyseq\in\L$;
        \hfill (\emph{non-emptiness})%
        \label{axiom:nonempty}
    \item If $\tau x\in\L$ then $x\notin\tilde\tau$;
        \hfill (\emph{simplicity})%
        \label{axiom:simple}
    \item If $\tau x\in\L$ then $\tau\in\L$;
        \hfill (\emph{heredity})%
        \label{axiom:hereditary}
    \item
        If $\tau\in\L$ and
        $\alpha x, \omega x\in\L$
        then
        $\tau x\in\L$.
        \hfill (\emph{interval property})%
        \label{axiom:interval}
    \end{axioms}
    The elements of $V$ are called \emph{nodes} and the sequences in $\L$ are
    called \emph{traversals}.
\end{definition}

\fixrestatablespace
\let\oldlabel=\label%
\begin{restatable}{remark}{greedoidremark}
    \def\F{\mathcal{F}}
    \label{rem:greedoids}
    \let\label=\oldlabel
    \ifinapx
    A \emph{greedoid} is a non-empty, simple, hereditary language \tup{V,\L}
    that satisfies the following \emph{greedoid exchange property}, for all
    $\alpha,\beta\in V^\hiast$~\citep{Korte:1991}:
    \begin{axioms}{G}
    \item
        \label{axiom:greedex}
        If $\alpha,\beta\in\L$ and $|\alpha|>|\beta|$ then $\beta x\in\L$ for
        some $x\in\tilde{\alpha}$.
    \end{axioms}
    \emph{Interval greedoids}
    satisfy the interval property~\axiomref{axiom:interval}, and are thus
    exactly the traversal repertoires that satisfy the greedoid exchange
    property~\axiomref{axiom:greedex}.%
    \endnote{\Citet{Korte:1991} include non-emptiness in their definition of
    heredity~[p.\,5]. They also explicitly include every prefix in the
    heredity axiom, i.e., $\alpha\beta\in\L\Rightarrow\alpha\in\L$, which
    is equivalent to my \cref{axiom:hereditary}. Finally, they define the
    interval property only for families of sets~[p.\,48], but my axiom
    coincides with theirs when applied to greedoids. The proof is quite
    straightforward using Lemma~I.1.1 of
    \etal\citeauthor{Korte:1991}.
    }
    It is not hard to verify that if \cref{axiom:interval} is strengthened by
    removing the upper bound $\omega x$, the resulting traversal repertoires
    satisfy the following \emph{stronger} property:
    \begin{axioms}{U}
    \item
        \label{axiom:antiex}
        If $\alpha,\beta\in\L$ and
        $\tilde\alpha\not\subseteq\tilde\beta$
        then
        $\beta x\in\L$
        for some
        $x\in\tilde\alpha$.
    \end{axioms}
    These are precisely the so-called \emph{upper} interval
    greedoids, or \emph{antimatroids},\endnote{Some definitions of
    antimatroids also require them to be \emph{normal}, i.e., that every
    element in $V$ occur in some sequence in $\L$.} and they represent
    traversals without elimination, such as the graph traversal of
    \cref{ex:graphtrav}~\citep[cf.][\null III.2.11]{Korte:1991}.
    If the lower bound $\alpha x$ is removed instead, we get
    traversals without discovery, where reachable nodes are available
    initially, but may be eliminated. (These might \emph{not} be \emph{lower}
    interval greedoids, i.e., \emph{matroids}.)
    \else
    See \cref{apx:remarks} for notes on kinship to \emph{greedoids}.
    \fi
\end{restatable}

\begin{example}
    \label{ex:graphtrav}
    The possible sequences of nodes visited when traversing a graph
    \tup{V,E}
    from a given start node $s\in V$ form a traversal
    repertoire \tup{V,\L}.
    The contents of the traversal queue, or \emph{fringe}, after traversing
    nodes $\seq{x_1,\dots,x_k}\eqqcolon\alpha$ are exactly the feasible
    continuations of $\alpha$. These are the nodes that have been discovered
    but not yet traversed, and that therefore are \emph{open} or
    \emph{available}. Any choice of priority, such as \emph{depth-first} or
    \emph{breadth-first}, yields one specific traversal.
\end{example}

\noindent
A traversal repertoire describes possible traversals of the index \emph{for
one given query}. A more complete description of the index structure requires
a mapping from queries to traversal repertoires, as described in the following
definition. An important class of such behaviors is characterized by
\cref{axiom:monotone}, which says that narrowing the query should not lead to
additional nodes becoming available.

\begin{definition}
    \label{def:scheme}
    Given a workload $\fml{Q_i}{i\in I}$
    in universe $U$, a \emph{traversal scheme} with
    ground set $V\subseteq U$ is a family
    $\tightfml{\tup{V,\L_i}}{i\in I}$
    of traversal repertoires. It is \emph{monotonic} if it satisfies the
    following axiom for all $i,j\in I$:
    \begin{axioms}[resume*=traversal]{T}
    \item If $Q_i\subseteq Q_j$ then $\L_i\subseteq\L_j$.
        \hfill (\emph{monotonicity})%
        \label{axiom:monotone}
    \end{axioms}
    A
    scheme is \emph{correct} if $Q_i\cap V\subseteq\tilde\tau$ for all $i\in
    I$ and each maximal traversal $\tau\in\L_i$.
\end{definition}

\begin{example}
    A procedure that traverses a graph \tup{V,E} from a fixed starting node,
    searching for any nodes in $Q$, yields a traversal scheme. This includes
    some \emph{single-pair shortest path} algorithms and querying any kind of
    search tree.
\end{example}

\begin{definition}
    \label{def:cmpfeatures}
    Given a comparison function $\delta:U\times U\to K$ and a tuple
    $p=p_1,\dots,p_m$ of \emph{sources} in $U$, the comparisons between $u$
    and $p$ are the \emph{comparison-based features} of $u$. That is, $p$
    defines a \emph{feature map}
    \[
        \textstyle 
        \phi
        :
        u\to K^\hilc{m}\?\times K^\hilc{m}
        :
        u\mapsto\tup{x,y}
        \eqcomma
    \]
    where
    $x_i=\delta(p_i,u)$ and $y_i=\delta(u,p_i)$.
    We call $x$ and $y$ the \emph{forward} and \emph{backward feature vectors}
    of $u$, respectively.
    If $\delta$ is symmetric, we omit the redundant comparisons, and get
    $\phi:U\to K^\hilc{m}:u\mapsto x$, with feature space $K^\hilc{m}$.
\end{definition}

\begin{example}
    If \tup{U,\delta} is a metric space, then $x$ and $y$ are both equal to
    $\Psi(u)$, where $\Psi:U\to[0,\infty]^\hilc{m}$ is the pivot mapping of
    metric indexing, which produces a vector of the distances from $u$ to each
    of the pivots $p_1,\dots, p_m$~\citep{Zezula:2006}.
    In this case, we let $\phi=\Psi$.
\end{example}

\begin{definition}
    \label{def:cmpbasedreg}
    Given a universe $U$ and comparison function $\delta:U\times U\to K$, a
    \emph{comparison-based region} $C[p,S]$ is the preimage of $S$ along the
    comparison-based feature map $\phi:U\to K^\hilc{m}\?\times K^\hilc{m}$
    with sources $p\deq\tup{p_1,\dots,p_m}$, i.e.,
    \[
        C[p,S]\deq\phi^{-1}[S]=\set{u\in U: \phi(u)\in S}\eqcomma
    \]
    for some region $S$ in the feature space $K^\hilc{m}\?\times K^\hilc{m}$
    or (if $\delta$ is symmetric) $K^\hilc{m}$.
\end{definition}

\begin{examples}
    \begin{paras}
    \item
        For an ordering relation, intervals like $(p_1,p_2)$ and $[p_1,p_2)$
        may be specified using $S\deq\set{1001},\set{1001,1101}$.
        Intervals of partial orders work similarly. For example,
        axis-orthogonal hyperrectangles in $\mathds{R}^k$ are intervals
        between two opposing corners.
    \item
        If \tup{U,\delta} is a metric space, $m=1$, and we specify the $S$
        by an upper endpoint or \emph{radius}, we get a metric ball
        $B[p,r]\deq C[p,[0,r]]$.
        If we provide $S\deq[\ell,r]$ with a positive lower endpoint (i.e., an
        \emph{inner radius}) $\ell$, we get a shell region, and with multiple
        sources, an intersection of several such shells, corresponding
        to axis-orthogonal hyperrectangles in feature space.
    \item If \tup{U,\delta} is an asymmetric distance space, there is both a
        \emph{forward} ball $B^+[p,r]$ and a \emph{backward} ball $B^-[p,r]$
        for every center $p$ and radius $r$, where $\delta(p,u)$ and
        $\delta(u,p)$, respectively, falls below $r$ for any member
        $u$~\citep{Mennucci:2013}. These correspond to $S$ having upper limits
        \tup{r,\infty} and \tup{\infty,r}, respectively. \item For an
        $L$-fuzzy relation, the comparison-based regions are crisp level sets
        containing objects that are related to the sources to certain degrees.
    \end{paras}
\end{examples}

\begin{remark}
    In an actual data structure, we would not store the full extension of any
    region, of course, but rather the parameters of some intensional
    definition---a \emph{formal} region. A formal closed ball, for example, is
    simply a pair \tup{p,r} of a center and a radius, corresponding to the
    actual closed ball $B[p,r]$, while a formal closed interval consists
    of two endpoints $p_1$ and $p_2$. In fact, as long as we have the
    information needed to detect overlap with any valid query, we might not
    even need to be able to reconstruct the region (or the query) at all.
\end{remark}

\section{The Sprawl Data Structure}
\label{sec:sprawl}

Conceptually, a sprawl is a network of regions and points. The regions tell us
where we might find points of interest: If we determine that a region might
contain something interesting, it indicates which points to examine next.
Conversely, the points help us stake out the regions: In order to determine
whether a region contains something of interest, we must first examine its
source points, or \emph{foci}.

The traversal of a sprawl has an \emph{and--or} nature. We may traverse a
region once \emph{all} its sources are traversed, or we wouldn't have the
necessary information. This is not true of points, however; we may traverse
them as soon as any regions have pointed us in their direction. This is
exactly how directed hypergraphs are traversed, and so it makes sense to base
the definition of sprawls on them:\footnote{For more on hypergraphs, see
\cref{sec:hypergraphs}.}

\begin{definition}
    \label{def:sprawl}
    A \emph{sprawl} \tup{V,E,P,N} in a universe $U$ consists of:
    \begin{stmts}[widest=ii]
    \item A directed hypergraph \tup{V,E}, with ground set $V\subseteq U$; and
    \item Two edge labelings $P,N:E\to \smash{2^{2^U}}\!$.
    \end{stmts}
    Each label $P(e)$ or $N(e)$ is a family of sets $R\subseteq U$ called the
    \emph{positive} or \emph{negative regions of $e$}, respectively.
    A sprawl is \emph{finite} if each label is finite.
    If a sprawl has explicitly specified root nodes (as in \cref{def:roots}),
    any implicit root edge $e$ is assumed to have $P(e),N(e)=\emptyset$.
\end{definition}

\noindent
Note that the nodes of the sprawl are points from the same universe as where
the regions are located.

Sprawl traversal is discussed in more detail in
\cref{sec:signedtraversal,sec:sprawltrav}, and is defined by
\cref{alg:sprawltraversal}. That definition, however, is spread over several
constructions and subroutines, so for reader convenience, the steps of the
algorithm have been collected here, in
\cref{alg:sprawtraversalcollected}.\endnote{In other words,
    \cref{alg:sprawltraversal} and \cref{alg:sprawtraversalcollected} are
descriptions of the same algorithm.} (For some pointers on implementation, see
\cref{sec:hypergraphs}.) Recall that each edge $e\in E$ has a \emph{target}
$\tgt(e)\in V$ and a set of \emph{sources} $\src(e)\subseteq V$ (cf.\
\cref{def:hypergraph}).

\begin{algorithm}
    \label{alg:sprawtraversalcollected}
    For a query $Q$, the sprawl \seq{V,E,P,N} is \emph{traversed} as follows:

\medskip\noindent
{\def\pseudolineheight{1.15}%
\begin{pseudo}%
    initially, no nodes are \emph{available}, \emph{traversed} or
    \emph{eliminated} \nl
    initially, no edges are \emph{active} or have been \emph{used} \nl
    \kw{repeat} \nl
        \>\label{ln:activation}%
        \emph{unused} edges whose sources have all been \emph{traversed} become
            \emph{active} \nl
        \>\kw{for} each \emph{active} edge $e$ \nl
        \>\>\kw{if} the target of $e$ is neither \emph{traversed} nor
        \emph{eliminated} \nl
        \>\>\>\label{ln:elimcond}%
                \kw{if} $Q$ is disjoint from some region in $N(e)$ \nl
        \>\>\>\>the target of $e$ is \emph{eliminated} and no
                    longer \emph{available} \nl
        \>\>\>\label{ln:disccond}%
                \kw{else} \kw{if} $Q$ intersects all regions in $P(e)$ \nl
        \>\>\>\>the target of $e$ is \emph{discovered} and is now
                        \emph{available} \nl
        \>\>$e$ has now been \emph{used} and is no longer \emph{active} \nl
        \>\label{ln:available}%
        \kw{if} there are nodes \emph{available} \nl
            \>\>\label{ln:heuristic}%
            heuristically select and examine an \emph{available} node $v$ \nl
            \>\>$v$ has now been \emph{traversed} and is no longer
                \emph{available} \nl
    \kw{until} no nodes are \emph{available}
\end{pseudo}}
\end{algorithm}

\noindent
Every node that is examined and found to belong to $Q$ is included in the
search result. The only way for an edge $e$ to be activated in
line~\ref{ln:activation} during the first iteration is for $\src(e)$ to be
empty; its target $\tgt(e)$ then becomes a \emph{starting node} or
\emph{root}. As a shorthand, we may simply declare such roots, in which case
we assume an incoming sourceless, regionless edge for each. That is, for each
root $v$, we assume an edge $e$ with $\tgt(e)=v$ and
$\src(e),P(e),N(e)=\emptyset$, so that when we get to line~\ref{ln:available}
in the first iteration, $v$ is guaranteed to have been discovered (though not
necessarily available, as it may have been eliminated by some other sourceless
edge).

The heuristic choice in line~\ref{ln:heuristic} may be managed by some form of
priority queue. The priority of a node $v$ may be modified whenever we examine
an edge $e$ with $\tgt(e)=v$, i.e., whenever $v$ is disovered or fails to be
eliminated (i.e., whenever line~\ref{ln:disccond} is executed).

\begin{remark}
    \label{rem:unconditional}
    An edge without regions leads to unconditional discovery. To have an
    edge lead to unconditional \emph{elimination} once its sources have
    been traversed, we can let $N(e)=\set{\emptyset}$, which will necessarily
    satisfy the condition in line~\ref{ln:elimcond}.
    If you would like to have exactly one region for each edge (cf.\
    \cref{rem:singleregions}), you could use $P(e)=\set{U}$ for the positive
    case. Although this would \emph{not} lead to discovery if $Q=\emptyset$,
    the correct (empty) search result would still be guaranteed.%
\end{remark}

\begin{adhoc}{On $\boldsymbol{k}$NN}
    \label{rem:knn}
    I'm assuming that we have a fixed, formal description of the set $Q$,
    which we use to determine intersection with the various regions. For
    \emph{relative} kinds of queries, however, such as finding the $k$ nearest
    neighbors to some query point, our description of $Q$ might get more
    precise as the search progresses. For example, we may maintain the $k$
    nearest points found \emph{so far}, as well as a bounding radius for them,
    giving us a ball that is certain to \emph{contain} $Q$, but that isn't a
    precise description of $Q$ itself~\citep[see, e.g.,][]{Hjaltason:2003}.
    While I don't explore such updates further, they could be accommodated
    simply by modifying the description of $Q$ along with the priorities used
    in the heuristic, as discussed in the main text, above.
\end{adhoc}

\begin{examples}
    \label{ex:sprawls}
    Three examples are given in \cref{fig:sprawl}. We say that a sprawl edge
    $e$ is \emph{positive} if $N(e)$ is empty (it can only lead to discovery)
    and \emph{negative} if $P(e)$ contains the empty set (it can only lead to
    elimination).
    The black and red edges in the diagrams are positive and negative,
    respectively. Also, each edge has exactly one (nonempty)
    region.\endnote{That is, $|P(e)|=1$ for positive edges and $|N(e)|=1$ for
    negative edges} As alluded to in the figure, hyperdigraphs may be
    implemented using bipartite digraphs, and a natural optimization is then
    to permit each hyperedge to have multiple targets, as in
    \cref{subfig:23t}, and to combine edges with identical sources and
    regions, as has been done in \cref{subfig:pmt}.\endnote{Some definitions
    of directed hypergraphs do permit multiple targets; the sprawl definition
    could easily have been rewritten to accommodate that, or to use a
    bipartite digraph of nodes and signed regions.}
\end{examples}%
\begin{figure}
\centering%
\def\setup{\tikzmath {
    \r = 1.5;
    \topy = \r;
    \boty = \r * sin(-30);
    \midy = (\topy + \boty) / 2;
    \sep = cos(30)*\r/5;
    \sep = \sep * 1.1;
    }
    \def\pt(##1,##2){
        (##1*\sep,\pty) node[point] (p##2) {}
        ++(0,2.5 * \sgn ex) node[point name] {$v_{##2}$}
    }
    \def\sgn{+1}
    \def\pty{\topy}
}
\def\switchtobottom{
    \def\sgn{-1}
    \def\pty{\boty}
}
\subcaptionbox{\label{subfig:23t}2--3 Tree~\citep{Aho:1974}}{\begin{tikzpicture}%
[sprawl diagram]

\path[use as bounding box] (-1.6,-1.5) rectangle (1.6,2);

\path (0,-1.5);
\path (0,-1.5);
\setup
\draw

    \pt(-2,1)
    \pt(+2,2)

;

\switchtobottom

\draw

    \pt(-5,3)    \pt(-1,5)    \pt(+3,7)
    \pt(-3,4)    \pt(+1,6)    \pt(+5,8)

;

\draw[every edge/.append style={src edge}]

    (-4*\sep,\midy) node[region] (R1) {$e_1$}
    (+0*\sep,\midy) node[region] (R2) {$e_2$}
    (+4*\sep,\midy) node[region] (R3) {$e_3$}

    (p1) edge (R1) edge (R2)
    (p2) edge (R2) edge (R3)
;

\draw[every edge/.append style={tgt edge}]
    (R1) edge (p3) edge (p4)
    (R2) edge (p5) edge (p6)
    (R3) edge (p7) edge (p8)
;

\begin{pgfonlayer}{background}
    \def\r{\rootradius}
    \draw[pin]
    (p1.center) ++(0,\r) arc[radius=\r, start angle=90, end angle=270]
    --
    ($(p2.center) + (0,-\r)$) arc[radius=\r, start angle=-90, end angle=90]
    -- cycle
    (p2.center) ++(\r + 1.2pt,0) to[out=0,in=180] ++(.5,.5) node[right,legend] {root
        nodes
    }
    ;
\end{pgfonlayer}

\end{tikzpicture}}
\hfill
\subcaptionbox{AESA~\citep{Ruiz:1986}}{\begin{tikzpicture}[sprawl diagram]
\path[use as bounding box] (-1.6,-1.5) rectangle (1.6,2);

\path (0,-1.5);
\draw
    \foreach[count=\i] \a/\b in {90/90,210/-90,-30/-90} {
        (\a:1.5)   node[point] (p\i) {}
      ++(\a:2.5ex) node[point name] {$v_\i$}
    }
;
\begin{pgfonlayer}{background}

\draw[pin]
    \foreach[count=\i] \a in {0,120,240} {
        (p\i) ++(210+\a:\rootradius)
        arc[radius=\rootradius, start angle=210+\a, end angle=-25+\a]
    }
    ;

\draw[sred,->]
    \foreach[count=\i] \j/\p in {2/.4,3/.5,1/.6} {
        (p\i) edge[bend right=15, src tgt edge]
              node[pos=\p, region, text=sred] (R\i\j) {$e_{\i\j}$} (p\j)
        (p\j) edge[bend right=15, src tgt edge]
              node[pos=\p, region, text=sred] (R\j\i) {$e_{\j\i}$} (p\i)
    }
;
\end{pgfonlayer}
\end{tikzpicture}}
\hfill\subcaptionbox{\label{subfig:pmt}PM-Tree~\citep{Skopal:2004a}}{%
\begin{tikzpicture}[sprawl
        diagram]

\path[use as bounding box] (-1.6,-1.5) rectangle (1.6,2);

    \setup
    \path (0,-1.5);
    \tikzmath {
        \r = 1.5;
        \bsep = cos(30)*\r/2;
    }
    \draw
        (-5*\sep,\topy) node[point] (p1) {}
            +(0,+2.5ex) node[point name] {$v_1$}
        (+5*\sep,\topy) node[point] (p3) {}
            +(0.04,+2.5ex)
            node[point name] {$v_3$}
        ++(-\bsep,0) node[point] (p2) {}
            +(0,+2.5ex) node[point name] {$v_2$}

        (0,\boty) node[point] (p5) {}
            +(0,-2.5ex) node[point name] {$v_5$}
        ++(-\bsep,0) node[point] (p6) {}
            +(0,-2.5ex) node[point name] {$v_6$}
        ++(2*\bsep,0) node[point] (p4) {}
            +(0,-2.5ex) node[point name] {$v_4$}
    ;
    \def\pct{.5}
    \draw
        (p1) edge[src tgt edge] node[pos=\pct, region] (R1) {$e_1$} (p5)
        (R1) edge[tgt edge] (p4) edge[tgt edge] (p6)
        ;

\begin{pgfonlayer}{background}

\draw
    ($(p1)!.5!(p2)$) coordinate (A)
    ($(p2)!.2!(p4)$) coordinate (B)
    (A |- B) node[legend] (lazy) {lazy}
;

\def\r{\rootradius}
\draw[pin]
    (p1) circle[radius=\r]
    (p2) ++(220:\r)
         arc[radius=\r, start angle=220, end angle=-85]
    (p3) ++(210:\r)
         arc[radius=\r, start angle=210, end angle=-95]
    ;

\draw[pin] ($(p2)!.5!(p6)$) to (lazy);

\draw [every edge/.append style={src tgt edge}]
        (p2)
            edge[sred] (p5)
            edge[sred] (p4)
            edge[sred] (p6)
        (p3)
            edge[sred] (p5)
            edge[sred] (p4)
            edge[sred] (p6)
        ;
    \path ($(p2)!.5!(p3)$) coordinate (pp);
    \draw ($(pp)!\pct!(p5)$) node[region,text=sred] {$e_2\cdots e_7$};
\end{pgfonlayer}
\end{tikzpicture}}
\caption{In a sorted search tree (a), the nodes are linearly ordered, with
interval regions between them, and the topmost nodes available at the outset.
For metric pivot filtering~(b), all points are known initially, and they use
negative regions to eliminate each other. Some structures, such as the
PM-Tree~(c), use both kinds of regions, with the negative regions implemented
\emph{lazily}, used only if we discover their children}%
\label{fig:sprawl}%
\end{figure}
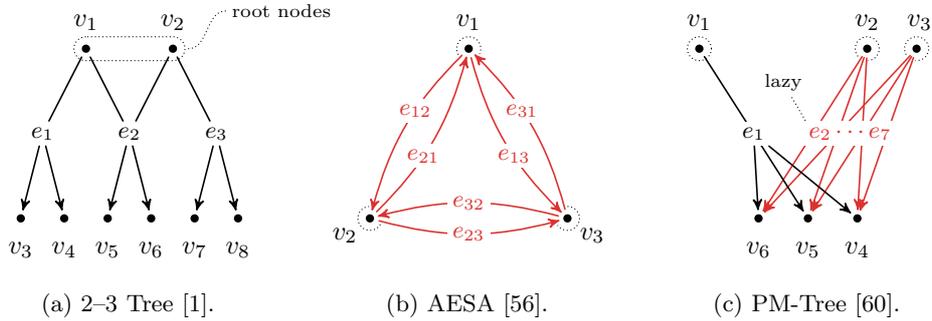%

\vspace{-\bigskipamount}

\begin{remark}
    \label{rem:lazy}
    The \emph{laziness} in \cref{fig:sprawl} is not strictly necessary; it is
    primarily an optimization intended to prevent premature work in
    eliminating points that are never discovered to begin with (cf.\
    \ref{rem:laziness}).
\end{remark}

\noindent
In order to emulate all traversal behaviors (cf.\ \cref{thm:schemes}), simple
hyperdigraph traversal is not sufficient; we also need to be able to
\emph{eliminate} points that are found to be irrelevant, as indicated
in~\cref{alg:sprawtraversalcollected}. So, before we take a more in-depth look
at sprawls proper, I'll make a slight detour and introduce a new form of
traversal that extends to \emph{signed} hyperdigraphs.

\subsection{Signed Directed Hypergraph Traversal}
\label{sec:signedtraversal}

The established traversal algorithm for directed hypergraphs
(\cref{alg:old:hypertraversal}) may be augmented with the idea that
\emph{negative edges inhibit traversal}. That is, while positive edges lead to
discovery, negative edges lead to the \emph{elimination} of their targets, as
in the added third step of the following algorithm.

\begin{algorithm}
    \label{alg:hypertraversal}
    A signed hyperdigraph $\tup{V,E,\sigma}$ is \emph{traversed} by
    \emph{discovering}, \emph{traversing} and \emph{eliminating} nodes, as
    described in the following.
    Nodes are \emph{available} if they have been discovered but not traversed
    or eliminated.
    Edges are \emph{active} once their sources have been traversed.
    The following steps are repeated, starting with the second step in the
    first iteration:
    \begin{steps}
        \item
            One of the available nodes is selected and traversed.
            \label{step:exploration}
        \item The targets of active positive edges are discovered.
            \label{step:discovery}
        \item The targets of active negative edges are eliminated.
            \label{step:elimination}
    \end{steps}
    The steps are repeated as long as there are nodes available.
    When there are several nodes available, the choice is made using a
    \emph{traversal heuristic}.
\end{algorithm}

\begin{definition}
    Any given traversal heuristic will make \cref{alg:hypertraversal} traverse
    one specific sequence of nodes, for any given signed hyperdigraph
    $G$.
    The \emph{traversal repertoire of $G$} is the set of traversals attainable
    by varying the heuristic.
\end{definition}

\noindent
Note that for the traversal of an unsigned hyperdigraph, any traversal order
will end up traversing the same set of nodes, i.e., all those that are
reachable. With elimination, however, this is no longer true---different
traversal heuristics may lead to completely different results.

\begin{restatable}{remark}{forwardchainingremark}
    \ifinapx
{\newcounter{marknha}\setcounter{marknha}{1}%
\newcounter{marknhb}\setcounter{marknhb}{2}%
\def\marknha{(\fnsymbol{marknha})}%
\def\marknhb{(\makebox[\widthof{\fnsymbol{marknha}}]{\fnsymbol{marknhb}})}%
\Cref{alg:hypertraversal} may be seen as an application of the generalized
forward chaining construction of \citet{Marek:1999} to the following
general logic program, built from the hypergraph $\tup{V,E}$:
\begin{IEEEeqnarray*}{lCl"s+s}
\tau_v & \leftarrow & \alpha_v,\neg\mkern1mu\omega_v & for all $v\in V$;
&\llap{\marknha}\\
\alpha_v & \leftarrow & \seq{\tau_u : u \in\src e} & for all
positive $e\in E$ where $v = \tgt e$;\\
\omega_v & \leftarrow & \seq{\tau_u : u \in\src e},\neg\mkern1mu\tau_v & for all
negative $e\in E$ where $v = \tgt e$. &\llap{\marknhb}
\end{IEEEeqnarray*}
The variables $\tau_v$, $\alpha_v$ and $\omega_v$ indicate that $v$ has
been traversed, discovered and eliminated, respectively.\endnote{More
precisely, $\omega_v$ means that $v$ has been eliminated \emph{in time
to prevent traversal}. Without this proviso, some eliminations would
be forbidden by the forward chaining construction, because they would
contradict previous traversals.}
The required well-ordering of the non-monotonic (i.e., non-Horn)
clauses~\citep[\S\,3.1]{Marek:1999} would place all elimination
rules~\marknhb\ first, and then order the traversal rules~\marknha\ by the
traversal heuristic.
In the absence of negative edges, we may simply use a definite logic
program of Horn clauses such as
$\tau_v\leftarrow\seq{\tau_u:u\in\src e}$, where $v=\tgt e$,
resulting in ordinary forward chaining \citep[\S\,6.4.2]{Crama:2011}.
This is equivalent to ordinary digraph or hyperdigraph traversal, under
the hyperpath definition of, e.g.,
\etal\citet{Ausiello:2001}.
} 
    \else
    See \cref{apx:remarks} for notes on \emph{forward chaining}.
    \fi
\end{restatable}

Algorithm~\ref{alg:hypertraversal} is a straightforward generalization of
directed graph traversal, where
\begin{inl}
\item edges may have multiple sources that must be traversed, and
\item negative edges trigger the elimination, rather than discovery, of their
    targets.
\end{inl}
As it turns out, this generalization is all that is needed to exactly capture
the concept of a traversal repertoire, as shown in the following.

\begin{construction}
    \label{constr:reptohyper}
    Given a traversal repertoire \tup{V,\L_1}, let \tup{V,E,\sigma} be a
    signed hyperdigraph where the set $E$ of edges is the smallest set
    satisfying the following axioms, for all $x\in V,\alpha,\omega\in\L$,
    where $\tilde\alpha\subsetneq\tilde\omega$, and $x\notin\tilde\omega$.
    \begin{axioms}{E}
        \item If $\alpha x\in\L$ then there is a positive edge
            \edge{\tilde\alpha}{x} in $E$.
            \label{axiom:posedge}
        \item If $\alpha x\in\L, \omega x\notin\L$ then there is
            a negative edge
            \edge{\tilde\omega}{x} in $E$.
            \label{axiom:negedge}%
    \end{axioms}
    Let \tup{V,\L_2} be the traversal repertoire of \tup{V,E,\sigma}.
\end{construction}

\fixrestatablespace
\begin{restatable}{theorem}{thmdihyper}
    \label{thm:dihyper}
    The following statements are equivalent:
    \begin{stmts}[widest=ii]
    \item \tup{V,\L} is a traversal repertoire.
        \label{stmt:trav}
    \item \tup{V,\L} is the traversal repertoire of some signed hyperdigraph
        \tup{V,E,\sigma}.
        \label{stmt:hyptrav}
    \end{stmts}
    In particular, any traversal repertoire is the traversal repertoire of the
    corresponding signed hyperdigraph produced by
    \cref{constr:reptohyper}.
\end{restatable}
\begin{proof}
    See \cref{apx:proofs}.
\end{proof}

\fixrestatablespace
\begin{restatable}{remark}{antimatroidremark}
    \ifinapx
    As mentioned in Rem.~\ref{rem:greedoids},
    if our repertoire features only discovery, it is in fact an
    antimatroid. The hyperdigraph resulting from \cref{constr:reptohyper} will then
    have only positive edges, and the source sets correspond exactly to the
    \emph{alternative precedence constraints} of these targets, as used in one
    of several equivalent definitions of
    antimatroids~\citep[Thm.\,1.4]{Korte:1991}.
    \else
    See \cref{apx:remarks} for notes on kinship to \emph{antimatroids}.
    \fi
\end{restatable}

\noindent
Traversal repertoires are only the first step, however; to capture index
structure behavior in general, we need to implement traversal \emph{schemes},
where each query yields its own traversal repertoire. That is where the sprawl
comes in.

\subsection{Sprawl Traversal}
\label{sec:sprawltrav}

Let's say we have an index structure whose behavior corresponds to a traversal
scheme $\tightfml{\tup{V,\L_i}}{i\in I}$, with workload $\fml{Q_i}{i\in I}$.
For any query $Q_i$, then, we have a repertoire $\tup{V,\L_i}$, which we know
we can implement with a signed dihypergraph. In order to implement the entire
scheme, we would need to map each query to one of several signed
dihypergraphs, all of which have the same ground set $V$. We can consolidate
all these into a single hypergraph, and simply add predicates to the edges
indicating their presence and sign.\endnote{Edges with the same sources and
target may be identified, so the resulting dihypergraph will have a finite
number of edges, even if the workload is infinite.}

At this point we have what we could call a \emph{generalized sprawl}.
Similarly to the generalized search trees of \citet{Hellerstein:1995}, such a
generalized sprawl could have arbitrary predicates attached to its edges.
However, as we shall see (cf.\ \cref{thm:dihyper}), the assumption of
monotonicity~\axiomref{axiom:monotone} lets us restrict our attention to
\emph{regions} and \emph{region intersection}, yielding \cref{def:sprawl}.
As alluded to initially, the regions are used to determine the presence and
sign of each edge in the presence of a query, reducing the search to a signed
dihypergraph traversal.
In the following, I describe this reduction in more detail, and show that it
is indeed sufficient for implementing any monotone traversal scheme.

\begin{construction}
    \label{constr:sprawltohyper}
    Given a sprawl \tup{V,E,P,N} in $U$ and a \emph{query} $Q\subseteq U$, we
    produce a signed hyperdigraph \tup{V,E'\!,\sigma} as follows. Let $E'$ be
    the set of edges $e\in E$ for which either
    \begin{stmts}[widest=ii]
    \item $Q$ intersects every region in $P(e)$ and $N(e)$, in which case
        $\sigma(e)=+1$; or
    \item $Q$ does not intersect every region in $N(e)$, in which case
        $\sigma(e)=-1$.
    \end{stmts}
\end{construction}

\begin{algorithm}
    \label{alg:sprawltraversal}
    For a given query $Q\subseteq U$, a sprawl \tup{V,E,P,N} in $U$
    is traversed by traversing the corresponding signed hyperdigraph
    \tup{V,E'\!,\sigma} of \cref{constr:sprawltohyper}.
\end{algorithm}

\noindent
For a more explicit listing of the steps involved, see
\cref{alg:sprawtraversalcollected}, which describes the same algorithm.

\begin{definition}
    \label{def:sprawlscheme}
    Let $\tup{V,E,P,N}$ be a sprawl and
    $\fml{Q_i}{i\in I}$
    a workload, both in universe $U$.
    The corresponding \emph{traversal scheme of the sprawl} is the collection
    $\tightfml{\tup{V,\L_i}}{i\in I}$,
    where each language \tup{V,\L_i} is the
    repertoire of traversals attainable by \cref{alg:sprawltraversal} for the
    query $Q_i$, i.e., the repertoire of the signed hyperdigraph of
    \cref{constr:sprawltohyper}.
    The sprawl is \emph{correct} for this workload if its scheme is
    correct.
\end{definition}

\fixrestatablespace
\begin{restatable}{theorem}{thmschemes}
    \label{thm:schemes}
    For a workload $\fml{Q_i}{i\in I}$ in $U$, these statements are
    equivalent:
    \begin{stmts}[widest=ii]
    \item
        $\tightfml{\tup{V,\L_i}}{i\in I}$
        is a monotone traversal scheme with ground set $V\subseteq U$.
    \item
        $\tightfml{\tup{V,\L_i}}{i\in I}$
        is the traversal scheme of some sprawl \tup{V,E,P,N} in $U$.
    \end{stmts}
    In particular, any monotone traversal scheme is the traversal scheme of
    the corresponding sprawl produced by \cref{constr:schemetosprawl}, for any
    appropriate universe.
\end{restatable}
\begin{proof}
    See \cref{apx:proofs}.
\end{proof}

\noindent
In other words, \emph{the sprawl exactly maps out the design space of monotone
traversal schemes}. This, however, also includes \emph{incorrect} traversal
schemes; to be of any use, we need to ensure that our sprawl is
\emph{correct}.

\subsection{Ensuring Correctness}
\label{sec:correctness}

Correctness is defined in \crefnosort{def:scheme,def:sprawlscheme}.
Intuitively, we seek to determine whether we, for any query $Q_i$, are
\emph{guaranteed} to traverse any nodes in $V$ that are specified by $Q_i$,
regardless of our choice of priority heuristic. As it turns out, making this
determination is quite hard in the general case.

\begin{problem}
The \textsc{sprawl correctness} problem is defined as follows.

\medskip\noindent
\textbf{Instance:}
A sprawl \tup{V,E,P,N} and a workload $\fml{Q_i}{i}$, both in universe $U$.

\medskip\noindent
\textbf{Question:}
    Is the sprawl correct for this workload?
\end{problem}

\begin{theorem}\label{thm:conpc}
\textsc{sprawl correctness} is strongly coNP-complete.
\end{theorem}
\begin{proof}
    Sprawl \emph{in}correctness is polytime-verifiable with an incorrect
    maximal traversal as certificate, so the problem is in coNP.
    Given that it is not a number problem~\citep{Garey:1979},
    we can show strong coNP-hardness by polytime reduction from
    \mbox{\textsc{dnf tautology}}.
    Given a DNF formula $\varphi$, we first construct a root node representing
    a positive and negative literal for each variable in $\varphi$, and add
    unconditionally negative edges (i.e., with $N(e)=\set{\emptyset}$; cf.\
    \cref{rem:unconditional}) between them. These nodes are available at the
    outset, and as long as any additional edges in the sprawl are positive
    (i.e., with $N(e)=\emptyset$), the subset of these nodes traversed
    corresponds to a valid truth assignment for the variables; that is, for
    each variable $x$, we traverse exactly one of the nodes representing $x$
    and $\neg x$.\endnote{For notational convenience, I use $\varphi$ and the
    literals as names also of the nodes in the sprawl.}

    For each clause $x\land y\land \cdots\land z$ in $\varphi$, we add an edge
    \edge{\{x,y,\dots, z\}}{\varphi} with no negative regions and a single
    positive region \set{\varphi}, as shown in \cref{fig:dnftautreduction}.
    Each of these edges will trigger the discovery of $\varphi$ just in case
    its sources are traversed and $\varphi$ is relevant, and thus $\varphi$
    will be traversed exactly when at least one of the clauses of $\varphi$ is
    true for the given truth assignment given by the roots.\endnote{Note that
    it does not matter to the truth value if some of the roots are traversed
    \emph{after} $\varphi$.}
    Finally, we let the only query be $\set{\varphi}$, so that we are not
    required to traverse any specific roots.
    The sprawl will be correct for this query if and only if for every
    possible subset of literal nodes (i.e., without contradictory literals)
    that is traversed, we also traverse $\varphi$. This holds exactly when
    $\varphi$ is true for every possible truth assignment, i.e., when it is a
    tautology.
\end{proof}%
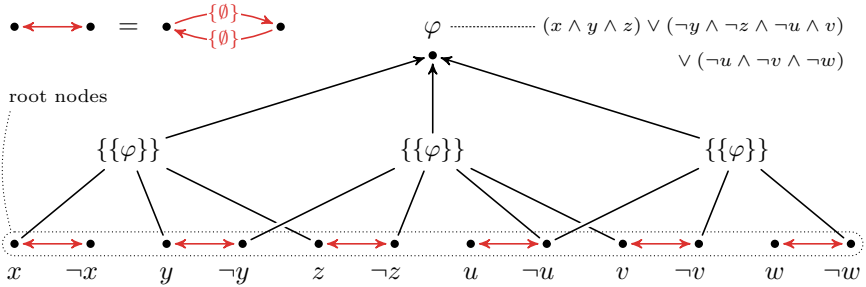
\begin{figure}
\centering
\tikzset{
    clause/.style = {
        font=\small,
        fill=white,
        inner sep=1.5pt,
    }
}
\begin{tikzpicture}[sprawl diagram]
    \def\r{\rootradius}

\draw

    \foreach[count=\i] \name in {x,y,z,u,v,w} {
        (2*\i - 7.5,0) node[point] (\name) {}
        ++(0,-2.5ex) node[point name] {$\name$}
        (2*\i - 6.5,0) node[point] (not-\name) {}
        ++(0,-2.5ex) node[point name] {$\mathllap{\neg}\name$}

        (\name) edge[sred, double tgt edge]
            (not-\name)
    }

    (0,2.5)
    node[point] (truth) {}
    ++(0,2.5ex) node[point name] (truthlabel) {$\varphi$}

    \foreach \i in {1,...,3} {
        (4*\i - 8, 1.25) node[clause] (clause-\i) {$\set{\set{\varphi}}$}
        (clause-\i) edge[tgt edge] (truth)
    }

;

\draw[legend]
    (truthlabel)
        ++(5.5,0)
        node[anchor=east]
        (formula) {
        $(x\land y\land z)
        \lor
        (\neg y\land \neg z\land \neg u\land v)$}
        +(0,-3ex) node[anchor=east]
        {$\mbox{}\lor (\neg u\land \neg v\land \neg w)$}
    ;

\draw (truthlabel) edge[pin] (formula);

\draw
    (truthlabel)
    ++
    (-5.5,0)
        node[point] (A) {}
    +(1,0) node[point] (B) {}
    +(1.5,0) node[point name] {$=$}
    ++(2,0) node[point] (C) {}
    +(1.5,0) node[point] (D) {}

    (A) edge[sred, double tgt edge] (B)

    ;

\begin{pgfonlayer}{background}
\draw
    \foreach \src/\tgt/\name in {D/C/\alpha,C/D/\beta} {
    (\src)
        edge[bend left=22,sred,src tgt edge]
        node[midway, fill=white,
                inner sep=0.2pt]
                {\scriptsize $\set{\emptyset}$}
        (\tgt)
    }
    ;
\end{pgfonlayer}

\draw
    (A)
    ++(-1ex,-6ex)
    node[legend, anchor=west]
    (gamma) {root nodes}
    ;

\begin{pgfonlayer}{background}

\begin{pgfinterruptboundingbox}
\tikzmath{
    \bend = 15;
    \anglea = 90 + \bend;
    \angleb = 270 - \bend;
}
\path
    (x.center) coordinate (a)
    (A.center) ++(0,-6ex) coordinate (b)
    ;

    \draw[pin, shorten <=\r + 1pt, shorten >=\r + 1pt]
        (a) to[out=\anglea,in=\angleb] (b);
\end{pgfinterruptboundingbox}
\end{pgfonlayer}

\begin{pgfonlayer}{foreground}
\draw[pin] (x) +(0,\r)
    arc[start angle=90, end angle=270, radius=\r]
    --
    ++(11,0)
    arc[start angle=-90, end angle=90, radius=\r]
    -- cycle
    ;
\end{pgfonlayer}

\draw[every edge/.append style={reverse src edge, over}]
    (clause-1)
        edge (x)
        edge (y)
        edge (z)
    (clause-2)
        edge (not-y)
        edge (not-z)
        edge (not-u)
        edge (v)
    (clause-3)
        edge (not-u)
        edge (not-v)
        edge (not-w)
    ;

\end{tikzpicture}
\caption{%
A DNF formula represented by a sprawl, with $\set{\varphi}$ as the only
possible query. The traversed bottom-row nodes form an arbitrary truth
assignment, and we are guaranteed to traverse the node representing $\varphi$
iff the formula is a tautology}\label{fig:dnftautreduction}%
\end{figure}%

\vspace{-1.5ex}
\begin{remark}
    Note that the hardness proof requires only purely positive and negative
    edges, and uses at most a single (nonempty) region per edge. The result
    also holds if, say, we must permit queries containing arbitrary subsets of
    $V$, because we may use negative regions containing the targets (rather
    than empty ones) to protect relevant roots from elimination. (Otherwise,
    incorrectness would be trivial to detect.) In this case, the original
    proof still holds for the query where $\varphi$ is the only relevant node,
    and for any other queries, we only risk \emph{additional} roots being
    traversed, and so $\varphi$ is still safe. Thus the sprawl is still
    correct just in case the formula is a tautology.
\end{remark}

\noindent
Despite the previous hardness result, we are, of course, able to determine the
correctness of existing index structures with ease. This is in part because of
their limited structure (many have only positive regions, and are simply
trees, for example), but perhaps even more importantly, they are restricted by
a natural rule: \emph{\label{rule:mappingout}The regions map out parts of the
index structure itself.} For example, a region associated with a subtree
contains all points in that subtree.

This simple interpretation of what a region should be can be generalized to
what I call \emph{responsibility}, where each edge $e$ is \emph{responsible}
for a set $\res(e)$ of nodes, and it is up to the regions of $e$ to fulfill
that responsibility. In the simplest cases, this might merely mean that the
regions must contain the nodes, as for search trees. More generally, however,
we only require that some hyperpath $\Pi_v$ is necessarily traversed for a
query $Q$ whenever $v\in Q$. More precisely (though perhaps a bit
opaquely---see also \cref{obs:localres}):

\begin{definition}
    \label{def:responsible}
    For a given workload, a sprawl \tup{V,E,P,N} is \emph{responsible} if, for
    every node $v\in V$ and every query $Q$ with $v\in Q$, there is a
    hyperpath $\Pi_v$ in \tup{V,E} from $\emptyset$ to $v$, with node set
    $V'\subseteq V$, such that for all edges $e,e'\in E$ and all regions $R\in
    P(e), R'\in N(e')$, the following axioms hold.
    \begin{axioms}{R}
    \item If $e\in\tilde\Pi_v$ then $Q$ and $R$ intersect;
        \hfill (\emph{discovery})%
        \label{axiom:posintersect}
    \item If $\tgt(e')\in V'$ then $Q$ and $R'$ intersect.
        \hfill (\emph{non-elimination})%
        \label{axiom:negintersect}
    \end{axioms}
    The smallest binary relation $\res\subseteq E\times V$ where
    $\tilde\Pi_v\subseteq \res^{-1}(v)$ for every path $\Pi_v$ as described
    above, is called a \emph{responsibility assignment}, under which any edge
    $e\in E$ is \emph{responsible for} each node $v\in\res(e)$.
    As a shorthand, we also assign responsibilities to any \emph{node} $v$, as
    follows:
    \begin{equation}
        \label{eq:noderes}
        \res(v) = \set{v}
        \enskip
        \cup
        \enskip
        \bigcup_{\mathclap{e\mkern1mu:\mkern1mu v\mkern.55mu\in\mkern.7mu\src(e)}}
            \enskip\res(e)
    \end{equation}
    That is, nodes are responsible for themselves, and share the
    responsibilities of outgoing edges. Thus the nodes responsible for $v$ are
    exactly those in some $\Pi_v$.
\end{definition}
\begin{remark}
    An irresponsible sprawl has no responsibility assignments.
\end{remark}
\begin{remark}
    \label{rem:larger}
    Replacing a region with a superset preserves responsibility.
\end{remark}

\noindent
Informally, existing index structures are easily verified to be responsible.
This follows naturally from the interpretation of regions as a partitioning or
coarsening of the data set \citep[cf.][]{Hellerstein:1995,Chavez:2001}, and
thus containing the points for which they are responsible. This might not be
the only reason, however: Most current index structures are based purely on
discovery, and for \emph{irresponsible} sprawls, correctness relies crucially
on \emph{elimination}:

\begin{observation}
    \label{obs:correctresp}
For any given sprawl and workload, the following holds:
\begin{stmts}[widest=ii]
\item If the sprawl is responsible, then it is correct;
\item If the sprawl is correct, then \cref{axiom:posintersect} holds.
\end{stmts}
Thus, a correct sprawl with no negative regions is responsible.
\end{observation}
\begin{proof}
    If the sprawl is responsible, then for every $v\in Q$, it is easily shown
    by induction over $\Pi_v$ that the targets of every edge
    $e\in\tilde\Pi_v$, including $v$ itself, are discovered but not
    eliminated, and so the sprawl is correct.
    Conversely, in a correct sprawl, any node $v\in Q$ is reachable via
    positive edges in the hypergraph from \cref{constr:sprawltohyper}, and so
    there must be a path $\Pi_v$ satisfying \cref{axiom:posintersect}.
\end{proof}
Intuitively, the reason \cref{axiom:negintersect} does not follow directly
from correctness is that incorrect elimination may be prevented even when
\ref{axiom:negintersect} does \emph{not} hold, by blocking the traversal of
one of the sources of an edge with a negative region.
This must be achieved either by eliminating the source or by not discovering
it in the first place, resulting in such Goldbergian contraptions as those
underlying \cref{thm:conpc}.
If there is no way to make such a guarantee (as is the case in, e.g.,
\cref{fig:sprawl}), the sprawl must necessarily be responsible in order to be
correct.

Beyond being responsible, all existing index structures are, to my knowledge,
\emph{acyclic} in the following sense.

\begin{definition}
    A hyperdigraph is \emph{acyclic} if its nodes may be strictly
    ordered so that $u<v$ whenever
    $u\in S$ for some edge \edge{S}{v}~\citep[p.\,195]{Gallo:1993}.
    A sprawl is acyclic if its edges with positive regions form an acyclic
    hyperdigraph.
\end{definition}

This definition is motivated by the fact that \emph{discovering} a node
depends on having traversed certain nodes, while in a responsible sprawl,
\emph{not eliminating} a node does not, so purely negative edges cannot
introduce cyclic dependencies.\endnote{For irresponsible sprawls, this is no
longer the case, as traversing a node may trigger an elimination that in
turn is instrumental in preventing \emph{another} elimination.}

\begin{remark}
    Not only do existing index structures seem to all be acyclic; the vast
    majority are, in fact, tree-structured, disregarding negative edges.
    The em-VP-forest and the D-Index~\citep{Yianilos:1999,Dohnal:2004} are two
    exceptions, both structured as \emph{forests}, the nodes of one tree
    having an extra edge to the root of the next.
\end{remark}

\noindent
We are now poised to set up a simplified, \emph{local} characterization of
responsibility.

\begin{restatable}{observation}{obslocalres}
    \label{obs:localres}
    Given a sprawl \tup{V,E,P,N} and an atomistic workload,
    let $v$ be any node and let $\set{e_1,\dots,e_k}$ be the edges with
    $\tgt(e_i)=v$.
    If $\res$ is a responsibility assignment, then for any $e\deq e_i$ and
    $R\in P(e), R'\mkern-1mu\in N(e)$:
    \begin{axioms}{L}
    \item $\res(e)\subseteq R$;
        \label{it:subr}
    \item $\res(\spcc{e}{v})\subseteq R'$;
        \label{it:subrprime}
    \item $\res(\spcc{e}{v})\subseteq \res(e_1)\cup\cdots\cup\res(e_k)$.
        \label{it:subrunion}
    \end{axioms}
    If the sprawl is \emph{acyclic} and every node is the target of \emph{at
    least one edge}, then \cref{it:subr,it:subrprime,it:subrunion},
    using the shorthand~\eqref{eq:noderes}, are necessary \emph{and
    sufficient} conditions for any relation $\res\subseteq E\times V$ to be a
    responsibility assignment.
\end{restatable}
\begin{proof}
    See \cref{apx:proofs}.
\end{proof}
\noindent
In other words, positive and negative regions contain edge and target
responsibilities, respectively, and edges with a common target jointly
share its responsibilities.

\begin{adhoc}{On false positives}
    \label{rem:falsepos}
    Intuitively, $R\setminus\res(e)$ is a waste of space: A query $Q$
    intersecting $R$ without intersecting $\res(e)$ leads to the superfluous
    discovery of $v\deq\tgt(e)$, even though the traversal would be correct
    without it. We would expect that shrinking $R$ or $Q$ to avoid such a
    false positive would lead to improved performance, in general, but this
    might not be the case! Although discovering $v$ is not necessary for
    correctness, it could lead us to negative regions that eliminate much of
    the data set. Such behavior is akin to that of correct, irresponsible
    sprawls, and although I do not deal with the issue formally, I shall
    generally assume that \emph{reducing the number of false positives is
    desirable}. That is, the reason to discover and not eliminate a node $v$
    is that we are looking for some of the nodes in $\res(v)$. If there are no
    negative regions, this will always be the case. Otherwise, the implication
    is, informally, that the sources of negative regions are readily available
    when needed (and not contingent on some false positive).
\end{adhoc}

\begin{remark}
    \label{rem:singleregions}
    Assuming we are trying to limit false positives (cf.\@
    \ref{rem:falsepos}), for responsible sprawls with atomistic workloads,
    edge labels might as well be single regions. From \cref{obs:localres}, we
    know that $\res(e)\subseteq\cap P(e)$, so any query intersecting each
    $R\in P(e)$ but not $\cap P(e)$ will be disjoint from $\res(e)$, leading
    only to superfluous traversal. Similarly, intersecting every
    $R'\mkern-1mu\in N(e)$ but not $\cap N(e)$ will needlessly protect
    $\tgt(e)$ from elimination.
\end{remark}

\subsection{Emulating Existing Indexes}
\label{sec:emulatingidx}

In many cases, the correspondence between an index structure and a sprawl may
be obvious; otherwise, one can start out by describing the index behavior as a
traversal scheme, and then construct a sprawl from that. A formal construction
is given in \cref{constr:schemetosprawl}; what follows is a less formal (and
less redundant) version.

\begin{steps}
\item For every $t\in V$ and any minimal $S\subseteq V$ whose traversal
    leads to the discovery or elimination of $t$ for some query, add an
    edge \edge[e]{S}{t}.\label{st:emul1}
\item For every edge \edge[e]{S}{t} and any maximal query $Q$ for which
    traversing $S$ does \emph{not} lead to to the discovery of $t$, add
    $U\setminus Q$ to $P(e)$.\label{st:emul2}
\item For every edge \edge[e]{S}{t} and any maximal query $Q$ for which
    traversing $S$ leads to the elimination of $t$, add $U\setminus Q$ to
    $N(e)$.\label{st:emul3}
\end{steps}

\noindent
In the first step, $S$ need not be minimal---the formal construction has no
such requirement. However, any non-minimal sets are entirely
redundant.\footnote{For an explanation of why the steps work, see the proof of
\cref{thm:schemes} in \cref{apx:proofs}.}

Following these steps will generally produce an infinite number of regions.
For the ubiquitous case where the index corresponds to a responsible sprawl
and the workload is atomistic, one might instead maintain their intersection
(cf.~\cref{rem:singleregions}). A more practical approach would be to simplify
the last two steps, as follows, considering only individual nodes as queries,
and using a single region per label:

\begin{steps}[start=2, label=\textit{Step \arabic*\rlap{'}.}]
\item For every edge \edge[e]{S}{t} and any query $\set{v}\subseteq V$ for
    which traversing $S$ leads to the discovery of $t$, add $v$ to the single
    region in $P(e)$.
\item For every edge \edge[e]{S}{t} and any query $\set{v}\subseteq V$ for
    which traversing $S$ does \emph{not} lead to the elimination of $t$, add
    $v$ to the single region in $N(e)$.
\end{steps}

\noindent
Here, the single regions are, most likely, members of some parameterized
family of regions (as discussed in \cref{sec:ambitregiontype}), and adding a
point $v$ means adjusting the region to accommodate $v$ (e.g., increasing a
radius). While this modified procedure \emph{will} preserve correctness, it is
not guaranteed to emulate the original behavior perfectly. It may, however, be
a useful approach in practice, subject to subsequent verification.

\begin{examples}
    \label{ex:emulating}
    Consider how one might arrive at the first two sprawls of
    \cref{ex:sprawls}, starting with the behaviors of the emulated index
    structures.
\begin{paras}
\item In the 2--3 Tree~\citep{Aho:1974}, there is no elimination in
    step~\ref{st:emul1}, and the
    minimal sets that would lead to the discovery of a node
    (step~\ref{st:emul1}) consist of the the two endpoints of the interval
    containing a subtree. This interval becomes the sole (positive) region of
    the edge (step~\ref{st:emul2}). Step~\ref{st:emul3} becomes irrelevant.
\item In AESA~\citep{Ruiz:1986}, there is no discovery in step~\ref{st:emul1}
    (also making step~\ref{st:emul2} irrelevant)---all nodes are roots. Each
    node may be eliminated by any other (alone), leading to a complete,
    directed graph. The regions become metric shells (cf.\
    \cref{sec:emulatingregions}) around the source, containing the target.
\end{paras}
In both cases, we will also have a root edge $e$ targeting each root, with no
sources or regions. These root edges fall out of the procedure, as the minimal
set $S$ whose traversal would lead to the discovery of a root is empty
(step~\ref{st:emul1}), and there are no maximal queries $Q$ that prevent
discovery or cause elimination initially, so steps~\ref{st:emul2}
and~\ref{st:emul3} will not add any regions to these root edges.
\end{examples}

\begin{adhoc}{On laziness}
    \label{rem:laziness}
    Some structures, such as the PM-Tree~\citep{Skopal:2004a}, have points
    that are the sources of a large number of negative edges, whose targets
    \emph{may never be discovered}. In such cases, one may wish to perform
    elimination \emph{lazily}, i.e., at some point \emph{after} the point has
    been discovered. The traversal scheme does not differentiate between these
    approaches, so to emulate such behavior, one would have to introduce extra
    points. However, this would not be true laziness, as we would still have
    to maintain the state of every out-edge when examining a point (see
    p.~\pageref{p:impl}). In an implementation, we would probably wish to have
    a separate kind of lazy or \emph{inverted} edge, which is referenced by
    its target and references its sources, rather than in the other direction.
    Then at the very last moment before traversing the target, we can follow
    the lazy edge in reverse to determine if its sources have been traversed,
    and if its negative regions intersect the query.
\end{adhoc}

\begin{adhoc}{On sublinear discovery}
    \label{rem:sublinear}
    Another optimization is to process multiple positive edges as at the same
    time, determining which of them to activate in sublinear time. For
    example, in B-Trees~\citep{Bayer:1972}, children are generally ordered, so
    one may use bisection to determine which of them are discovered. In
    theory, this bisection may itself be instantiated as further edges, in the
    form of a binary search tree, but in practice simply using binary search
    on an array of edges is much more efficient.
    Spaghettis~\citep{Chavez:1999} use this idea on the comparison values
    (i.e., distances) themselves.
    Going further, the D-Index~\citep{Dohnal:2004} selects among an
    \emph{exponential} number of children using what amounts to a simple
    hashing scheme, where the edges are looked up directly by patterns of
    overlap. Such optimizations are not directly present in the sprawl, of
    course, and they do not affect the traversal behavior, only the efficiency
    of its implementation.
\end{adhoc}

\section{The Ambit Region Type}
\label{sec:ambitregiontype}

Until now, we have assumed very little about the queries $Q$ and regions $R$,
at least explicitly. Implicitly, however, it has been clear all along that
these sets are comparison-based, in the sense of \cref{def:cmpbasedreg}, and
that we must have some descriptions that let us reason about them. The
assumption underlying the ambit construction is that \emph{such descriptions
are scarce}, so we wish to leverage each description to define \emph{multiple
regions}.

Having a single description $\psi$ refer to more than one region is a matter
of \emph{reinterpretation}, i.e., interpreting $\psi$ in a different
mathematical structure $\B$. In isolation, this tactic gives us unlimited
variability, but we cannot simply create regions that are completely detached
from our initial domain of discourse~\A, as they could then no longer be used
as part of our search procedure. The solution is to employ some
\emph{structure-preserving map} taking any given objects or parameters from
$\A$ to $\B$. In general, the only way to ensure that we preserve correctness
is to use an \emph{isomorphism}, but this would buy us nothing, as we would
end up with the exact same range of regions. Instead, we must preserve
\emph{some} properties, while forsaking others. For arbitrary sprawls, this is
not really feasible; however, if we assume our sprawls are \emph{responsible}
(\cref{def:responsible}), we can ensure that \emph{overlap} is preserved,
while \emph{non-overlap} may not be, and search would still be
correct~(cf.~\cref{rem:larger}).

This strategy is discussed in broader terms in \cref{sec:preserve}. For now,
we can make three simplifying assumptions. The first is that our map takes the
form
\begin{equation}
    [\loglike{id},f]\circ \Delta_m:\A\to \B
\end{equation}
for some function $f:K^\hilc{m}\?\to L$, were $\loglike{id}$ is the identity
function $x\mapsto x$ and $\Delta_m$ is the diagonal embedding
$x\mapsto\tup{x,\dots,x}$. The notation $[\loglike{id},f]$ indicates a
coproduct, where $\loglike{id}$ is applied to object tuples and $f$ is applied
to tuples of comparison values, i.e., feature vectors. This kind of mapping
lets us multiply the number of foci our description can accommodate. For
example, if $\psi$ is a unifocal description, this construction gives us an
$m$-focal one.

Though any number of descriptions could be pursued, the second simplifying
assumption is a defining one for the ambit in particular, and the source of
its name, namely that it is defined by its \emph{extent}, in some sense, as
determined by the comparisons (e.g., distances) to its foci.\footnote{Recall
that we use the edge sources as the foci of the edge regions.} A motivation
for this is making the common practice of incremental construction feasible:
Inserting an additional object simply requires increasing this extent, or
\emph{radius}, to cover the object. An ambit, then, is \emph{the preimage of a
ball}, along some structure-preserving map.\footnote{For a visual
presentation of an ambit in $\A$ corresponding to a ball in $\B$, see
\cref{fig:ambitexample}, whose details are explained in
\cref{ex:metricpreservation}.}

\begin{definition}
\label{def:ambit}
An \emph{ambit} of \emph{degree} $m$ in a universe $U$ with symmetric
comparison function $\delta:U\times U\to K$ is a comparison-based region
\begin{equation}
    B[p,r;\?f] \deq C\bigl[p,\set{x:f(x)\leq r}\bigr]\eqcomma
\end{equation}
as defined by
\begin{stmts}[widest=iii]
\item a tuple $p$
    of sources or \emph{foci}
    $p_1,\dots, p_m \in U$;
\item a \emph{radius} $r\in L$; and
\item a \emph{remoteness map}
    $f:K^\hilc{m}\?\to L$, with $L$ partially ordered.
\end{stmts}
The features $x_i$ of $u$ are called its \emph{radients}, and $f(x)$ is its
\emph{remoteness}.\endnote{This meaning of the term ``radient'' is taken from
\citet{Maxwell:1851}.}
If $\delta$ is asymmetric, we distinguish between \emph{forward} and
\emph{backward} ambits $B^+[p,r;\?f]$ and $B^-[p,r;\?f]$, applying $f$ to $x$
and $y$, respectively, where $x,y=\phi(u)$.
\end{definition}

\begin{remarks}
\label{rem:newcmpandmixed}
\begin{paras}
\item
    \label{rem:newcmp}%
    Both $\delta^\hilc{m}$ and $f\circ\delta^\hilc{m}$ are comparison
    functions on $m$-tuples from $U$, with codomains $K^\hilc{m}$ and $L$,
    respectively. Each may retain some of the properties of the original
    comparison function $\delta$. The ambit is thus a preimage of the ball
    $B[p,r]$ in the space \tup{U^\hilc{m},f\circ\delta^\hilc{m}}, as discussed
    in \cref{sec:preserve}.
\item
    It would also be possible to define ambits with mixtures of forward and
    backward comparisons, rather than just the two mentioned here. I do not
    pursue this option further.
\end{paras}
\end{remarks}

\noindent
Figure~\ref{fig:miscambits} shows a handful of sample bifocal ambits in in the
euclidean plane. Each inset shows the corresponding defining region $S$ in the
feature space $K^\hilc{m}$, i.e., the feature vectors $x$ for which $f(x)\leq
r$.

\begin{figure}
\centering
\def\D{.5}%
\def\Ax{1}%
\def\xyscale{1.7}%
\def\insetscale{0.45}%
\subcaptionbox{a ball}{\begin{tikzpicture}[larger, semithick,
        xscale=\xyscale, yscale=\xyscale]
    \footnotesize

\fill[use as bounding box, lightshade]
    (-\D - .1, -\D - .2) rectangle +(3*\D + .2, 2*\D + .4);

\draw[thick, fill=darkshade] (0,0) circle[radius=0.5*\D];

\draw
    (0,0) node[pointnode] (F1) {}
    (\D,0) node[pointnode] (F2) {}
    ;

\draw (0,.55) node (xr) {$x_1\leq r$};
\draw[thin] (0,.125) to[out=60, in=-110] (xr);

\begin{scope}[
    xshift=2*\D cm + .1cm - 2*(\D+.1)*\insetscale cm,
    yshift=-\D cm - .2 cm,
    xscale=\insetscale, yscale=\insetscale]

    \begin{scope}[xshift=.1 cm, yshift=.1 cm]

    \fill[white]
        (-.1,-.2) rectangle (2*\D + .2, 2*\D + .1)
        ;

    \begin{scope}
    \clip
        (2*\D,\D) --
        (\D,0) --
        (0,\D) --
        (\D,2*\D) --
        (2*\D,2*\D) -- cycle
        ;
    \fill[lightshade] (0,0) rectangle (2*\D,2*\D);

    \fill[darkshade] (0.5*\D,0) -- ++(0,1) -- ++(-\D,0) -- (0,0) -- cycle;
    \draw[thick] (0.5*\D,0) -- +(0,1);

    \end{scope}

    \draw[thin,<->]
        (\Ax,0) -- (0,0) -- (0,\Ax);

    \end{scope}

\end{scope}

\end{tikzpicture}}%
\hfill
\subcaptionbox{a hyperboloid}{\begin{tikzpicture}[larger, semithick,
        xscale=\xyscale, yscale=\xyscale]
\footnotesize

\def\outline{(-\D - .1, -\D - .2) rectangle +(3*\D + .2, 2*\D + .4)}

\fill[use as bounding box, lightshade] \outline;

\pgfmathsetmacro{\r}{0.5*\D}   
\pgfmathsetmacro{\c}{0.5*\D}
\pgfmathsetmacro{\a}{0.5*\r}
\pgfmathsetmacro{\b}{sqrt((\c)^2 - (\a)^2)}

\begin{scope} 

    \clip \outline;

    \fill[darkshade]
        (-1,-1) --
        plot[domain=-2:2] ({0.5*\D - \a*cosh(\x)},{\b*sinh(\x)}) --
        (-1,1) -- cycle
        ;

\end{scope} 

\begin{scope}[even odd rule] 
    \clip
        \outline
        ;

    \draw[thick]
        plot[domain=-2:2] ({0.5*\D - \a*cosh(\x)},{\b*sinh(\x)});

\end{scope} 

\draw (\D,.55) node (xm) {$x_1\!-x_2\leq r$};
\draw[thin] (-0.125,.25) to[out=10, in=220] (xm);

\path
    (0,0) node[pointnode] (F1) {}
    (\D,0) node[pointnode] (F2) {}
    ;

\begin{scope}[
    xshift=2*\D cm + .1cm - 2*(\D+.1)*\insetscale cm,
    yshift=-\D cm - .2 cm,
    xscale=\insetscale, yscale=\insetscale]

    \begin{scope}[xshift=.1 cm, yshift=.1 cm]

    \fill[white]
        (-.1,-.2) rectangle (2*\D + .2, 2*\D + .1)
        ;

    \begin{scope}
    \clip
        (2*\D,\D) --
        (\D,0) --
        (0,\D) --
        (\D,2*\D) --
        (2*\D,2*\D) -- cycle
        ;
    \fill[lightshade] (0,0) rectangle (2*\D,2*\D);

    \fill[darkshade] (-\r,0) -- ++(1,1) -- (-1,1) -- (-1,0) -- cycle;
    \begin{scope}
        \draw[thick] (-\r,0) -- ++(1,1);
    \end{scope}

    \end{scope}

    \draw[thin,<->]
        (\Ax,0) -- (0,0) -- (0,\Ax);

    \end{scope}

\end{scope}
\end{tikzpicture}}%
\hfill
\subcaptionbox{a half-space}{\begin{tikzpicture}[larger, semithick,
        xscale=\xyscale, yscale=\xyscale]
    \footnotesize

\fill[use as bounding box, lightshade]
    (-\D - .1,-\D - .2) rectangle +(3*\D + .2, 2*\D + .4);

\fill[darkshade] (-\D - .1, -\D - .2) rectangle (0.5*\D, \D + .2);

\draw
    (0,0) node[pointnode] (F1) {}
    (\D,0) node[pointnode] (F2) {}
    ;

\draw[thick, line cap=butt]
    (0.5*\D,-.7) -- +(0,1.4)
    ;

\draw (-.175,.55) node {$x_1\leq x_2$};

\begin{scope}[
    xshift=2*\D cm + .1cm - 2*(\D+.1)*\insetscale cm,
    yshift=-\D cm - .2 cm,
    xscale=\insetscale, yscale=\insetscale]

    \begin{scope}[xshift=.1 cm, yshift=.1 cm]

    \fill[white]
        (-.1,-.2) rectangle (2*\D + .2, 2*\D + .1)
        ;

    \begin{scope}
    \clip
        (2*\D,\D) --
        (\D,0) --
        (0,\D) --
        (\D,2*\D) --
        (2*\D,2*\D) -- cycle
        ;
    \fill[lightshade] (0,0) rectangle (2*\D,2*\D);

    \fill[darkshade] (0,0) -- (1,1) -- (0,1) -- cycle;
    \draw[thick] (0,0) -- (1,1);

    \end{scope}

    \draw[thin,<->]
        (\Ax,0) -- (0,0) -- (0,\Ax);

    \end{scope}

\end{scope}
\end{tikzpicture}}\\[1ex]
\subcaptionbox{an ellipsoid}{\begin{tikzpicture}[larger, semithick,
        xscale=\xyscale, yscale=\xyscale]
    \footnotesize

\def\R{0.75}

\fill[use as bounding box, lightshade]
    (-\D - .1, -\D - .2) rectangle +(3*\D + .2, 2*\D + .4);

\def\Ellipse{
    (0.5*\D,0)
    ellipse[
        x radius=0.5*\R,
        y radius=sqrt(\R*\R - \D*\D)/2
    ]
}

\fill[darkshade] \Ellipse;

\path
    (0,0) node[pointnode] (F1) {}
    (\D,0) node[pointnode] (F2) {}
    ;

\draw (2*\D + .05,.55) node[left=0] (xxr) {$x_1\!+x_2\leq r$};
\draw[thin] (.25,.15) to[out=40, in=-120] (xxr);

\begin{scope}[
    xshift=2*\D cm + .1cm - 2*(\D+.1)*\insetscale cm,
    yshift=-\D cm - .2 cm,
    xscale=\insetscale, yscale=\insetscale]

    \begin{scope}[xshift=.1 cm, yshift=.1 cm]

    \fill[white]
        (-.1,-.2) rectangle (2*\D + .2, 2*\D + .1)
        ;

    \begin{scope}
    \clip
        (2*\D,\D) --
        (\D,0) --
        (0,\D) --
        (\D,2*\D) --
        (2*\D,2*\D) -- cycle
        ;
    \fill[lightshade] (0,0) rectangle (2*\D,2*\D);

    \fill[darkshade] (0,\R) -- (\R,0) -- (0,0) -- cycle;
    \draw[thick] (0,\R) -- (\R,0);

    \end{scope}

    \draw[thin,<->]
        (\Ax,0) -- (0,0) -- (0,\Ax);

    \end{scope}

\end{scope}

\draw[thick] \Ellipse;

\end{tikzpicture}}%
\hfill
%
%
\subcaptionbox{\label{fig:hamacher}Hamacher product~\citep{Zimmermann:2001}}%
{%
\begin{tikzpicture}[larger, semithick, xscale=\xyscale, yscale=\xyscale]
\footnotesize

\def\outline{(-\D - .1, -\D - .2) rectangle +(3*\D + .2, 2*\D + .4)}

\fill[use as bounding box, lightshade] \outline;

\begin{scope} 

    \clip \outline;

    \draw[thick, fill=darkshade]
        plot file{hamacher.dat} -- cycle;

\end{scope} 

\draw (.5*\D,.55) node
    (cx) {$x_1\!\ast x_2\leq r$};
\draw[thin] (.125,0) to[out=40, in=-120] (cx);

\path
    (0,0) node[pointnode] (F1) {}
    (\D,0) node[pointnode] (F2) {}
    ;

\begin{scope}[
    xshift=2*\D cm + .1cm - 2*(\D+.1)*\insetscale cm,
    yshift=-\D cm - .2 cm,
    xscale=\insetscale, yscale=\insetscale]

    \begin{scope}[xshift=.1 cm, yshift=.1 cm]

    \fill[white]
        (-.1,-.2) rectangle (2*\D + .2, 2*\D + .1)
        ;

    \begin{scope}
    \clip
        (2*\D,\D) --
        (\D,0) --
        (0,\D) --
        (\D,2*\D) --
        (2*\D,2*\D) -- cycle
        ;
    \fill[lightshade] (0,0) rectangle (2*\D,2*\D);

    \fill[darkshade]
        plot file{tlfocusspace.dat} -- (-1,-1) -- cycle;

    \end{scope}

    \begin{scope}
    \clip
        (2*\D,\D) --
        (\D-1,0-1) --
        (0-1,\D-1) --
        (\D,2*\D) --
        (2*\D,2*\D) -- cycle
        ;
        \draw[thick] plot file{tlfocusspace.dat};
    \end{scope}

    \draw[thin,<->]
        (\Ax,0) -- (0,0) -- (0,\Ax);

    \end{scope}

\end{scope}
\end{tikzpicture}}%
\hfill%
\subcaptionbox{Cantor
    function~\citep{Cantor:1884}\label{fig:cantor}}{\begin{tikzpicture}[larger,
        semithick, xscale=\xyscale, yscale=\xyscale]

\def\outline{(-\D - .1, -\D - .2) rectangle +(3*\D + .2, 2*\D + .4)}
\path[use as bounding box] \outline;

\begin{scope}[
    spy using outlines={circle, magnification=4, size=1.5cm, connect spies}
]
\spy on (1.065,.45) in node at (2*\D,0.85*\D)
    ;

\footnotesize

\fill[use as bounding box, lightshade] \outline;

\begin{scope} 

    \clip \outline;

    \draw[very thin, fill=darkshade, line join=miter]
        plot file{cantorlipse.dat} -- cycle;

\end{scope} 


\draw (-\D - .05,.55) node[right=0] (cx) {$\psi(x_1)\!+\psi(x_2)\leq r$};
\draw[thin] (.25,.2) to[out=140, in=-60] (cx);

\path
    (0,0) node[pointnode] (F1) {}
    (\D,0) node[pointnode] (F2) {}
    ;

\begin{scope}[
    xshift=2*\D cm + .1cm - 2*(\D+.1)*\insetscale cm,
    yshift=-\D cm - .2 cm,
    xscale=\insetscale, yscale=\insetscale]

    \begin{scope}[xshift=.1 cm, yshift=.1 cm]

    \fill[white]
        (-.1,-.2) rectangle (2*\D + .2, 2*\D + .1)
        ;

    \begin{scope}
    \clip
        (2*\D,\D) --
        (\D,0) --
        (0,\D) --
        (\D,2*\D) --
        (2*\D,2*\D) -- cycle
        ;
    \fill[lightshade] (0,0) rectangle (2*\D,2*\D);

    \draw[thick, fill=darkshade]
        plot file{clfocusspace.dat}
            -- ++(-1,0) -- (-1,-1) -- (1,-1) -- cycle;

    \end{scope}

    \draw[thin,<->]
        (\Ax,0) -- (0,0) -- (0,\Ax);

    \end{scope}

\end{scope}
\end{scope} 

\draw[thick, line join=miter]
    plot file{cantorlipse.dat} -- cycle;

\end{tikzpicture}}%
\caption{Bifocal ambits with linear (a--d), non-linear (e) and fractal (f)
remoteness usable in metric spaces. All but the Hamacher product preserve the
triangle inequality}%
\label{fig:miscambits}%
\end{figure}
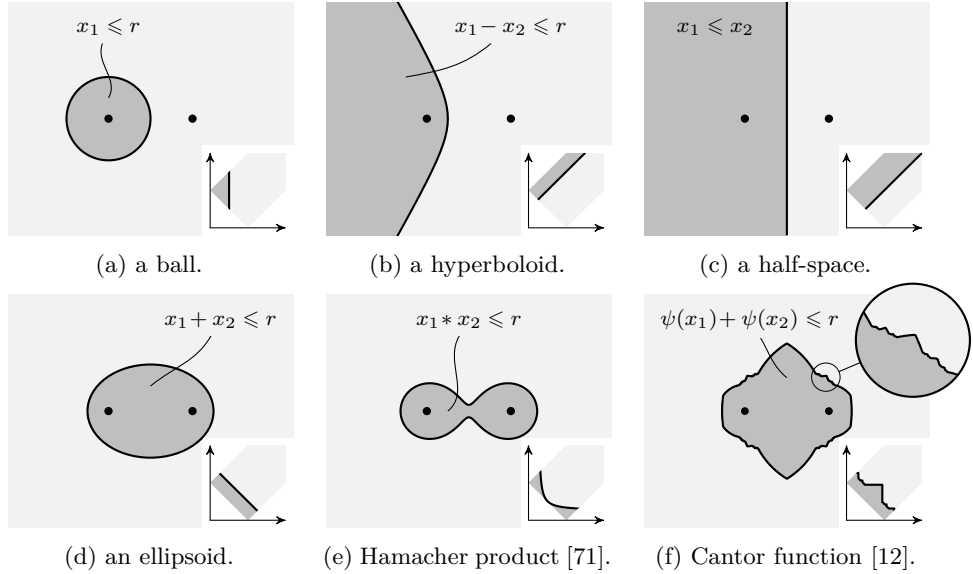

The third simplifying assumption is a temporary one, for
\cref{sec:linearmetric,sec:focusselection,sec:emulatingregions,sec:optcoeff},
namely that the remoteness map $f$ is \emph{linear}. We refer to the
resulting regions as \emph{linear ambits}. For a linear ambit with remoteness
map $f(x)=Ax$, where $A$ is a $d\times m$ matrix, we may simply write
$B[p,r;A]$, as no confusion is likely to arise.
I discuss the more general case in
\cref{sec:preserve,sec:nonl}, but choosing the specific case of linearity
affords us the opportunity to go into more detail. It also seems to cover the
region types used in existing indexing method (cf.\
\cref{sec:emulatingregions}).

For one thing, using linear remoteness means $K^\hilc{m}$ and $L$ are ordered
modules; I'll simply use $\R^\hilc{m}$ and $\R^\hilc{d}$, for some $d$. For
another, structure preservation simply amounts to linearity, so that is what
we have to work with in our axioms for $\A$ and $\B$. In $\B$ we are working
with balls (cf.\ \cref{rem:newcmpandmixed}\ref{rem:newcmp}), and so we might
wish to have overlap between regions in $\A$ to imply overlap between balls in
$\B$, meaning we'd need a linear inequality expressing a necessary condition
for the overlap of two balls $B[p,r]$ and $B[q,s]$ in terms of the comparison,
$\delta$. The obvious choice for such a condition is
\begin{equation}
    \label{eq:simpleballoverlap}
    r + s \geq \delta(p,q)
\end{equation}
which is equivalent to the (oriented) triangle inequality of \emph{metric and
quasimetric spaces}.\endnote{The triangle inequality may be derived from ball
overlap and vice versa.} More generally, it is equivalent to any composition
law~\citep{Lawvere:1973}, including the transitivity of partial orders, which
is used to detect overlap between intervals, for example.\endnote{For
quasminetric spaces, the balls should have opposite direction. The
triangle inequality immediately implies \cref{eq:simpleballoverlap}, but the
converse also holds: For any $u$, let $r=\delta(p,u)$ and $s=\delta(u,q)$, and
triangularity follows.}

\subsection{Linear Ambit Overlap in Quasimetric Space}
\label{sec:linearmetric}

A linear ambit $B[p,r;A]=C[p,S]$ is determined by an arbitrary polyhedron
\[
    S \deq \set{x\in\R^\hilc{m}:Ax\leq r}
\]
in the feature space, so deciding whether two linear ambits $C[p,S_1]$ and
$C[p,S_2]$ intersect is equivalent to deciding whether two arbitrary polyhedra
$S_1,S_2\subseteq\R^\hilc{m}$ intersect. The intersection is also an
arbitrary polyhedron, so we can reduce from the problem of linear program
feasibility. Such a check may be time-consuming, and becomes unwieldy when
generalizing to ambits with differing foci.

Instead, I will examine the uniradial case ($d=1$), and handle the multiradial
case ($d>1$) by performing intersection checks for each defining half-space of
the polyhedron $S$ individually. If some other region $P$ is disjoint from
any of these half-spaces, it cannot intersect $S$, so this approach will incur
no false negatives. In principle, however, it is possible for $P$ to intersect
every defining half-space of $S$, and yet not intersect $S$ itself, producing
a false positive.\endnote{Note that for the existing regions emulated in
\cref{sec:emulatingregions}, overlap checks with ball queries will \emph{not}
add any false positives, as all defining hyperplanes are axis-orthogonal.}

Alternatively, rather than viewing this as an approximation, it can be seen as
merely a compact way of representing multiple uniradial ambits sharing a set
of foci, where we require intersection with each ambit; this is exactly what
is required in sprawl traversal (\cref{sec:sprawltrav}), meaning that each
edge label (\cref{def:sprawl}) may be represented by a single
augmented matrix $(A\mid r)$, and be treated as a single multiradial ambit
using the overlap check described here.

We'll be examining two uniradial ambits $R\deq B[p,r;a]$ and $Q\deq B[q,s;c]$,
where $a$ and $c$ are row vectors of \emph{focal weights}, one of which is
non-negative (i.e., $a\geq 0$ or $c\geq 0$).\endnote{In older text on weighted
polyellipses, the weight is also referred to as the \emph{power} of the
focus.} For arbitrary comparisons, we have no way of relating the two regions;
this is exactly why we are assuming a quasimetric space. Specifically, for
foci $p_i$ and $q_j$, and some other object $u$, the \emph{oriented triangle
inequality} holds:
\begin{equation}
    \label{eq:orienttriang}
    \delta(p_i,u) + \delta(u,q_j) \geq \delta(p_i,q_j)
\end{equation}
We can create what amounts to a weighted sum of such triangles, as in the
following \lcnamecref{thm:ambits}.\footnote{Here $\|\blank\|_1$ is the 1-norm,
which for the non-negative coefficient vector is simply a sum.}

\begin{lemma}
    \label{thm:ambits}
    For a quasimetric space
    \tup{U,\delta}, let
        $x_i = \delta(p_i,u)$,
        $y_j = \delta(u,q_j)$
        and
        $z_{ij} = \delta(p_i,q_j)$,
    where $p_i,q_j,u\in U$, for $i=1\dots m, j=1\dots k$. Then
    \begin{equation}
        \label{eq:ambits}
        \|c\|_1ax + \|a\|_1cy \geq a\wZ\tr{c}\,,
    \end{equation}
    whenever $a\in\R^{1,m}$, $c\in\R^{1,k}$ and either $a$ or $c$ is
    non-negative.
\end{lemma}
\begin{proof}
We want to show that the inequality
\begin{equation}
    \label{eq:ambitscomponent}
    |c_j|a_ix_i + |a_i|c_jy_j \geq a_ic_jz_{ij}
\end{equation}
holds for any indices $i$ and $j$; summing then produces~\eqref{eq:ambits}.
Assume (\textsc{wlog}) that $c$ is the non-negative
one. If $a_i\geq 0$, \eqref{eq:ambitscomponent} follows directly. If $a_i <
0$, we have
\begin{align*}
    a_ic_jx_i
    &\geq
    a_ic_jy_j
    +
    a_ic_jz_{ij}
    \\
    a_ic_jx_i
    -a_ic_jy_j
    &\geq
    a_ic_jz_{ij}\,,
\end{align*}
which again yields~\eqref{eq:ambitscomponent}.
\end{proof}

\noindent
This yields our desired linear ambit overlap check for quasimetric spaces:

\begin{theorem}
    \label{thm:overlap}
    Let $R\deq B^+[p,r;a]$ and $Q\deq B^-[q,s;c\mkern1mu]$
    be linear ambits for a quasimetric~$\delta$,
    with either $a$ or $c$ non-negative and
    $\|a\|_1,\|c\|_1=1$. If $R$ and $Q$ intersect, then
    \begin{equation}
        \label{eq:overlap}
        r + s \geq a\wZ\tr{c}\eqcomma
    \end{equation}
    where $z_{ij}$ is the distance $\delta(p_i,q_j)$ between focus $p_i$ of
    $R$ and focus $q_j$ of $Q$.
\end{theorem}
\begin{proof}
    Follows from \cref{thm:ambits}, given that $ax\leq r$ and $cy\leq s$.
\end{proof}

\noindent
The requirement $\|a\|_1,\|c\|_1=1$ is easily enforced by scaling the radii;
it simplifies the check, but is not crucial. An intuition for the more
general overlap check is given in \cref{fig:weightedellipse}, for the case
where $Q$ is a ball. (See also \cref{fig:ambitexample}.)


\begin{figure}%
%
\input{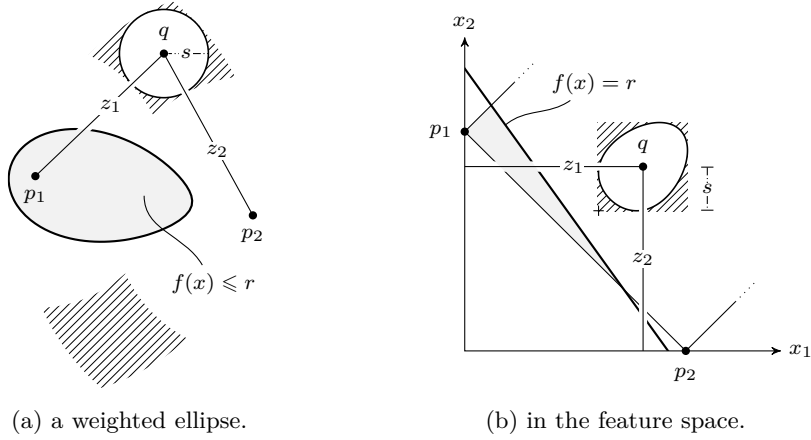}
%
\def\hl{.2pt}
\hfill
\subcaptionbox{a weighted ellipse}{
\begin{tikzpicture}[larger, semithick]

    \footnotesize

    \def\regionpath{plot file{ellipse.dat} -- cycle}

    \def\MyBBox{(-1.38,-2) rectangle (1.2,2)}

    \path[use as bounding box] \MyBBox;

    \fill[lightshade] \regionpath;

    \coordinate (Q) at (\Qx,\Qy);

    \coordinate (u) at (\ux,\uy);

    \coordinate (F1) at (\pax,\pay);
    \coordinate (F2) at (\pbx,\pby);

    \begin{scope}[even odd rule]


    \foreach \k in {1,2} {

        \clip
            let \p1 = ($(Q)-(F\k)$) in
                (F\k) circle ({veclen(\x1,\y1)-\s cm-\hl})
                (F\k) circle ({veclen(\x1,\y1)+\s cm+\hl})
            ;

    }

    \foreach \x in {0,1,...,48} {
        \draw[very thin] (-2.6,-.5) ++(\x*1.2 pt,-\x*1.2 pt) -- +(3,3);
    }

    \end{scope}

    \draw[fill=white] (Q) circle[radius=\s];
    \draw[thin, densely dashdotted] (Q) -- +(\s,0)
        node[midway, fill=white,inner sep=.5pt] {$s$};

    \path[name path=horiz] (-1,.77) -- ++(2.5,0);

    \begin{scope}[even odd rule]
        \clip \MyBBox
        ($(F1)!1.5pt!90:(Q)$)
        --
        ($(Q)!1.5pt!-90:(F1)$)
        --
        ($(Q)!1.5pt!90:(F1)$)
        --
        ($(F1)!1.5pt!-90:(Q)$)
        -- cycle
        ;
        \draw[thick] \regionpath;
    \end{scope}

    \draw[line width=3pt,draw=white]
        (Q) -- ($(Q)!.5!(F1)$)
        ;

    \draw[thin] (F1) -- (Q)
        node[pos=0.575, fill=white, inner sep=1pt] {$z_1\!$}
        ;

    \draw[line width=3pt,draw=white]
        (Q) -- (F2)
        ;

    \draw[thin] (F2) -- (Q)
        node[pos=0.4, fill=white, inner sep=1pt] {$z_2$}
        ;

    \fill
        (Q) node[pointnode] {} node[above=2pt] {$q$}
        ;

    \fill
        (F1) node[pointnode] {}
        (F2) node[pointnode] {}
        ;

    \draw
        (F1) node[below=2pt] {$p_1$}
        (F2) node[below=2pt] {$p_2$}
        ;

    \path
        (F1) -- (F2)
        coordinate[pos=.35] (lab)
        ;

    \draw (.7,-.9) node (H) {$f(x)\leq r$};
    \draw[thin] (0,0) to[out=-90, in=120] (H);

\end{tikzpicture}}%
\hfill
\hfill
\hfill
\subcaptionbox{\label{fig:featurespace}in the feature space}{
\begin{tikzpicture}[larger, semithick]

    \footnotesize

    \path (0,-.4) -- (0,3.6);

    \def\pdist{veclen(\pax-\pbx,\pay-\pby)}
    \coordinate (F1) at ({\pdist},0);
    \coordinate (F2) at (0,{\pdist});

    \path
        (F1) ++(.5,.5) coordinate (R)
        (F2) ++(.5,.5) coordinate (T)
        ;

    \coordinate (Q) at
        (
            {veclen(\Qx-\pax,\Qy-\pay)},
            {veclen(\Qx-\pbx,\Qy-\pby)}
        )
        ;
    \coordinate (u) at
        (
            {veclen(\ux-\pax,\uy-\pay)},
            {veclen(\ux-\pbx,\uy-\pby)}
        )
        ;
    \coordinate (O) at (0,0);

    \coordinate (A) at (Q -| O);
    \coordinate (B) at (Q |- O);

    \def\regionpath{(0,\roverwb) -- (\roverwa,0)}
    \path[name path=locus] \regionpath;
    \path[name path=feasible] (T) -- (F2) -- (F1) -- (R);
    \fill[lightshade, name intersections={of=locus and feasible}]
        (intersection-1) -- (F2) -- (intersection-2) -- cycle
        ;

    \begin{scope}[even odd rule]
        \clip
            (0,0) rectangle (3,3)
            (A) ++(0,-1.5pt) rectangle ++(2,3pt)
            ;

        \draw[thin] (F1) -- (F2);

        \draw[thin] (F2) -- (T);
        \draw[thin, dotted] (T) -- +(.2,.2);

        \draw[thin] (F1) -- (R);
        \draw[thin, dotted] (R) -- +(.2,.2);


        \draw[thick, line cap=rect]
            \regionpath
            coordinate[pos=0.2] (lab)
            ;

    \end{scope}

    \draw (lab)
        ++(.6,.25)
        ++(.3,0)
        node[above,inner sep=2pt] (H)
        {$f(x) = r$};

    \draw[thin] (lab) to[out=40,in=-145] (H);

    \begin{scope}
    \clip
        (Q) ++(-\s cm-\hl,-\s cm-\hl) rectangle +(\hl+2*\s cm,\hl+2*\s cm);

    \foreach \x in {0,1,...,35} {
        \draw[very thin] (0,1.5) ++(\x*1.2 pt,-\x*1.2 pt) -- +(2,2);
    }

    \end{scope}

    \draw[thin, densely dashdotted, |-|] (Q) ++(.2cm + \s cm,0) -- +(0,-\s)
        node[midway, fill=white,inner sep=0.5pt] {$s$};

    \fill
        (F1) node[pointnode] {} node[below] {\strut$p_2$}
        (F2) node[pointnode] {} node[left=1.5pt] {$p_1$}
        ;

    \draw[fill=white] plot file{scaledquery.dat} -- cycle;

    \draw[line width=3pt,draw=white] (Q) -- ++(-.5,0);

    \draw[thin] (A) -- (Q)
        node[pos=0.6, fill=white, inner sep=1pt] {$z_1\!$}
        ;

    \draw[line width=3pt,draw=white] (Q) -- (B);
    \draw[thin] (Q) -- (B)
        node[pos=0.5, fill=white, inner sep=1pt] {$z_2$}
        ;

    \fill
        (Q) node[pointnode] {} node[above=2pt] {$q$}
        ;

    \draw[thin,<->]
        (3.2,0) node[right] {$x_1$} -- (0,0) -- (0,3.2) node[above] {$x_2$};

    \draw[thin] (Q) ++(-\s,-\s)
        +(-1.5pt,0) -- +(1.5pt,0)
        +(0,-1.5pt) -- +(0,1.5pt)
        ;

\end{tikzpicture}}%
\hfill\mbox{}
\caption{For an unknown metric, any point in the hatched box could be
relevant. The shaded region could only be of interest if the lower left corner
$z-s$ falls within it. This happens when $f(z-s)\leq r$ or, equivalently,
$f(z)-f(\Delta(s))\leq r$. These two conditions generalize, though
differently, to some nonlinear cases, as discussed in
\cref{sec:nonl,sec:preserve}%
\label{fig:weightedellipse}}%
\end{figure}

\begin{remark}
    In the euclidean plane, linear ambits correspond to weighted polyellipses,
    which were described already by James Clerk Maxwell~\citep{Maxwell:1851}.
\end{remark}

\subsection{Emulating Existing Regions}
\label{sec:emulatingregions}

The intervals, hyperrectangles and other polytopes of ordered and spatial
indexing are quite easily implemented using linear ambits, by introducing one
pseudo-focus for each dimension, and using projection as the comparison, i.e.,
$\delta(x,i\mkern1mu)=x_i$. Then the linear ambit devolves to a convex
polyhedron in the space being indexed. For non-convex polyhedra, one could
simply use the union of several ambits, each belonging to a separate sprawl
edge. Of more interest, perhaps, is the case where $\delta$ is actually a
quasimetric or metric, using the overlap check of \cref{sec:linearmetric}.

Existing methods for indexing metric spaces are quite varied, and, as
\citeauthor{Lokoc:2014} put it, ``After two decades of research, the
techniques for efficient similarity search in metric spaces have combined
virtually all the available tricks''~\citep{Lokoc:2014}. These ``tricks'' have
been available from the beginning. When \citeauthor{Novak:2011} use
``practically all known principles of metric space partitioning, pruning, and
filtering'' in their method~\citep{Novak:2011}, they are still essentially
referring to the balls, spheres and planes introduced by
\citeauthor{Fukunaga:1975}, \citeauthor{Ruiz:1986}, and
\citeauthor{Uhlmann:1991a}~\citep{Fukunaga:1975,Ruiz:1986,Uhlmann:1991a}. And
in metric indexing in general, the query (for range search) or its
inclusion-wise upper bound (for $k$NN, cf.\ \ref{rem:knn}) is simply a
ball.\endnote{One notable exception is \citet{Uhlmann:1991}, who briefly
discusses half-plane queries.}

These regions are, of course, simple special cases of the linear ambit, all
with coefficients of~$1$ or $-1$. \Cref{tab:existing} provides an overview of
the main region types.
\begin{table}%
\def\tabwidth{.6\textwidth}%
\def\capwidth{.355\textwidth}%
\renewcommand{\arraystretch}{1.3}%
\def\PL{\phantom{+}}%
\makeatletter
\def\dotf{\leavevmode \leaders \hb@xt@ .75em{\hss .\hss }\hfill \kern \z@}
\makeatother
\newcommand{\twocite}[2]{%
    \spcl{\mbox{00}}{\mbox{\tabcite{#1}}},\,%
    \spcl{\mbox{00}}{\mbox{\tabcite{#2}}}%
}
\begin{minipage}[t]{\tabwidth}\mbox{}\\[2pt]
\small%
\begin{tabularx}{\textwidth}{@{}X@{}r@{\hspace{.6cm}}rr@{}}
\toprule
\bf Region type
&
\textbf{Ref.}
& $\textbf{A}\phantom{]}\transpace$
& $\textbf{\textit{r}}\phantom{]}\transpace$ \\
\midrule
Ball \dotf & \twocite{Kalantari:1983}{Uhlmann:1991} & $\pm 1$\phantom{]}\transpace & $\pm r\phantom{]}\transpace$\\
Sphere \dotf & \twocite{Ruiz:1986}{Mico:1994} & $\tran{[1, -1]}$ & $\tran{[r{}\mkern-1mu, -r]}$\\
Shell \dotf & \twocite{Brin:1995}{Skopal:2004a} & $\tran{[1, -1]}$ & $\tran{[r\mathrlap{'}\mkern-1mu, -r]}$\\
Plane \dotf & \twocite{Uhlmann:1991a}{Bugnion:1993} & $[1, -1]\transpace$ & 0\phantom{]}\transpace\\
Ellipse \dotf & \twocite{Uhlmann:1991}{Dohnal:2001} & $[1, \PL 1]\transpace$ & $r$\phantom{]}\transpace\\
Hyperbola \dotf & \twocite{Dohnal:2001}{Lokoc:2010} & $[1, -1]\transpace$ & $r$\phantom{]}\transpace\\
Voronoi cell \dotf & \twocite{Navarro:2001}{Navarro:2002} & $[\spcc{1}{\one},
-\spcc{1}{\Id}]\transpace$ & $0$\phantom{]}\transpace\\
Cut region \dotf & \twocite{Skopal:2004a}{Lokoc:2014} & $\tran{[\spcc{1}{\Id}, -\spcc{1}{\Id}]}$ &
$\vec{r}\phantom{]}\transpace$
\\
\bottomrule
\end{tabularx}%
\end{minipage}%
\hfill
\begin{minipage}[t]{\capwidth}\mbox{}\\[-1.7\baselineskip]
\caption{Some existing region types viewed as ambits. The radii $r$ and $r'$
are non-negative scalars, while $\vec{r}$ is a vector with $m$ positive and
$m$ negative elements; \one\ is an all-ones column vector and \Id\ is the
identity matrix.
Some region types (e.g., balls) are very common, while others (e.g.,
hyperbolas and ellipses) are quite rare. The references are somewhat arbitrary
examples of their use}\label{tab:existing}%
\end{minipage}
\end{table}
The ball is sometimes used as a bisector between inside and outside; the
complemented ball can be represented by negating the coefficient and
radius.\endnote{Note that in some existing methods, the complemented ball is
defined by a strict inequality.} A sphere is essentially the intersection of a
ball and its (closed) complement, and this is reflected in the coefficients
and radii. The technique of using sphere regions is generally known as
\emph{pivot filtering}. A shell simply lets the two radii of a sphere vary
independently. The plane (or \emph{generalized hyperplane}), ellipse and
hyperbola are metric equivalents of the corresponding conics, and the Voronoi
cells (or Dirichlet domains) and cut regions are intersections of plane and
shell regions, respectively.\endnote{As originally described, one of the
components of a cut region is a ball. This corresponds to dropping the last
row of the coefficient matrix.}

It can be instructive to see how the linear ambit overlap
check~\eqref{eq:overlap} reduces to the commonly used conditions in metric
indexing when instantiated with the parameters from \cref{tab:existing}.
These overlap checks assume that $Q$ is a ball, which simplifies the check to
\begin{equation}
    \label{eq:balloverlap}
    r + \|a\|_1s \geq az\,,
\end{equation}
if we don't require normalized coefficients. This check then applies
individually to each facet of the polyhedron (i.e., each row of $\mA$ and
$r$).\endnote{The feature polyhedra of these regions are all axis-orthogonal,
so checking one facet at a time will \emph{not} produce additional false
positives.}
For example, for the three staples ball, sphere and plane, we get the
respective conditions
\begin{align*}
    r + s \geq z
          &&
    \begin{matrix}
    \phantom{-}r + s \geq \phantom{-}z\\
    -r + s \geq -z
    \end{matrix}
    &&
    0 + 2s \geq z_1-z_2\mathrlap{\,,}
\end{align*}
or the perhaps more familiar
\begin{align*}
    s \geq z - r
        &&
    s \geq |z - r|
        &&
    s \geq \frac{z_1 - z_2}{2}\mathrlap{\,,}
\end{align*}
precisely as expected.
The other overlap checks described
by~\mbox{\citet[\S\,7]{Zezula:2006}}, for example, may be derived similarly.
Note, though, that I am focusing on \emph{necessary} conditions for
relevance; bounds providing \emph{sufficient} conditions are a separate
matter, but could be dealt with in a similar fashion.
Some structures use weakened versions of the filtering criteria, where a
pivot-based bound is substituted for
$z$~\citep[e.g.,\kern-1.5pt][]{Zezula:1996,Mico:1996}. The result is still
linear, and is easily emulated by letting the pivots be parents of the region
as well, increasing the number of radii as appropriate.

\begin{adhoc}{Ambits as queries}
    \label{rem:ambitqueries}
    If the regions of existing index structures are simply reinterpreted as
    linear ambits, one is free to query any current such data structure using
    a linear ambit \emph{as the query}. For example, rather than querying an
    M-Tree~\citep{Zezula:1996} or the like using a single object, one could
    use a weighted sum of \emph{multiple} objects, some of which could even
    act as contrasts to be avoided, using negative weights. In these cases,
    one could also use the common $k$NN approach of updating the query radius
    during the search, thus returning the $k$ points with the lowest query
    remoteness.
\end{adhoc}

\subsection{Finding Optimal Coefficients}
\label{sec:optcoeff}

Rather than fixing the ambit coefficients arbitrarily, we might wish to
determine their optimal values, based on the given data.
An obvious goal in setting the coefficients is to minimize the probability of
query overlap, which in turn will reduce the number discoveries and increase
the number of eliminations, improving the running time overall.\footnote{Note,
though, that in theory spurious overlaps could improve performance (cf.\
\ref{rem:falsepos})!}
Assume that the foci and responsibilities are given, but that the number of
radii is unconstrained. If intersection of the feature polyhedron and another
fixed set is what determines overlap, the following
proposition
provides an optimal solution, regardless of any probability distributions
involved.
The constants hidden by the asymptotic notation depend on $m$, which is fixed.

\begin{proposition}
    \label{thm:convexhull}
    Given a set of $m$ foci, an inclusion-wise minimum linear ambit can be
    found with worst-case running time \BigOh{n\log n + n^{\lfloor
    m/2\rfloor}}, as a function of the number of responsibilities, $n$. This
    running time is optimal for metric spaces.
\end{proposition}
\begin{proof}
    The task is to determine the coefficients $\mat{A}$ that make the defining
    polyhedron \mbox{\set{x:\mat{A}x\leq r}} an inclusion-wise minimum. This
    minimum is simply the convex hull (in $\R^m$) of the $n$ responsibility
    feature vectors, which can be found with the stated running
    time~\citep{Chazelle:1993}. Each facet of the hull then corresponds to one
    row of $\mat{A}$ and one component of $r$.

    That the running time is optimal can be shown by reduction from the
    general convex hull problem, for which the given running time is
    optimal~\citep[Th.\,3.5]{Chazelle:1993}. That is, for any instance of the
    general convex hull problem, we construct points in a metric space for
    which the optimal linear ambit would provide us with the original convex
    hull. Let $\X\deq\set{x_{ij}}$ be a an $m\times n$-matrix whose column
    vectors are the input points to the convex hull problem. Introduce two
    non-overlapping, arbitrary sets of foci and responsibilities $p_i$ and
    $u\mkern-1mu_j$, and some large value~$\omega$, such as
    $\|\mat{X}\|_\infty+\varepsilon$ ($\varepsilon>0$). Let self-distances be
    zero, let the distance between foci and between responsibilities be
    $\omega$, and let
    \[
        \delta(u\mkern-1mu_j,p_i) = \delta(p_i,u\mkern-1mu_j) = x_{ij} +
        \omega\eqdot
    \]
    Finding an inclusion-wise minimum linear ambit will produce a convex hull
    for $\mat{X}+\omega$, whose vertices may be shifted by $-\omega$ to
    produce the final result. It is easily verified that $\delta$ is a metric:
    Self-distances are zero, and all other distances are symmetric and
    strictly positive; because they are all in the range $[\omega,2\omega]$,
    they also satisfy the triangle inequality. The running time of this
    reduction is linear in $n$, for any given $m$, which means that the
    worst-case bound of $\Omega(n\log n + n^{\lfloor m/2\rfloor})$ carries
    over from the convex hull problem.
\end{proof}

\begin{figure}
\input{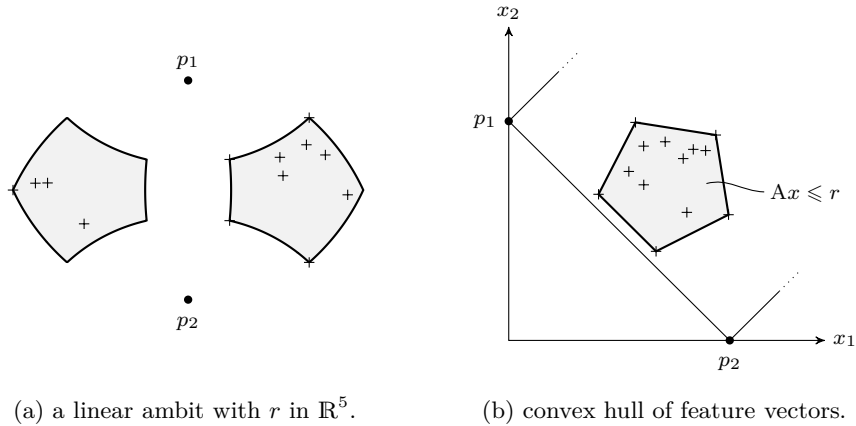}%
\newsavebox{\leftfig}
\newsavebox{\rightfig}
\sbox{\leftfig}{%
\begin{tikzpicture}[larger, semithick]
\footnotesize
\foreach \i in {1,...,2} {
    \draw[fill=lightshade, thick, line join=miter]
        plot file{hull_\i.dat} -- cycle;
}
\draw plot[resp] file{hullpts.dat};
\draw
    (\pax,\pay) node[pointnode] (P1) {} node[above] {\strut$p_1$}
    (\pbx,\pby) node[pointnode] (P2) {} node[below] {\strut$p_2$}
    ;
\end{tikzpicture}}%
\sbox{\rightfig}{%
\begin{tikzpicture}[larger, semithick]
\footnotesize
\draw[fill=lightshade, thick] plot file{hullfeatvert.dat} -- cycle;
\draw plot[resp] file{hullfeat.dat};
\draw[thin,<->,>=stealth']
    (0,3.2) node[above] {$x_2$} -- (0,0) -- (3.2,0) node[right] {$x_1$};

\def\d{.5}
\filldraw
    (\pdist,0) node[pointnode] {} node[below] {\strut$p_2$}
    (0,\pdist) node[pointnode] {} node[left=1.5pt] {$p_1$}
    ;
\draw[thin] (\d,\pdist + \d) -- (0,\pdist) -- (\pdist,0) -- (\pdist + \d,\d);
\draw[thin, dotted]
    (\d,\pdist + \d) -- +(.2,.2)
    (\pdist + \d,\d) -- +(.2,.2)
    ;

\draw (3,1.5) node[inner sep=2pt] (ineq) {$\mA x\leq r$};
\path[thin] (2,1.6) edge[out=15,in=178] (ineq);
\end{tikzpicture}}%
\centering
\subcaptionbox{a linear ambit with $r$ in $\R^\sides$}[.48\linewidth]{%
\raisebox{(\ht\rightfig-\ht\leftfig)/2}{\usebox{\leftfig}}%
}%
\hfill
\subcaptionbox{convex hull of feature vectors}[.48\linewidth]{%
\usebox{\rightfig}%
}%
\caption{The optimal linear ambit is given by a convex hull in the feature space}
\label{fig:convexhull}%
\end{figure}

\noindent
For an example of an ambit corresponding to the convex hull of the
responsibility feature vectors, see \cref{fig:convexhull}.
Note that if we have fewer restrictions on~$\delta$ (e.g., if it is only a
quasimetric), the optimality of \cref{thm:convexhull} would still hold; this
would merely simplify the reduction. Other constraints may also be
accommodated in a similar manner.

While this result provides us with a method of finding optimal linear ambits,
it is not entirely practical, in that we might end up with $\Theta(n^{\lfloor
m/2\rfloor})$ radii, which is unacceptable in any actual
implementation.\endnote{The worst case is
$
\smash{
f_{m-1}(C_m(n))=
\binom{n-\lceil m/2\rceil}{\lfloor m/2\rfloor} +
\binom{n - 1 - \lceil(m-1)/2\rceil}{\lfloor (m-1)/2\rfloor}
}
$
facets~\mbox{\citep[p.\,25 \& \S\,8.4]{Ziegler:1995}}.}
Let's consider, instead, the case where we restrict ourselves to a single
facet---a hyperplane in the feature space.

In the following, let $z=\Z\mkern1mu\tr c$. For unifocal queries with $c=1$,
$z$ retains its interpretation as the feature vector of the query center $q$.
In general, the overlap condition from \cref{thm:overlap} becomes $r+s\geq
az$, and we wish to minimize the probability of this being true. Absent a
probability distribution for $s$ and $z$, we resign ourselves to maximizing
the expected filtering lower bound $\ell$, that is,
\begin{equation}
    \label{eq:optcoeff}
    \ell_\opt\, \deq \,
    \textstyle
    \max_{a,r} \E[az - r]
    \quad
    \stlap
    \quad
    \|a\|_1=1\,,\enskip \forall x\; ax\leq r
    \eqcomma
\end{equation}
where $x$ ranges over the feature vectors of the ambit's responsibilities.
This expectation is, of course, a heuristic proxy for the actual overlap
probability, though not an unreasonable one, and not one without precedent in
metric indexing.\endnote{For example, \citet{Bustos:2003} use the same
heuristic for pivot selection.}
If $\ell_\opt\leq 0$, we say the corresponding ambit is \emph{degenerate}.

In the following, let $\one = \tr{[1\,\dots\,1]}$ be an all-ones column vector
and let $\X\deq\set{x_{ij}}$ be a non-negative $m\times n$ matrix whose
columns are the responsibility feature vectors.  Let $\hat{z}=\E[z]$, in
practice estimated by the average of $z$ over a set of \emph{training
queries}~\citep{Edsberg:2010}. We can efficiently determine the optimal
coefficients for a non-degenerate ambit using the following linear program.
\begin{equation}\label{eq:optcoefflp}\arraycolsep=2pt
    \def\myrowsep{1.7ex}
    \begin{array}{ll@{\hspace{1.25em}}lclclcl@{\,}l}%
        \mathllap{\ell'_\opt\deq}&
        \max_{u,v,r} & u\hat{z} & - & v\hat{z} & - &r\\[\myrowsep]
        &\st &u\one & + & v\one &&& = & 1 &\mathrlap{,}\\[\myrowsep]
        && u\X & - & v\X & - & r & \leq & 0 &\mathrlap{,}\\[\myrowsep]
        &\& & u\eqcomma && v &&& \geq & 0 &\mathrlap{.}
    \end{array}
\end{equation}

\begin{proposition}\label{thm:opt}
    $\ell'_\opt = \max{}\set{\ell_\opt,\mkern1mu 0}$, with $a=u-v$ if
    $\ell_\opt$ is positive.
\end{proposition}
\begin{proof}
\def\probA{\eqref{eq:optcoeff}}
\def\probB{\eqref{eq:optcoefflp}}
Given an arbitrary feasible solution to~\probA, let
\begin{equation*}
u_i = \max{}\set{a_i,0}
\quad\text{and}\quad
v_i = -\min{}\set{a_i,0}\,,
\end{equation*}
for all $i$, so $a=u-v$ and $|a_i|=u_i+v_i$, where $u_i, v_i\geq
0$. By construction, the objective value is preserved. Furthermore,
\mbox{$(u+v)\one = \|a\|_1 = 1$}, and, for each column vector $x$ of $\X$, we
have $(u-v)x = ax \leq r$, so the solution is feasible for~\probB. Of course,
this holds \emph{a fortiori} for positive optimal solutions to~\probA.

Now assume a positive optimal solution to \probB, and let $a=u-v$. The
objective value is then preserved in~\probA, and the radius $r$ remains valid;
we need only show that $\|a\|_1=1$, or more specifically, that
\mbox{$|u_i-v_i|=u_i+v_i$} for each~$i$. Assume, for the sake of
contradiction, that for some $i$ we have $|u_i-v_i|<u_i+u_i$. We can
then reduce $u_i+v_i$ by subtracting $\min{}\set{u_i,v_i}$ from both
variables, preserving the objective value and the second constraint. Let $\xi$
be the resulting value of $(u+v)\one$.
The optimum is positive, meaning $u\neq v$, so even if $u_i=v_i$, we have
$\xi>0$.
We can therefore restore the first constraint by dividing all variables by
$\xi$, and because $\xi<1$ and the objective is positive, this increases the
objective value. As we started with the optimum, this is a contradiction.

In other words, for every positive optimal solution to either
problem, there is an equivalued feasible solution to the other, where
$a=u-v$. Thus if either optimum is positive, the two are identical, and so
for this case, the
proposition
holds.
A solution where $u=v$ and $r=0$ is always feasible for~\probB,
which means that $\ell'_\opt\geq 0$. Consequently, if $\ell_\opt\leq 0$, the only
value available to $\ell'_\opt$ is $0$, proving the remaining case.
\end{proof}

\noindent
If $\ell'_\opt=0$, the ambit is degenerate, and can never be eliminated, so
restructuring the sprawl is probably a better strategy than determining the
optimal coefficients for the given foci and responsibilities. See
\cref{fig:samplesinglefacet} for some ambits constructed
using~\programref{eq:optcoefflp}.

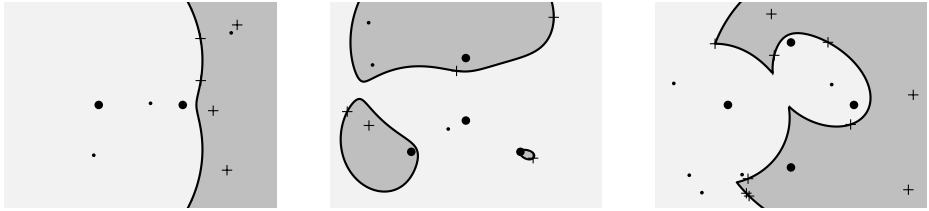
\begin{figure}
\centering
\def\inside{darkshade}
\def\outside{lightshade}
\def\D{.5*1.7}%
\def\outline{(-\D - .1, -\D - .2) rectangle +(3*\D + .2, 2*\D + .4)}
\begin{tikzpicture}[larger, semithick]
\fill[use as bounding box, \inside] \outline;
\begin{scope}
    \clip \outline;
\foreach \i in {1} {
    \draw[fill=\outside,
        thick, line join=miter]
        plot file{unifac_1_\i.dat} -- cycle;
}
\draw plot[resp] file{unifacpts_1.dat};
\draw plot[only marks, mark=*, mark size=1pt] file{unifacfoci_1.dat};
\draw plot[only marks, mark=*, mark size=.3pt] file{unifacqrys_1.dat};
\end{scope}
\end{tikzpicture}
%
%
%
%
\hfill
\begin{tikzpicture}[larger, semithick]
\fill[use as bounding box, \outside] \outline;
\begin{scope}
    \clip \outline;
    \foreach \i/\f in {1/\inside,2/\inside,3/\inside} {
    \draw[fill=\f,
        thick, line join=miter]
        plot file{unifac_2_\i.dat} -- cycle;
}
\draw plot[resp] file{unifacpts_2.dat};
\draw plot[only marks, mark=*, mark size=1pt] file{unifacfoci_2.dat};
\draw plot[only marks, mark=*, mark size=.3pt] file{unifacqrys_2.dat};
\end{scope}
\end{tikzpicture}
%
%
%
%
\hfill
\begin{tikzpicture}[larger, semithick]
\fill[use as bounding box, \outside] \outline;
\begin{scope}
    \clip \outline;
\foreach \i in {1} {
    \draw[fill=\inside,
        thick, line join=miter]
        plot file{unifac_3_\i.dat} -- cycle;
}
\draw plot[resp] file{unifacpts_3.dat};
\draw plot[only marks, mark=*, mark size=1pt] file{unifacfoci_3.dat};
\draw plot[only marks, mark=*, mark size=.3pt] file{unifacqrys_3.dat};
\end{scope}
\end{tikzpicture}
\caption{Optimal ambits with fixed facet counts. The first two have one facet,
while the last one has three, with queries (small dots) assigned to each by
$k$-means clustering}%
\label{fig:samplesinglefacet}%
\end{figure}

While the (primal) formulation of \programref{eq:optcoefflp} was motivated by
maximizing the expected filtering bound, the \emph{dual} problem has an
interesting interpretation in its own right.\footnote{With dual variables
$\varepsilon$ and $\beta_1,\dots,\beta_n\geq 0$, minimize $\varepsilon$ s.t.\
$-\varepsilon\leq\X\beta - \hat{z}\leq\varepsilon$ and $\onet\beta=1$.}
It is equivalent to a form of
regression problem, where the feasible region is the convex hull of the
feature vectors, and we are attempting to approximate~$\hat{z}$. That is, the
goal is to construct a convex combination $\X\beta$ that minimizes the
Chebyshev error $\|\X\beta - \hat{z}\|_\infty$. As the optimal error is equal
to the optimal lower bound, however, we would hope that a good fit is
impossible; a perfect fit would indicate a degenerate ambit.

\begin{proposition}
    \label{thm:poscoeff} If $\hat{z}_i > \max_j x_{ij}$, then $a_i \geq 0$ in
    every optimal solution.
\end{proposition}
\begin{proof}
We can restate our problem as maximizing
\begin{equation}
\label{eq:altobj}
\textstyle
\ell
\,\,
=
\,
\Bigl(\sum_ia_i\hat{z}_i - \max_j\sum_ia_ix_{ij}\Bigr)\,/\,\|a\|_1\,,
\end{equation} subject to $\|a\|_1\neq 0$. Assume that for an optimal
solution, we have $a_i<0$. We may then increase $a_i$ by some positive amount
$\varepsilon<|a_i|$, decreasing the denominator by~$\varepsilon$. Let~$k$ be
an index that attains the maximum in~\eqref{eq:altobj} after this increase.
Then $\max_j\sum_ia_ix_{ij}$ has increased by $\varepsilon x_{ik}$, or
possibly less if $k$ did not attain the maximum before this change. The
numerator, therefore, has increased by at least $\varepsilon\hat{z}_i -
\varepsilon x_{ik}$, which is positive because \mbox{$\hat{z}_i > x_{ik}$.}
The objective~$\ell$ will increase, and this contradicts optimality.
Consequently, $a_i\geq 0$.
\end{proof}

\noindent
In other words, as one might very well intuit, a negative focal weigth
indicates that the focus acts as a \emph{contrast}, a prototypical
\emph{non}-member; if the foci are more representative of the responsibilities
than they are of the queries, the focal weights will all be
non-negative.\endnote{A weight of zero eliminates the focus, so only positive
weights will actually be used.} If we assume or determine that this is the
case,~\programref{eq:optcoefflp} can be simplified by eliminating the column
associated with $v$, and letting $a=u$. If we do not have access to training
queries, and have no prior weighting of the foci, the problem reduces to
limiting the maximum remoteness, i.e., minimizing the radius.
We can safely assume (and easily check) that the foci are distinct, which
means \mbox{\X\ will} have no all-zero columns, so the radius is
non-zero.\endnote{A focus will usually not be a responsibility, so we could
even assume $\X>0$.} We can then find the optimal radius as follows.
\def\ha{\hat{a}}%
\begin{subequations}%
\def\wt{w}%
\def\wtt{\hat{w}}%
\begin{alignat}{4}
    \label{eq:minrad1}
    r_{\opt} \enskip
    &
    \mathrel{\phantom{=}\mathllap{\deq}}
    \enskip
    &&
    \min_{\mathclap{a,r\geq 0}}\enskip
    &&
    r
    \qquad
    &&
    \st
    \enskip
    a\one=1
    \,
    ,
    \enskip
    a\X\leq
    \mathrlap{r}\phantom{1} 
    \\
    \label{eq:minrad2}
    &=
    &&\mathop{\min_{a\geq 0}}\enskip
    &&
    1/
    a\one
    \qquad
    &&
    \st
    \enskip
    a\one\neq 0
    \,,
    \enskip
    a\X\leq 1\,.
\end{alignat}
The linear program in~\eqref{eq:minrad1} is a direct formulation of the
problem.
The reformulation in~\eqref{eq:minrad2} is justified by he
following proposition.

\begin{proposition}
    \label{thm:minradequiv}
    Optimization problems~\eqref{eq:minrad1} and~\eqref{eq:minrad2} are
    equivalent.
\end{proposition}
\begin{proof}
    As noted, the definition of \X\ implies $r>0$.
    Let $a,r$ and $\ha$ be solutions to~\eqref{eq:minrad1}
    and~\eqref{eq:minrad2}, respectively.
    If $a$ and $r$ are feasible,
    $\ha\deq a/r$ is feasible, with the same objective value. Conversely, if
    $\ha$ is feasible, $a\deq\ha/\ha\one$ and $r=1/\ha\one$ are feasible, with
    the same objective value. Hence the problems are equivalent, and yield the
    same minimum.
\end{proof}
\noindent
The second formulation is simply the inverse of the problem\footnote{The
definition of \X\ means $a=0$ could never be optimal, so the constraint
$a\one\neq 0$, which served to keep the objective in~\eqref{eq:minrad2}
well-defined, is now superfluous.}
\begin{equation}
    \label{eq:inverseminrad}
    1/r_\opt
    \enskip\deq\enskip
    \mathop{\smash{\max_{a\geq 0}}}\enskip\mathrlap{a\one}
    \phantom{1/a\one\qquad}
    \st\enskip a\X\leq 1\,,
\end{equation}
\end{subequations}
which is a positive linear program---a so-called \emph{packing} LP, a class of
optimization problems for which specialized algorithms
exist~\citep[e.g.,\kern-1.5pt][]{Koufogiannakis:2014}. One could even use
simple algorithms such as \emph{multiplicative weights update} to approximate
the optimum~\citep{Arora:2012}.

Between the two extremes of permitting an arbitrary number of facets
(\cref{thm:convexhull}) and limiting ourselves to a single facet
(\cref{thm:opt}), we could use the following heuristic procedure as a
compromise:

\medskip\noindent
\begin{pseudo}
cluster the $z$-vectors to get $k$ centroids \nl
\kw{for} $i=1$ \kw{to} $k$ \nl
\>$\hat{z}\gets\text{centroid $i$}$ \nl
\>compute $r_i$ and row $i$ of $\mA$ using \eqref{eq:optcoefflp}
or \eqref{eq:inverseminrad}
\end{pseudo}

\medskip\noindent
See the last panel of \cref{fig:samplesinglefacet} for an example ambit
constructed in this way. This approach reduces the choice of facet
construction heuristic to the choice of a suitable clustering algorithm, of
which there are many---even those that could cluster the queries
in the original metric space, with feature centroids computed
afterward~\citep{Berkhin:2006,Xu:2015}. Each facet is optimized to separate or
\emph{defend} the ambit from queries in one of the clusters. The number of
facets (i.e., clusters) may be fixed, or may be a result of the clustering.

\subsection{The Complexity of Focus Selection}
\label{sec:focusselection}

If we are constructing an ambit with $n$ responsibilities and $m$
pre-designated foci, using a set of $k$ sample queries, we can now efficiently
find optimal focal weights. However, if we aren't told which foci to use,
things get trickier.
If a sprawl is built incrementally, as many index structures are, the foci
will be determined ahead of the set of responsibilities, so finding an optimal
focal set is never an option. Still, the problem does merit some
investigation, as it would be relevant in rebuilding or bulk loading, for
example.
Let the \emph{metric linear focus selection problem}, or \textsc{focus
selection}, be defined as selecting $m$ foci from among \mbox{$m+n$} points in
a metric space, with $m\leq n$, leaving the remaining $n$ points as
responsibilities, so that $\max{}\set{\ell_\opt,0}$ is maximized, where
$\ell_\opt$ is defined by \programref{eq:optcoeff} and $\E[z]$ is
estimated as the average over a set of $k$ training queries.
When seen as a parameterized problem, we use $m$ as the parameter.%
\endnote{For more on parameterization, see, e.g., the book by
\citet{Cygan:2015}.} This problem is \emph{slice-wise polynomial} (XP), as an
optimal solution can be found by solving $\BigOh{n^m}$ instances
of the linear program~\programref{eq:optcoefflp}, each using one $m$-subset of
the \mbox{$n+m$} points as foci.
It is highly unlikely, however, that we will find a
\emph{fixed-parameter tractable} algorithm, i.e., with a running time
of~\mbox{$f(m)\cdot n^{\BigOh{1}}$} for some computable function $f$. This is true
even if if all we want is an \emph{approximation} of the optimum to within
some factor $\alpha(n)$, as the following \lcnamecref{thm:focalhardness}
shows.\endnote{For more on approximation, see, e.g., the book by
\citet{Williamson:2011}.}
\begin{theorem}
    \label{thm:focalhardness}
    Approximating
    \textsc{focus selection} is $\mathrm{W}[2]$-hard.
\end{theorem}
\begin{proof}
    By parameterized reduction from \textsc{dominating
    set}~\citep{Cygan:2015}.
    We are given a graph $G=(V,E)$, and wish to determine whether it has a
    dominating set with at most $m$ nodes.
    We define a metric space \tup{U,\delta}, with $U = V\cup\set{q}$ for some
    new object $q\notin V$, and
    \[
        \delta(u,v)
        \,
        =
        \,
        \begin{cases}
            \,\, 0 & \text{if $u=v\,;$} \\
            \,\, 2 & \text{if $\set{u,v}\in E\,;$} \\
            \,\, 3 & \text{if $\set{u,v}\in{} [V]^2\setminus E\,;$ and} \\
            \,\, 3+\varepsilon & \text{otherwise, i.e., if $q\in\set{u,v}$\,,}
        \end{cases}
    \]
    where $0<\varepsilon\leq 1$, and $[V]^2$ is the set of 2-subsets of $V$.
    This distance satisfies the metric axioms.
    Consider a set of~$m$ optimal foci from $\mkern.5muV$ with the training
    query $q$ in the metric space \seq{U,\delta}.
    We know that $\hat{z}_i = 3+\varepsilon > x_{ij}$ for all~$i$ and~$j$. By
    \cref{thm:poscoeff}, it follows that $a\geq 0$, and since $\|a\|_1=1$,
    $a\hat{z}$ is constant.
    This leaves us with the restricted radius minimization problem
    of~\eqref{eq:minrad1}.

    Recall that all nodes are, by definition, adjacent to at least one
    node in any dominating set.
    Assume that there is a dominating set $D\subseteq V$ of size at most $m$.
    If we use $D$ as our foci, each point will have at least one focal
    distance of~2, i.e., each column~$j$ of $\X$ contains an entry~$x_{ij}=2$.
    One feasible solution has $a_i=1/m$ for each $i$, yielding an $r_\opt$ of
    at most $3-1/m$, and thus an $\ell_\opt$ of at least $\varepsilon + 1/m$.
    If, however, $G$ does \emph{not} have a dominating set of size at most
    $m$, then for any set $D\subseteq G$ of $m$ foci, there will always be at
    least one point whose focal distances are all~3, as it is not adjacent to
    any of the corresponding nodes. In this case, $r_\opt=3$, and
    $\ell_\opt=\varepsilon$. To discern between the two values, we need
    $\alpha(n)\cdot(\varepsilon + 1/m)$ to be strictly greater than
    $\varepsilon$. For the inexact case, this means making $\varepsilon$
    strictly less than
    \[
        \frac{\alpha(n)}{m\cdot(1-\alpha(n))}\eqcomma
    \]
    ensuring that any $\alpha(n)$-approximation will be greater
    than~$\varepsilon$ if and only if $G$ has a dominating set of at most $m$
    nodes, completing the reduction.
\end{proof}

\noindent
As previously discussed, the choice to maximize the lower bound is somewhat
heuristic, and altering the problem slightly could make it more tractable.
For example, let \textsc{simple focus selection} be the problem of finding foci
minimizing $r_\opt$, as defined by~\eqref{eq:minrad1}, still with $\delta$
being a metric. The reduction in \cref{thm:focalhardness} applies, so this
simplified problem is also~$\mathrm{W}[2]$-hard; however, given that $a\geq
0$, the radius turns out to be an easier objective to approximate, as the
following \lcnamecref{thm:radapprox} shows. Recall from
\cref{def:cmpfeatures} that for a given tuple $p$ of foci,
$\phi(u)=[\delta(u,p)]_{i=1}^m$.
\begin{observation}
    \label{thm:upperbound}
    If $\|a_1\|=1$ and $a\geq 0$ then $\delta(u,v)
    \leq
    a\phi(u) + a\phi(v)
    $.
    \qed
\end{observation}

\begin{theorem}
    \label{thm:radapprox}
    \textsc{simple focus selection} has an \BigOh{mn} 2-approximation.
\end{theorem}
\begin{proof}
    For any input point $u$, its \emph{$k$-radius}, the minimum radius
    of a ball containing $k$ other points $u$,
    can be found in linear time.\endnote{For example, using
    \citeauthor{Chazelle:2000}'s soft heaps~\citep{Chazelle:2000}, or the
    classic algorithm of \citet{Blum:1973}.}
    Let $p^*$ be the optimal focal tuple, whose minimal radius is
    $r\mathrlap{^*}_\opt$,
    and let $u$ and $v$ be any non-focal
    points. By the definition of a radius, we have $a\phi(u),a\phi(v)\leq
    r\mathrlap{^*}_\opt$ and from \cref{thm:upperbound} it follows that
    $\delta(u,v)\leq2r\mathrlap{^*}_\opt$.
    Also, because $a\phi(u)\leq r\mathrlap{^*}_\opt$, $a\geq 0$ and
    $\|a\|_1=1$, there must be at least one $p^*_i$ for which
    $\delta(u,p^*_i)\leq r\mathrlap{^*}_\opt$.
    Since each of the $n$ nonfocal points would be within a distance of
    $r\mathrlap{^*}_\opt$ of at least one focus, and within a distance of
    $2r\mathrlap{^*}_\opt$ of the other $n-1$ nonfocal points, it follows that
    there must be at least $n$ points with $n$-radius at
    most~$2r\mathrlap{^*}_\opt$. To select one such point $p_1$ as the first
    focus, among all the $n+m$ points, simply use a point whose $n$-radius is
    no greater than that of at least~$m$ others, for a running time of
    \BigOh{mn}. All but $m-1$ points fall within $2r\mathrlap{^*}_\opt$ of
    $p_1$, and the furthest $m-1$ points become $p_2,\dots,p_m$. Even for
    $a_1=1$, we have a 2-approximation, so $r_\opt\leq 2r\mathrlap{^*}_\opt$.
\end{proof}

\noindent
Though a 2-approximation is a positive result, it seems likely that a radius
doubling would drastically increase the overlap probability. For many
approximation algorithms, one does of course see better results in
practice, but as a heuristic, the algorithm of \cref{thm:radapprox} is
somewhat unsatisfying: It approximates the optimal $m$-ambit by a \emph{ball},
and then picks up the \mbox{$m-1$} potential stragglers that fall outside it,
giving them a free pass by promoting them to foci. There is no expectation
that they will actually perform any work in this capacity.
One natural alternative, which works for full \textsc{focal selection} and
which seems to yield acceptable results in practice, is the following
two-round or \emph{runoff} heuristic:

\medskip\noindent
\begin{pseudo}
use a large candidate set of foci, and solve \eqref{eq:optcoefflp} \nl
use only the $m$ candidates $p_i$ with highest $|a_i|$ as foci, and re-solve
\eqref{eq:optcoefflp}
\end{pseudo}

\medskip\noindent
For example, one could in the first round let all the responsibilities, or a
random sample, act as candidate foci, while at the same time acting as
responsibilities. In the second round, the selected foci are no longer
responsibilities, and so do not constrain the radius.
The value of $m$ could either be fixed, or determined by a threshold for
$|a_i|$. Alternatively, one could have both such a threshold \emph{and} a cap
on $m$, and run multiple rounds using the threshold, until the number of foci
falls below~$m$.

For the case where the set of candidate foci is separate from the
responsibilities, the rationale is that rounding down low-magnitude
coefficients to zero will have limited impact on the objective.
Selecting foci from among the responsibilities could be more risky, as the
constraints on the radius will change between the rounds. The number of
changes will typically be low, as $m\ll n$, but even a single dropped
responsibility could change the nature of the problem, leading to another set
of foci outperforming those originally selected.
Even so, if the objective value of the first round is good, there may be
reason to believe it will be in the second round as well: The original
rationale still applies to the objective function itself, and dropping
constraints can only improve the optimum.

\subsection{Non-Linear Metric-Preserving Remoteness}
\label{sec:nonl}

As hinted at in the introduction to \cref{sec:ambitregiontype}, we now leave
the assumption of linearity behind, while remaining in the realm of metric
spaces. In other words, for $m=1$, the structure-preserving map
$f:\R_+\to\R_+$ will now be an arbitrary \emph{metric-preserving function}, or
\emph{metric transform} which takes any metric $\delta$ to another metric
$f\circ\mkern1mu\delta$. The image of $\A$ in $\B$ will be a metric space,
even though $\B$ need not be. \Citeauthor{Deza:2013} provide an extensive list
of such functions~\citep[ch.\,4]{Deza:2013}, while
\citeauthor{Dobos:1995}~\citep{Dobos:1995,Dobos:1998} and \citet{Corazza:1999}
deal with the topic in more depth. The Cantor function (see \cref{fig:cantor})
is a rather forbidding example, but there are many quite straightforward
metric-preserving functions such as $x/(1+x)$ or~$x^\alpha$, for
$\alpha\in[0,1]$.

Metric transforms of multiple parameters also exist, giving us the desired
remoteness $f:\R_+^m\to\R_+$ for $m>1$. One could use any isotone
norm, for example~\citep{Deza:2013,Dobos:1998,Borsik:1981}. Vector-valued
transforms can, in turn, be viewed as ensembles of scalar-valued ones, as in
$f(x) \deq \tran{[f_1(x)\,\dots\,f_d(x)]}$, making
$f\circ \delta^m$ a vector-valued metric, which lets us use multiple radii as
in the linear case.

Metric-preserving functions must start at the origin,\footnote{Strictly
speaking, they must have $f^{-1}(0) = \set{0}$.} they must be subadditive
(e.g., concave), and any such function that is non-decreasing is, in fact, a
metric-preserving function~\citep[p.\,9]{Dobos:1998}.
A metric-preserving function need not be monotone; to preserve relevance,
however, our map must be. If, rather than non-increasing subadditive,
$f$ is non-\emph{increasing} \emph{super}-additive, our target inequality
could be
\[
    f(x) - f(y) \geq f(z)
    \qquad\text{instead of}\qquad
    f(x) + f(y) \geq f(z)\mathrlap{\,,}
\]
where $x_i=\delta(u,p_i)$, $y_i=\delta(u,q)$ and $z_i=\delta(p_i,q)$. This
follows from the fact that $x \leq y + z$. The general overlap check
becomes
\[
    r \pm f(s,\dots,s) \geq f(z)\,,
\]
using $+f(\blank)$ to transform $s$ for the non-decreasing case, and
$-f(\blank)$ for the non-increasing case.

There are many ways to construct monotone metric-preserving functions. We have
already used the fact that the negation of a non-increasing superadditive
function is non-decreasing subadditive. If we are faced with a non-decreasing
superadditive or non-increasing subadditive function, we could negate their
parameters as well; if $f(x)$ is non-decreasing superadditive, for example,
$g(x)=-f(-x)$ is be non-decreasing subadditive.
Indeed, we can compose non-increasing or non-decreasing sub- or
super-additive function in various ways, with predictable
results.\endnote{See, e.g., Theorem~7.2.1 of \citet{Kuczma:2009}. He discusses
concave and convex functions in particular, but the results generalize.}

A practical approach to non-linear ambits might be to construct
single-parameter scalar transforms for each focal distance, and then combining
those with a single multiparameter function.
One particularly manageable version of this would be a linear combination of
some transform that is shared by all focal distances. For example, with the
well-known metric-preserving power transform~\citep[][p.\,81]{Deza:2013}, the
remoteness map becomes
\begin{equation}
    \label{eq:rbfn}
    f(x) = \sum_{i=1}^m a_ix\mathrlap{_i}^\alpha\,,
\end{equation}
with the overlap check $r + \|a\|_1s^\alpha \geq f(z)\,$, where
$\alpha\in[\mkern1.25mu 0,1]$.
We could easily have a separate $\alpha_i$ or even some
other transform for each $x_i$ (applying them individually to $s$).
Perhaps less obviously, we could have ambits in the form of
\citeauthor{Blinn:1982}-style metaballs~\citep{Blinn:1982}, using
\citeauthor{Blinn:1982}'s original density function (inverted, to get
remoteness, $a, b > 0$):
\begin{equation}
    \label{eq:blinn}
    f(x) = \sum_{i=1}^m \bigl( 1-b_i e^{-a_ix_i} \bigr)
\end{equation}
This is the kind of ambit seen in \cref{fig:toc}.

\begin{observation}
    \label{obs:blinn}
    The remoteness map in \eqref{eq:blinn} is metric-preserving.
\end{observation}
\begin{proof}
    It is easily verified that $f^{-1}(0)=\set{0}$, so a sufficient condition
    would be for $f$ to be non-decreasing subadditive, which in turn may be
    ensured by its summands having this property. The property obviously holds
    for the constant summands~1. The exponential $e^x$ is non-decreasing
    superadditive (indeed, convex) and so $-e^{-x}$ is non-decreasing
    subadditive (concave), as is $-be^{-ax}$ (for $a, b > 0$).
\end{proof}
If we use a linear combination such as the one in \eqref{eq:rbfn} or
\eqref{eq:blinn} as a tool for function approximation, i.e., adapting the
region shape to the layout of its responsibilities, it can be interpreted as a
\emph{radial basis function network}~\citep{Lowe:1988}, generalized from
euclidean to arbitrary metric spaces, but restricted to non-decreasing
subadditive radial basis functions. Alternatively, if we view ambits as
generalized metaballs, the transform would correspond to the metaball
\emph{bump function}; different parameter values would vary what
\citeauthor{Blinn:1982} calls \emph{blobbiness} (see \cref{fig:blobbiness} for
an illustration).
\begin{figure}
\centering
\def\inside{lightgray}
\def\outside{lightshade}
\def\xyscale{1.7}%
\def\D{.5}%
\def\outline{(-\D - .1, -\D - .2) rectangle +(3*\D + .2, 2*\D + .4)}
\footnotesize
\begin{tikzpicture}[larger, semithick, xscale=\xyscale, yscale=\xyscale]
\fill[use as bounding box, \outside] \outline;
\begin{scope}

    \clip \outline;

    \draw[thick, fill=darkshade]
        plot file{powerlipse1.dat} -- cycle;
    \draw plot[only marks, mark=*, mark size=0.5882352941176471pt]
        file{pwlfoci.dat};

\draw (2*\D + .05,.55) node[left=0] {$\alpha=1$};
\draw (-\D - .05,.548) node[right=0] (lab) {$\sum_i x_i^\alpha\leq r$};
\draw[thin] (.07,0.05) to[out=120, in=-90] (lab);

\end{scope}
\end{tikzpicture}
\hfill
\begin{tikzpicture}[larger, semithick, xscale=\xyscale, yscale=\xyscale]
\fill[use as bounding box, \outside] \outline;
\begin{scope}
    \clip \outline;
    \draw[thick, fill=darkshade]
        plot file{powerlipse2.dat} -- cycle;
    \draw plot[only marks, mark=*, mark size=0.5882352941176471pt]
        file{pwlfoci.dat};
\draw (2*\D + .05,.55) node[left=0] (xxr) {$\alpha=0.5$};

\end{scope}
\end{tikzpicture}
\hfill
\begin{tikzpicture}[larger, semithick, xscale=\xyscale, yscale=\xyscale]
\fill[use as bounding box, \outside] \outline;
\begin{scope}

    \clip \outline;
    \draw[thick, fill=darkshade]
        plot file{powerlipse3.dat} -- cycle;
    \draw plot[only marks, mark=*, mark size=0.5882352941176471pt]
        file{pwlfoci.dat};

\draw (2*\D + .05,.55) node[left=0] (xxr) {$\alpha=10^{-8}$};

\end{scope}
\end{tikzpicture}
\caption{Ambits of varying blobbiness, using the power transform as bump
function}%
\label{fig:blobbiness}%
\end{figure}
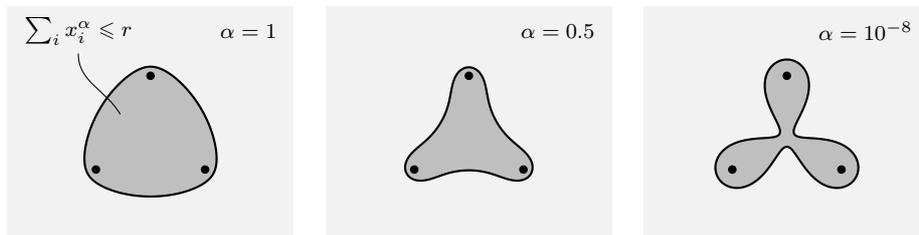
The blobbiness could be adjusted to maximize the lower bound---for example,
one could use some heuristic optimization procedure to set $\alpha$, where the
objective function is calculated in each iteration, as follows:

\medskip\noindent
\begin{pseudo}
for each focal distance $x_i$, calculate $x\mathrlap{_i}^\alpha$ \nl
using the transformed focal distances, calculate $\ell'_\opt$ using
\eqref{eq:optcoefflp} \nl
\kw{return} $(\ell'_\opt)^{1/\alpha}$
\end{pseudo}

\medskip\noindent
The reason for the inverse transformation in the last line is that the lower
bound computed by our linear program is for $\delta(q,u)^\alpha$, and in order
for the objective values to be comparable for different values of $\alpha$, we instead
want the lower bound for $\delta(q,u)$. There is a link here to the
optimization of power tranforms used by \citeauthor{Skopal:2007} for
approximate search using metrics or almost metric
distances~\citep{Skopal:2007}. Other transforms could be used in the same
manner, as long as they are invertible. Preliminary experiments indicate that
such an approach may yield improvements over a purely linear remoteness.

\subsection{Repurposing Region Definitions More Generally}
\label{sec:preserve}

Consider a query and a region respectively defined as
\begin{equation}
    Q = \set{x : \varphi(x)}
    \quad\text{and}\quad
    R = \set{x :\psi(x)}\eqcomma
\end{equation}
where $\varphi(x)$ and $\psi(x)$ are logical formulas with their non-logical
symbols drawn from some signature $\Sigma$.\footnote{See
\cref{sec:modeltheory} for relevant terminology.} We wish to know whether $Q$
and $R$ intersect, and even though we might not be able to determine this
perfectly, for correctness we must \emph{detect} it (cf.~\cref{rem:larger}),
i.e., we have some \emph{necessary condition} $\beta$, so that
\begin{equation}
    \forall x
    \bigl(
    \varphi(x)\to\psi(x)\to\beta
    \bigr)\eqdot
    \label{eq:overlapdetector}
\end{equation}
In other words, if $Q$ and $R$ intersect then $\beta$ must be true.

We would like to repurpose the region definition $\psi$ to extend our supply
of regions, while still being able to use our overlap detector $\beta$.
We do this by \emph{reinterpreting} $\psi$ in a manner that preserves the
implication expressed by \eqref{eq:overlapdetector}.
Let $\A$ be the $\Sigma$-structure underlying our formulas, i.e., our actual
domain of discourse.
Our region $R$ is then the set $\psi(\A)$ of objects in $\dom\A$ that satisfy
$\psi$. Now consider a map $h:\A\to\B$ to some \emph{other} $\Sigma$-structure
$\B$. Because $\B$ has the same signature as $\A$, we can reuse our definition
to produce a new set $\psi(\B)$. The \emph{preimage}
$h^{-1}\bigl[\psi(\B)\bigr]$ will be a new region in $\A$, so if we have
several maps such as $h$, we end up with several new regions, all described by
the original $\psi$.

Even a single map $h$ can yield multiple new regions, however. Chances are,
the region definition is \emph{parameterized}, and $\B$ may offer new objects
we can use as parameters. For example, we may start out with $\psi(\A,a)$, for
some $a$ in $\dom\A$. Then $h$ yields one new region
$h^{-1}\bigl[\psi(\B,b)\bigr]$ for \emph{every applicable $b$} in $\dom\B$.
Given that $\dom\B$ might be much larger than $\dom\A$, this could increase
our options considerably. This might be true even if our parameter originates
in $\A$, as it could be specified by \emph{multiple} $\A$-objects
$a_1,\dots,a_m$, producing regions of the form
$h^{-1}\bigl[\psi(\B,f(a_1,\dots,a_m))\bigr]$, for some map
$f:\A^\hilc{m}\?\to\B$.

As long as we use appropriate \emph{structure-preserving} maps, a
reinterpreted version of $\beta$, presumably with the same parameters as
$\psi$, will still work in $\B$.
The following \lcnamecref{thm:ambits:general} describes some general
conditions under which this is the case. For simplicity, I assume a single
parameter originating in each of $\A$ and $\B$, though the result generalizes
to multiple parameters in the obvious way.\footnote{Recall that
$\Delta_m(a)=\seq{a,\dots,a}$.} (See \cref{ex:metricpreservation} for an
application of the theorem.)

\begin{theorem}\label{thm:ambits:general}
Let $\A$ and $\B$ be structures with signature $\Sigma$, which contains some
unary relation symbol $S$, and let $\varphi$, $\psi$, $\alpha$ and $\beta$ be
$\Sigma$-formulas.
Let
\begin{equation}
    Q=\varphi(\A)
        \quad\text{and}\quad
    R=h^{-1}\bigl[
        \psi(\B, f(a_1,\dots,a_m), b)
    \bigr]\eqcomma
    \label{eq:regionsfromformulas}
\end{equation}
where $h=f\circ\Delta_m$, for some homomorphism
$f:\A^\hilc{m}\?\to\B$, and $a_1,\dots,a_m\in S^\A$.
Suppose that
$T_\A$ and $T_\B$ are theories true in $\A$ and $\B$,
respectively, where
\begin{axioms}{A}
\item
    $T_\A$ is Horn,
    $\varphi$ existential, and $\alpha$ positive existential;
\item
    $T_\A$ entails
    $\forall xy\,\bigl(
        S(y)\to
        \varphi(x)\rightarrow\alpha(x,y)
        \bigr)
        $;
    \label{ax:entaila}
\item
    $\spc{T_\A}{T_\B}$ entails
    $\forall xyz\,\bigl(
        S(y) \to
        \alpha(x,y)
        \to
        \psi(x,y,z)\rightarrow\beta(y,z)
        \bigr)
        $.
    \label{ax:entailb}
\end{axioms}
Then $\beta$ is true of $f(a_1,\dots,a_m)$ and $b$ in $\B$ whenever $Q$ and
$R$ intersect.
\end{theorem}
\begin{proof}
    Assume that there is some element $a$ that lies in both $Q$ and $R$.
    Let $\bar a=\tup{a_1,\dots,a_m}$. Because $S$ is true of each of
    $a_1,\dots,a_m$ in $\A$, $S(\bar a)$ is true in $\A^\hilc{m}$ and, because
    $f$ is a homomorphism, $S(f(\bar a))$ is true in $\B$.
    Because~$T_\A$ is Horn, it is preserved in products~\citep[Cor.\
    9.1.6(a)]{Hodges:1993},
    and is consequently true in $\A^\hilc{m}$; $\varphi$ is existential, and
    therefore preserved by embeddings~%
    \citep[Thm.\ 2.4.1]{Hodges:1993}%
    , including the diagonal
    $\Delta_m$ (cf.\ \cref{ex:morphisms}\ref{it:diagonal}).
    By assumption, $a\in Q$, so by definition $\varphi(a)$ is true in~$\A$,
    and therefore $\varphi(\Delta_m(a))$ is true in~$\A^\hilc{m}$.
    From \cref{ax:entaila} it follows that $\alpha(\Delta_m(a),\bar a)$ is
    true in~$\A^\hilc{m}$.
    Because $\alpha$ is positive existential, it is preserved by the
    homomorphism
    $f$~\citep[Thm.\ 2.4.3]{Hodges:1993}, and thus $\alpha(h(a),f(\bar a))$ is
    true in~$\B$.
    Because $a\in R$, $\psi(h(a),f(\bar a),b)$ is true in $\B$ \emph{by
    definition}.
    From \cref{ax:entailb}, we conclude that $\beta(f(\bar{a}),b)$ is true
    in~$\B$.
\end{proof}

\begin{remarks}
    \begin{paras}
    \item
    The role of $S$ is to restrict the quantification of $y$, for cases where
    there are effectively different sorts of objects (cf.\
    \cref{rem:structvariations}\ref{rem:it:multisorted}). Presumably,
    $\varphi$ and $\psi$ have whatever such restrictions they need, but if
    $\alpha$ is to be positive, it cannot contain an implication, and so this
    predicate is extracted.
    \item
    Values or parameters that figure in the formulas, beyond~$x$, $y$ and
    $z$, such as any used in $\varphi$, may be handled by adding constants to
    $\Sigma$~\citep[cf.][\null 1.4]{Hodges:1993}.
    \item
    The details may be varied using other preservation results from model
    theory. For example, if $f$ is surjective, $\alpha$ need only be
    positive~\citep[Thm.\ 2.4.3]{Hodges:1993}.
    If $f$ is a \emph{strong} homomorphism, then $\alpha$ could be any
    quantifier-free formula without equality; if $f$ is also injective (an
    embedding) or surjective, respectively, the restrictions on equality and
    quantification could be removed~\citep[p.\,96]{Enderton:2001}. An
    isomorphism (a bijective strong homomorphism) would preserve \emph{all}
    formulas; in that case, $R$ would simply be an intersection, isomorphic to
    $\bigcap_{i=1}^m\psi\bigl(\A,a_i,h^{-1}(b)\bigr)$.
    \end{paras}
\end{remarks}

\begin{example}
    \label{ex:metricpreservation}
    Consider metric or quasimetric spaces in light of
    \cref{thm:ambits:general} A signature for these might be
    $\Sigma\deq\set{\delta,+,\leq,U,K}$, where $\delta$ is a binary function
    symbol representing the distance function, $+$ and $\leq$ represent
    addition and ordering of the distances, and $U$ and $K$ are unary relation
    symbols used for restricting the values of variables to points and
    distances (comparison values), respectively.
    The triangle inequality, for example, could be expressed as follows:
    \[
        \forall xyz\,\bigl(
        U(x)\land U(y)\land U(z)
        \rightarrow
        \delta(x,z)\leq\delta(x,y)+\delta(y,z)
        \bigr)
        \eqcomma
    \]
    In the same manner we could, if we wished, represent all the metric
    axioms, though probably without fully axiomatizing the real numbers.
    (Indeed, we might not even be using real-valued distances; cf.\
    \cref{ex:cmpfuncs}\ref{ex:cmpfuncs:generalized}.)
    With some care, the result, including the axioms for ordering and
    addition, may be formulated as a Horn theory,\footnote{E.g., turning
    $U(x)\land U(y)\to(\delta(x,y)=0\to x=y)$ into $U(x)\land
    U(y)\land\delta(x,y)=0\to x=y$.} so it could play the role of $T_\A$.

    Now consider two metric balls $R$ and $Q$ with centers $p$ and $q$ and
    radii $r$ an $s$, respectively. If these intersect, we have
    $\delta(p,q)\leq r + s$. We can dissect and repurpose this overlap check
    using \cref{thm:ambits:general}, as follows.
    Let $\varphi$, $\psi$, $\alpha$, and $\beta$ be defined as:
    \begin{IEEEeqnarray*}{rCl+rCl}
        \varphi(x) &=& U(x) \land \delta(x,q) \leq s
        &
        \alpha(x,y) &=& \delta(y,q) \leq \delta(y,x) + s\\
        \psi(x,y,z) &=& U(x) \land K(z) \land \delta(y,x) \leq z
        &
        \beta(y,z) &=& \delta(y,q) \leq z + s
    \end{IEEEeqnarray*}
    The first formula, $\varphi(x)$, represents a query ball $Q$ with center
    $q$ and radius $s$, where both $q$ and $s$ are treated as constants, drawn
    from $\Sigma$. It is easy to verify that the second formula,
    $\alpha(x,y)$, is entailed by $T_\A$ (i.e., the metric axioms) and
    $\varphi(x)$, for any point $y$, i.e., any object for which $U(y)$ is
    true.
    If we use $U$ instead of $S$ in \cref{thm:ambits:general}, we
    need not make this last caveat an explicit part of~$\varphi$.
    The region $R$ is described by $\psi(x,y,z)$, where $x$ is the potential
    member, $y$ is the center point and $z$ is the radius. We explicitly
    constrain $x$ and $z$, but $y$ is already constrained in
    \cref{ax:entailb}. Finally, combining $\alpha(x,y)$, $\psi(x,y,z)$ and
    some basic axioms $T_\B$ describing ordering and addition of
    comparisons (e.g., those of an ordered monoid), we can derive
    $\beta(y,z)$, as required.

    We can recreate the original overlap check between two balls by setting
    $m$ to~$1$ and letting $f$ be the identity on $\A$. If, however, $m\geq 2$
    and $\A\neq\B$, our region $R$ will have \emph{multiple} points $p_i$ in
    $U^\A$, collectively corresponding to its center, and the radius $r$ will
    be drawn from $K^\B$.
    The overlap check then becomes
    \begin{equation}
        \delta^\B(f(p_1,\dots,p_m),f(q,\dots,q))\leq r + f(s,\dots,s)\eqdot
    \end{equation}
    Because $f$ is a homomorphism and because
    $\delta$ acts elementwise in $\A^m$ (see
    \cref{def:cartesianproduct,def:homomorphism}), we can rewrite this as
    \begin{equation}
        f(\delta^\A(p_1,q),\dots,\delta^\A(p_m,q))\leq r + f(s,\dots,s)\eqdot
    \end{equation}
    As an example, take the \emph{egglipse}, a trifocal relative of the
    ellipse (and, indeed, an example of a linear ambit).
    To repurpose our overlap check, we set $m$ to~$3$ and let $f$ sum the
    distances,\endnote{For non-distance arguments, $f$ may simply act as the
    identity.}
    yielding
    \begin{equation}
        \delta^\A(p_1,q)+\delta^\A(p_2,q)+\delta^\A(p_3,q)\leq r + 3s\eqdot
    \end{equation}
    as illustrated in \cref{fig:ambitexample}. Note that we are actually
    defining a ball in the new metric space $\B$, whose points are triples of
    $\A$-points and whose distance is the sum of elementwise $\A$-distances.
\end{example}

\begin{figure}
\input{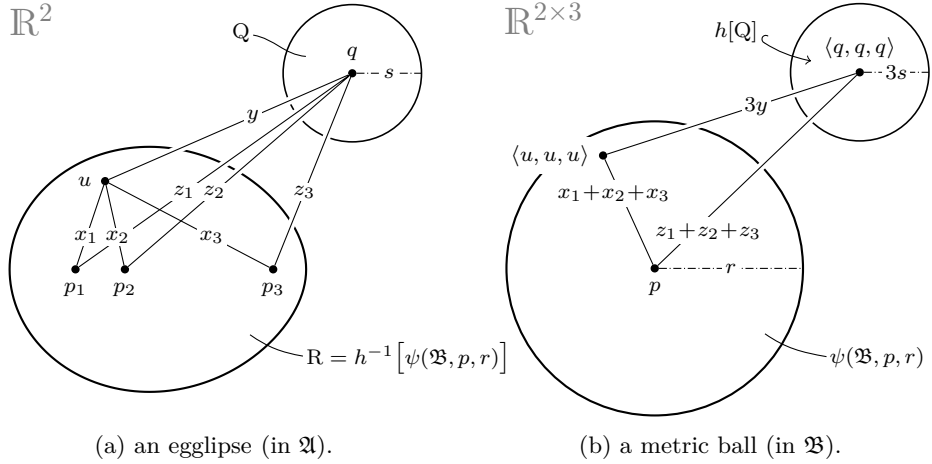}
\def\xh{0.33} 
\subcaptionbox{an egglipse (in \A)}{%
\begin{tikzpicture}[larger, semithick]
\footnotesize

    \path (0,-\rsc);

    \draw[thick] plot file{egglipse.dat} -- cycle;

    \fill
        \fa node[pointnode] (F1) {}
        \fb node[pointnode]   (F2) {}
        \fc node[pointnode] (F3) {}
        ;

    \path
        \u node[pointnode] (U) {}
            node[left=2pt] {$u$}
        \q node[pointnode] (Q) {} node[above=2pt] {$q$}
        ;

    \draw (Q) circle (\s);
    \draw[thin, densely dashdotted] (Q) -- +(0:\s)
        node[midway, fill=white, inner sep=1pt] {$s$};

    \draw[line width=3pt, white, shorten <=5pt, shorten >=5pt]
        (Q) -- (U)
        (Q) -- (F1)
        (Q) -- (F2)
        (Q) -- (F3)
        ;
    \draw[thin]
        (Q) -- (U)
        node[pos=.40, fill=white, inner sep=1pt] {$y$}
        \foreach \i in {1,2,3} {
            (Q) -- (F\i)
            node[pos=.61, fill=white, inner sep=1pt] {$z_\i$}
        }
        ;

    \path[name path=horiz] (-1,\xh) -- (1.5,\xh);
    \foreach \num in {1, 2, 3} {
        \draw[thin, name path=x\num]
            (F\num) -- (U)
            ;
        \draw[name intersections={of=horiz and x\num}]
            (intersection-1) node[inner sep=1pt,fill=white] {$x_\num$}
            (F\num) node[below=2pt] {$p_\num$}
            ;
    }

    \draw
        (Q) +(160:1.2) node[inner sep=1pt] (QQ) {$Q$}
        (F3) +(-65:1) node[inner sep=1pt] (RR)
        {$R\mathrlap{\null=h^{-1}\bigl[\psi(\B,p,r)\bigr]}$}
    ;
    \draw[thin]
        (QQ) edge[out=-10,in=175] +(-20:.7)
        (RR) edge[out=170,in=-35] +(160:.7)
    ;

    \draw (current bounding box.north west)
        node[gray,
            anchor=north west, inner sep=0, font=\Large] {$\R^2$}
    ;

\end{tikzpicture}}
\hfill
\subcaptionbox{a metric ball
    (in \B)}{\label{fig:tupleball}
\begin{tikzpicture}[larger, semithick]
    \footnotesize

    \draw
        (0,0) node[pointnode] (C) {} node[below=2pt] {$p$}
        ;

    \path[name path=horiz] (-.5,\xh) -- (0.5,\xh);

    \draw[thick] (C) circle (\rsc)
        ;

    \path

        \uu node[left=2pt, fill=white]
            (ulab) {$\tup{u,u,u}$}

        \uu node[pointnode] (U) {}

        \qq node[pointnode] (Q) {}
            node[above=2pt] {$\tup{q,q,q}$}

        ;

    \draw (Q) circle (\ssc);
    \draw[thin, densely dashdotted] (Q) -- +(0:\ssc)
        node[midway, fill=white, inner sep=1pt] {$3s$};

    \draw[thin,densely dashdotted]
        (C) -- +(0:\rsc) node[midway,fill=white,inner sep=1pt] {$r$};

    \draw[thin]
        (C) -- (U)
        ;

    \draw[line width=3pt,draw=white,shorten <=10pt, shorten >=10pt]
        (Q) -- (U)
        (Q) -- (C)
        ;
    \draw[thin] (Q) -- (U)
        node[pos=.40, fill=white, inner sep=1pt] {$3y$}
        ;
    \draw[thin] (Q) -- (C)
        ;
    \coordinate (Cup) at (C |- U);
    \draw
        ($(C)!.3333333!(Cup)$)
        node[right=-2pt, fill=white, inner sep=2pt] {$z_1\!+\!z_2\!+\!z_3$}
        ($(C)!.6666666!(Cup)$)
        node[left=-8pt, inner sep=2pt, fill=white] {$x_1\!+\!x_2\!+\!x_3$}
        ;

    \draw
        (Q) +(160:1.2) node[inner sep=1pt] (QQ)
            {$\spcl{Q}{h[Q]}\mkern-1mu$}
        (C) ++(-65:1) ++(1.35,0) node[overlay, inner sep=1pt] (SS)
            {$\mathrlap{\psi(\B,p,r)}$}
    ;
    \draw[thin]
        (QQ)
            edge[
                bend right=10,
                right hook-{Classical TikZ Rightarrow}
                ]
            +(-25:.7)
        (SS) edge[out=170,in=-35] +(160:.7)
    ;

    \draw (current bounding box.north west)
        node[gray,
        anchor=north west, inner sep=0, font=\Large] {$\R^{2\times 3}$}
    ;

\end{tikzpicture}}%
\caption{%
The ambit construction, as it applies to an egglipse~\citep{Sahadevana:1987}.
The foci form a central tuple $p\deq\tup{p_1,p_2,p_3\mkern-1mu}$,
while any other point $u$ is mapped to $\tup{u,u,u}$. The sum of elementwise
distances is the metric in \B, where the egglipse interior becomes a ball.
More generally, an $m$-focal polyellipse could only intersect $Q$
if $z_1+\cdots+z_m \leq r + ms$}\label{fig:ambitexample}%
\end{figure}

\begin{remarks}
\begin{paras}
\item
For the ordered vector spaces of
\cref{sec:linearmetric,sec:focusselection,sec:emulatingregions,sec:optcoeff},
homomorphisms are
non-decreasing linear maps, which cannot have negative coefficients.
We got around this limitation by exploiting the properties of linearity,
constructing an overlap check that permitted some components of the ordering
to be reversed. The bottom line, however, is that, except in the non-negative
case of~\eqref{eq:minrad1},
our remoteness was \emph{not} a homomorphism, so it would seem that there is
use in casting our net more widely than in \Cref{thm:ambits:general}.
Take, for example, the Hamacher product, as used
in~\cref{fig:hamacher}.\endnote{%
    $
    f(x_1,x_2) \deq\mkern1.5mu
    x_1x_2/(x_1+x_2-x_1x_2)
    $
with $x_i\in[\mkern1.25mu 0,1]$ and $f(0,0)\deq0$.}
It is not structure-preserving in any obviously useful way, but may still be
used to define a metric ambit. Consider the situation in the feature space:
The metric axioms let us circumscribe a query $Q$ with a hypercube, and the
image of the ambit is concave (see inset of \cref{fig:hamacher}); the only way
there can be overlap is if a query corner is inside the region image, and
because~$f$ is non-decreasing, it suffices to examine the corner $z-s$, so
$f(z-s)\leq r$ is a valid overlap check (cf.~\cref{fig:featurespace}). Similar
lines of reasoning about the geometry of the feature space could lead to
overlap checks for other axioms or other maps, whether they are homomorphisms
or not.

\item
There are ways of generalizing explicit structure preservation as well, by
relaxing how our map preserves functions. For example, if we for each source
symbol $F$ have some other target symbol $G$, and vice versa, we could relate
$F^\A$ to $G^\B$ rather than to $F^\B\?$, as in so-called \emph{weak}
homomorphisms~\citep{Schneider:2011}. For our purposes, this
could be achieved for homomorphisms as well, with a simple renaming of the
functions in $\B$, such as using a different function for comparison. We
may also weaken the \emph{equality} requirement of \cref{def:homomorphism},
yielding
\end{paras}
\end{remarks}
\begin{axioms}{H}
\item[{$[H\smash{\rlap{\raisebox{.1ex}{'}}}_{2}]$}]
$f(F^\A(a_1,\dots, a_n))$ $\spcc{\text{then}}{{}\approx_\B}$
$F^\B(f(a_1),\dots,f(a_n))$
\end{axioms}
for some relation symbol $\approx$.\endnote{We might also use a relation not
named in our signature, as long as it's defined on $\B$.} This a form of
abstraction found in so-called \emph{quasi}-homomorphisms. Instead of
equality, they require \emph{approximate} equality, where the difference
between the two values (as measured by some metric on $\B$) is
bounded~\citep{Fujiwara:2016}.
However, even if our only requirement is that $\approx_\B$ be
\emph{transitive}, we will preserve formulas of the form $x_0\approx
\delta(x_1,\dots,x_n)$.\endnote{%
\def\fa{f\mkern-1mu a}
We have $a_0\approx_\A F^\A(a)
\Rightarrow
f(a_0)\approx_\B f(F^\A(a)) \approx_\B F^\B(\fa)
\Rightarrow
f(a_0)\approx_\B F^\B(\fa)$,
where $a$ is the tuple \seq{a_1,\dots,a_n} and $\fa$ is the standard shorthand
for $\seq{f(a_1),\dots,f(a_n)}$.}

\section{Applicability, Limitations and Future Work}

As discussed initially, the framework presented in this work delineates a
design space intended to contain current indexing methods, as well as
countless variations and combinations, ripe for exploration and
experimentation. If one stays within the bounds of this design space, by
conforming to the necessary axioms, the various correctness results apply.
However, these axioms also specify very precisely the applicability of these
results, and possible jumping-off points for new indexing approaches. Deciding
to break with one or more of the assumptions laid down will necessarily
make some of my results inapplicable, but will also create opportunities for
wholly new designs.

Of course, the work presented here is limited in its scope to a
\emph{theoretical} study of sprawls and ambits. Some tentative experiments
have been performed, e.g., to verify that there are indeed cases where
non-linear ambits will outperform their linear counterparts (cf.\
\cref{sec:nonl}), and that the heuristic facet and focus selection procedures
of \cref{sec:optcoeff,sec:focusselection} do indeed seem to yield reasonable
results. However, proper experimentation is still needed to arrive at any
definitive conclusions. Beyond basic implementation and benchmarking work,
actually exploring the design space in a systematic, perhaps even automated,
fashion might also be an interesting topic for further research.

\section*{Acknowledgements}

\noindent
The author wishes to thank Ole Edsberg and Jon Marius Venstad for highly
fruitful discussions, and André Philipp, Joakim Skarding and Odd Magnus
Trondrud for their preliminary empirical work based on the ideas presented
here.

\appendix
\inapxtrue
\counterwithin{construction}{section}
\counterwithin{observation}{section}

\section{Proofs of Selected Theorems}
\label{apx:proofs}

\thmdihyper*
\begin{proof}
    $(i\Rightarrow ii)$
    Assume that \tup{V,\L_1} is a traversal repertoire.
    %
    Let $(V,E)$ and $\L_2$ be as defined by \cref{constr:reptohyper}.
    What remains is to show that $\L_2=\L_1$, i.e., that a sequence $\gamma$
    is a member of $\L_1$ if and only if it is a traversal of \tup{V,E}.

    By assumption, \tup{V,\L_1} obeys the traversal axioms; by the reverse
    implication ($ii\Rightarrow i$, proven below), so does \tup{V,\L_2}.
    %
    %
    We proceed by induction on sequence length, with the inductive step proven
    by exhaustion. First, the base case.

    \begin{case}
        $\gamma=\emptyseq$. In this case, $\gamma$ is in both $\L_1$ and
        $\L_2$, as guaranteed by \cref{axiom:nonempty}.
    \end{case}

    \noindent
    The remaining cases cover the inductive step, where we let $\gamma=\tau
    x$, with $x\in V$.

    \begin{case}
        $\tau\notin\L_1$. By the inductive hypothesis, we have
        $\tau\notin\L_2$ and \cref{axiom:hereditary} immediately yields the
        inductive step, with $\tau x$ absent from both languages.
    \end{case}

    \noindent
    Now assume $\tau$ is present in both languages, but that $\tau x$ is not.
    Our goal is to show that once \cref{alg:hypertraversal} produces the
    traversal $\tau$, $x$ either has been eliminated or remains undiscovered,
    as covered by the next two cases.

    \begin{case}
        $\tau\in\L_1$, $\tau x\notin\L_1$, but $\alpha x\in\L_1$ for some
        $\alpha\in V^\hiast$ where $\tilde\alpha\subsetneq\tilde\tau$.
        If $x\in\tilde\tau$, \cref{alg:hypertraversal} will not traverse $x$
        again. Otherwise, by \cref{axiom:negedge}, there is
        a negative edge \edge{\tilde\tau}{x} in $E$.
        Either way, we have $\tau x\notin\L_2$, as desired.
    \end{case}

    \begin{case}
        \label{case:posedge}
        $\tau\in\L_1$ and $\alpha x\notin\L_1$ for every
        $\alpha\in V^\hiast$ where $\tilde\alpha\subseteq\tilde\tau$. The
        condition in \cref{axiom:posedge} will never apply to any such
        $\tilde\alpha$, so $x$ can not have been discovered.
    \end{case}

    \noindent
    Finally, assume that $\tau x$ (and therefore $\tau$) \emph{is} in $\L_1$.
    We then wish to show that $x$ is available after traversing $\tau$ in
    \tup{V,E,\sigma}, i.e., that it has been discovered but not eliminated.
    \begin{case}
        $\tau x\in\L_1$. Because $\tau x$ is in $\L_1$,
        \cref{axiom:posedge} mandates
        a positive edge \edge{\tilde\tau}{x} in $E$, and \cref{axiom:simple}
        forbids $x\in\tilde\tau$, so the only obstacle to availability in
        \tup{V,E,\sigma} would be elimination.
        Assume there is a negative edge \edge{\tilde\omega}{x} in $E$,
        for some $\tilde\omega\subseteq\tilde\tau$, meaning
        that for some $\omega\in\L_1$, we have $\omega x\notin\L_1$. If
        \cref{axiom:negedge} is to apply, there must then be some $\alpha$
        with $\tilde\alpha\subsetneq\tilde\omega$ for which $\alpha x\in\L_1$.
        However, by the interval property~\axiomref{axiom:interval}, because
        $\tilde\alpha\subseteq\tilde\omega\subseteq\tilde\tau$
        and $\omega,\alpha x,\tau x\in\L_1$, we have $\omega x\in\L_1$, which
        is a contradiction. Therefore there can be no such edge, and
        consequently $\tau x\in\L_2$.
    \end{case}
    \noindent
    These cases are exhaustive, and in each case $\tau x\in\L_1$ if and only
    if $\tau x\in\L_2$, which proves the inductive step. This, in turn, means
    that $\L_1=\L_2$, which proves the implication from \ref{stmt:trav} to
    \ref{stmt:hyptrav}.

    $(ii\Rightarrow i)$
    We need to show that the traversals of \cref{alg:hypertraversal} satisfy
    the traversal axioms:
    Before the first iteration, the traversal is
    empty~\axiomref{axiom:nonempty}; after this, each node is traversed at
    most once~\axiomref{axiom:simple}, one at a
    time~\axiomref{axiom:hereditary}, provided they have been discovered but
    not eliminated~\axiomref{axiom:interval}. In particular, discovery and
    elimination depend only on the \emph{set} of nodes traversed,
    so the interval property holds for
    $\tilde\alpha\subseteq\tilde\tau\subseteq\tilde\omega$, even if $\alpha$
    and $\tau$ are not prefixes of $\omega$.
\end{proof}

\begin{construction}
    \label{constr:schemetosprawl}
    From a monotone traversal scheme $\tightfml{\tup{V,\L_i}}{i}$ with
    workload $\fml{Q_i}{i}$, in some universe $U$, we produce a sprawl
    \tup{V,E,P,N}.
    Let $\tightfml{\tup{V,E_i,\sigma_i}}{i}$ be the results of
    \cref{constr:reptohyper} for the corresponding traversal repertoires, with
    extra edges as mandated by the following axiom, for all $Q_i,Q_j$:
\begin{axioms}[start=3]{E}
\item If $e\in E_j$, $\sigma_j(e)=-1$ and $Q_i\subseteq Q_j$
    then $e\in E_i$ and $\sigma_i(e)=-1$.
    \label{axiom:extraedges}
\end{axioms}
    Let $E$ be the union of all the edge sets. For each edge $e\in E$,
\begin{stmts}[widest=2]
\item Let $\spc{\fml{Q_i}{i\in I''}}{\fml{Q_i}{i\in I'}}$\kern-2pt{}
    be the maximal queries for which
    $\sigma_i(e)=-1$ or $e\notin E_i$; and
    \label{item:iprime}
\item Let $\fml{Q_i}{i\in I''}$\kern-2pt{} be the maximal queries for which
    $\sigma_i(e)=-1$.
    \label{item:idoubleprime}
\end{stmts}
Finally, let $P(e)$ and $N(e)$ be the families $\fml{U\setminus Q_i}{i}$ for
$i$ in $I'$ and $I''$, respectively.
\end{construction}
The addition of negative edges in the previous
construction is a form of normalization, to prevent the somewhat arbitrary
transition from non-discovery to elimination whenever an expanding query
introduces a traversal $\alpha$ that activates \cref{axiom:negedge}. This
makes non-elimination monotone in the query (under inclusion), just like
discovery---a property that is needed in the proof of \cref{thm:schemes}. This
normalization is harmless, as the following observation
attests.
\begin{observation}
\label{obs:samescheme}
Each signed hyperdigraph \tup{V,E_i,\sigma_i} in \cref{constr:schemetosprawl}
has the corresponding language \tup{V,\L_i} as its traversal repertoire.
\end{observation}
\begin{proof}
    If $\sigma_j(\tilde\tau, x)=-1$, than any traversal $\tau$ that is an
    ordering of $\tilde\tau$ will result in the elimination of $x$. By
    \cref{thm:dihyper}, this means that $\tau x\notin\L_j$.
    Because $Q_i\subseteq Q_j$, \cref{axiom:monotone} requires $\L_i\subseteq
    \L_j$, and so we also have $\tau x\notin\L_i$ for every such $\tau$, and
    adding \edge[e]{\tilde\tau}{x} to $E_i$ with $\sigma_i(e)=-1$, if absent,
    does no harm.
\end{proof}

\thmschemes*
\begin{proof}
    $(i\Rightarrow ii)$
    Let $U$ be any appropriate universe, i.e., with $V,Q_i\subseteq U$, for
    $i\in I$, and let \tup{V,E,P,N} be the sprawl produced by
    \cref{constr:schemetosprawl}.
    Let $Q_i$ be any of the queries; we need to show that resolving $Q_i$ on
    our sprawl produces the traversal repertoire \tup{V,\L_i}. More
    specifically, we can show that the signed hyperdigraph \tup{V,E'\!,\sigma}
    resulting from applying $Q_i$ to the sprawl (as described by
    \cref{constr:sprawltohyper}) is the same as \tup{V,E_i,\sigma_i}, which
    has \tup{V,\L_i} as its repertoire (cf.~\cref{obs:samescheme}).
    We proceed by cases for an arbitrary edge
    \edge[e]{\tilde\tau}{x} in $E$.

    \begin{case}
        $\sigma_i(e) = +1$. We need to show that $\sigma(e)=+1$, and
        per \cref{constr:sprawltohyper}, the condition for this is that
        $Q_i$ intersect all sets in the families $P(e)$ and $N(e)$. For any
        region $R_j$ in $P(e)$ or $N(e)$, its complement $Q_j\deq U\setminus
        R_j$ is, by construction, a query for we do \emph{not} have
        $\sigma_j(e) = +1$, which, by \cref{axiom:posedge}, means that for no
        ordering $\tau$ of $\tilde\tau$ do we have $\tau x\in \L_j$. By
        \cref{axiom:monotone}, this is true for any subset of $Q_j$ as well.
        This undercuts the only way in which \cref{constr:reptohyper} could
        mandate $\sigma_i(e)=+1$, and so $Q_i\not\subseteq Q_j$, and
        consequently $Q_i$ must intersect $R_j$.
    \end{case}

    \begin{case}
        $\sigma_i(e) = -1$. We need to show that there is some region $R\in
        N(e)$ that does not intersect $Q_i$. Let $Q_j$ be a maximal query for
        which $\sigma_j(e) = -1$, and where $Q_i\subseteq Q_j$. $Q_i$
        cannot intersect $R_j$, and, by construction, we have $R_j\in N(e)$.
    \end{case}

    \begin{case}
        $e\notin E_i$. We need to show that $Q_i$ intersects every region in
        $N(e)$, but not every region in $P(e)$.
        First, Let $Q_j$ be some query where $Q_i\subseteq Q_j$. Because of
        \cref{axiom:extraedges} of \cref{constr:schemetosprawl}, it can not
        be the case that $\sigma_j(e)=-1$, and so $R_j\notin N(e)$, meaning
        that $Q_i$ intersects all regions in $N(e)$.
        Second, because $e\notin E_i$, by
        \cref{constr:schemetosprawl}\ref{item:iprime}, there is som $j\in I'$
        with $Q_i\subseteq Q_j$ and $R_j\in P(e)$, so there is at least one
        region in $P(e)$ that does not intersect $Q_i$.
    \end{case}

    $(ii\Rightarrow i)$
    By definition, each language in the traversal scheme of a sprawl is the
    traversal repertoire of some signed hyperdigraph
    (\cref{constr:sprawltohyper}, \cref{alg:sprawltraversal} and
    \cref{def:sprawlscheme}), and therefore an actual traversal repertoire
    (\cref{thm:dihyper}). Extending some query $Q_i$ to some superset
    $Q_j$ will never cause \cref{constr:sprawltohyper} to remove positive
    edges or add negative ones, and so any traversal present in $\L_i$ must
    still be present in $\L_j$, satisfying \cref{axiom:monotone},
    making $\tightfml{\tup{V,\L_i}}{i\in I}$ a monotone traversal scheme.
\end{proof}

\obslocalres*
\begin{proof}
    First, assume that $\res$ is a responsibility assignment. For any
    query \set{u},
    intersecting $R$ means $u\in R$; thus, by \cref{axiom:posintersect}, there
    must be a path $\Pi_u$ where $u\in R$ for every $R\in P(e),
    e\in\tilde\Pi_u$.
    The responsibilities $\res(e)$ are, by definition, exactly the nodes $u$
    for which this situation obtains, and thus \cref{it:subr} holds.
    Any edge $e'$ with target $v$ in such a path $\Pi_u$ would, by
    \cref{axiom:negintersect}, have $u\in R'$, for every $R'\in N(e')$.
    The node $v$ would either be the source of some edge $e\in\tilde\Pi_u$
    (i.e., $u\in\res(e)$) or the target of the last edge (i.e., $v=u$). The
    nodes $u$ for which this is the case are, by~\eqref{eq:noderes}, exactly
    the responsibilities $\res(v)$, and so \cref{it:subrprime} holds.
    Finally, consider any node $u\in\res(v)$. Whether $u=v$ or not, $v$ must
    be the target of some edge in $\Pi_u$, and so $e_i\in\tilde\Pi_u$, for
    at least one of its incoming edges $e_i$. This means that
    $u\in \res(e_i)$, which implies \cref{it:subrunion}.

    Conversely, assume that the sprawl is acyclic and that every node is the
    target of at least one edge, and consider any relation $\res\subseteq
    E\times V$ satisfying \cref{it:subr,it:subrprime,it:subrunion},
    using the shorthand from~\eqref{eq:noderes}.
    Let $u$ be any node and $Q$ be any query, with $u\in Q$. We can now show
    that there is a hyperpath $\Pi_v$ satisfying
    \cref{axiom:posintersect,axiom:negintersect} from $\emptyset$ to any node
    $v\in V$ for which $u\in\res(v)$, with $u=v$ being
    the special case we care about.
    Because the sprawl is acyclic, we may order the nodes $v_1,\dots,v_n$ so
    that if and $v_i=\tgt(e)$ and $v_j\in\src(e)$, then $j<i$, for any edge
    $e$ with at least one positive region.
    We proceed by induction on $i$. By \cref{it:subrunion} there is
    at least one edge $e$ with $\tgt(e)=v$ and $u\in\res(e)$, and thereby,
    by~\cref{it:subr}, $u\in R$, for each $R\in P(e)$.\endnote{Note that there
    may be no such regions, which is perfectly fine.}
    By assumption, we already have the requisite paths to each node in
    $\src(e)$, and so we have established \cref{axiom:posintersect}.
    All that remains in order to establish \cref{axiom:negintersect} is to
    consider any edges $e'$ with $\tgt(e')=v$. For any region $R'\in N(e')$,
    \cref{it:subrprime} tells us $u\in R'$, which yields the desired result.
\end{proof}

\section{Additional Remarks}
\label{apx:remarks}

\greedoidremark*

\forwardchainingremark*

\antimatroidremark*

\section{Auxiliary Definitions}

\subsection{Directed Hypergraphs}
\label{sec:hypergraphs}

\begin{definition}
    \label{def:hypergraph}
    A \emph{directed hypergraph} or \emph{hyperdigraph} is a generalization of
    a directed graph, where each edge may have multiple sources. Specifically
    a directed hypergraph \tup{V,E} consists of a finite, nonempty set $V$ of
    \emph{nodes} and a finite set $E$ of edges with \emph{sources}
    $\src(e)\subseteq V$ and \emph{target} $\tgt(e)\in V$. A \emph{signed}
    hyperdigraph \tup{V,E,\sigma} consists of a hyperdigraph \tup{V,E} and a
    sign function $\sigma:E\to\set{-1,+1}$. An edge $e\in E$ is said to be
    \emph{positive} (resp., \emph{negative}) if $\sigma(e)$ is positive
    (resp., negative).
    To indicate that $S=\src(e)$ and $t=\src(e)$, we may use the shorthand
    $e:S\to t$. If we need not name the edge, we write $S\to t$. Note,
    however, that we may have $e_1,e_2:S\to t$ with $e_1\neq e_2$.
\end{definition}

\begin{definition}
    \label{def:roots}
    The target of a sourceless positive edge is called a \emph{root node} or
    simply a \emph{root}. A hyperdigraph may be specified by giving a set $V$ of
    nodes, a set $E$ of edges, and a set $V_0\subseteq V$ of roots. Each
    root $v_i\in V_0$ implicitly defines a sourceless edge $e_i\notin E$
    with $\tgt(e_i)\deq v_i$, and the hyperdigraph thus specified is
    \tup{V,E\cup\set{e_i:v_i\in V_0}}.
\end{definition}

\begin{definition}
    \label{def:hyperpath}
    Given a signed hyperdigraph \tup{V,E}, a \emph{directed hyperpath} or
    simply \emph{path} from $S\subseteq V$ to $t\in V$ is a set
    $\tilde\Pi\subseteq E$ that may be ordered into a sequence
    $\Pi\deq\seq{e_1,\dots,e_k}$ of distinct edges subject to the following,
    for $i=1,\dots,k$:
    \begin{axioms}{P}
    \item
        $
        \src(e_i)
        \subseteq
        S
        \cup
        \set{\tgt(e_j):j < i}
        $.
    \item $t = \tgt(e_k)$.
    \item No strict subset of $\tilde\Pi$ is a path from $S$ to $t$ in
        \tup{V,E}.
    \end{axioms}
    The \emph{node set} of $\Pi$ is the set of all sources and targets of its
    edges.
\end{definition}
\begin{remark}
    The previous definition is essentially equivalent to that given by
    \citet{Ausiello:2001}, except their requiring $S\neq\emptyset$ and
    defining the path to be a subhypergraph.
\end{remark}

\begin{algorithm}
    \label{alg:old:hypertraversal}
    An unsigned hyperdigraph $\tup{V,E}$ is \emph{traversed}
    by \emph{discovering} and \emph{traversing}
    nodes, as described in the following.
    Nodes are \emph{available} if they has been discovered but not traversed.
    Edges are \emph{active} once their sources have been traversed.
    The following steps are repeated, starting with the second step in the
    first iteration:
    \begin{steps}
        \item One of the available nodes is selected and traversed.
            \label{step:old:exploration}
        \item The targets of active
            edges are discovered.
            \label{step:old:discovery}
    \end{steps}
    The steps are repeated as long as there are nodes available.
    When there are several nodes available in the first step,
    the choice is made using a \emph{traversal heuristic}.%
\end{algorithm}

\label{p:impl}%
In a practical implementation, available nodes will typically be kept in a
priority queue, with the traversal heuristic defining the priority. Activation
of edges can be handled efficiently by tracking their number of traversed
sources. Any such state information could be reset between traversals in
constant time without increasing asymptotic space usage~\citep[see,
e.g.,][]{Navarro:2012}.
The priority of a node may be updated, for example, whenever it is
rediscovered in the first step.

\begin{remark}
    The version of hyperdigraph traversal described in
    \Cref{alg:old:hypertraversal} does not take a set of starting nodes as a
    parameter. If it did, these would simply be available from the beginning,
    i.e., added to the priority queue before the main loop. This behavior may
    be emulated, however, by treating the starting nodes as roots (cf.\
    \cref{def:roots}), and traversing the resulting hyperdigraph.
\end{remark}

\begin{example}
    Directed graphs correspond to the special case where each edge has exactly
    one source. The signed hyperdigraph traversal described in
    \cref{alg:old:hypertraversal} corresponds to ordinary graph traversal in
    this case.
\end{example}

\subsection{Basic Model Theory}
\label{sec:modeltheory}

\begin{definition}
    \label{def:structure}
    A \emph{signature} is a set $\Sigma$ of \emph{function} and \emph{relation
    symbols}, each with non-negative \emph{arity}. A \emph{$\Sigma$-structure}
    has the following data:
    \begin{stmts}[widest=ii]
    \item A non-empty set $A\deq\dom \A$, known as the \emph{domain} of $\A$;
        and
    \item A family $\fml{S^\A}{S\in\Sigma}$ of functions and relations on $A$.
    \end{stmts}
    These data obey the following axioms:
    \begin{axioms}{S}
    \item $S^\A\spcc{\null\subseteq\null}{:}A\?^\hilc{n}\to A$ if $S$ is an
        $n$-ary function symbol; and
    \item $\spc{S^\A\spcc{\null\subseteq\null}{:}A\?^\hilc{n}\to
        A}{S^\A\subseteq A\?^\hilc{n}}$ if $S$ is an $n$-ary relation symbol.
    \end{axioms}
    We refer to $S^\A$ as the \emph{$\A$-interpretation} of $S$.
\end{definition}

\begin{remarks}
    \label{rem:structvariations}
    \begin{paras}
    \item
    It is common to permit only arities $n\geq 1$, and to have a separate
    class of \emph{constants} $S$, with $S^\A\in A$. Nullary function symbols
    essentially correspond to such constants~\citep{Tent:2012}, and nullary
    relations to truth values~\citep{Poizat:2000}.
    \item
    \label{rem:it:multisorted}
    It is possible to define \emph{many-sorted} logical languages and
    structures, where each term and each argument position is assigned a
    \emph{sort}, and these correspond to a partition of
    $A$~\citep[p.\,5]{Tent:2012}. Each sort $i$ is also assigned
    its quantifier $\forall_i$, but if we introduce
    a predicate $S_i$ that uniquely picks out the subset of $A$ corresponding
    to sort $i$, we may replace any quantification $\forall_ix\,\varphi(x)$
    with $\forall x\,S_i(x)\to\varphi(x)$, translating many-sorted formulas
    to equivalent single-sorted ones.\endnote{The common definition is
    used for existential quantification, i.e., $\exists_i
    x\,\varphi(x)=\neg\forall_i x\,\neg\varphi(x)$; the single-sorted
    translation becomes $\exists x\, S_i(x)\land\varphi(x)$.}
    \end{paras}
    When transforming a many-sorted structure to a single-sorted one,
    relations remain intact, but functions must be (arbitrarily) extended so
    each argument may be drawn from all of
    $A$~\citep[Sect.\,4.3]{Enderton:2001}
\end{remarks}

\begin{definition}
    \label{def:cartesianproduct}
    We define the \emph{Cartesian product} of structures $\A$ and $\B$ as
    the structure $\A\times\B$ with
    \[
        \dom\, (\A\times\mkern-1mu\B) \deq(\dom\A)\times(\dom\B)\eqcomma
    \]
    where functions and relations act elementwise: $F^{\A\times\mkern-1mu\B}$
    produces a pair of the results from $F^\A$ and $F^\B\?$, and
    $R^{\A\times\mkern-1mu\B}$ holds iff both $R^\A$ and $R^\B$ do.
    Powers and products of multiple structures are defined in the obvious
    manner.
\end{definition}

\begin{examples}\label{ex:struct}
    \begin{paras}
    \item A partially ordered set $P$ may be seen as a structure
        $\mathfrak{P}$ with domain $P$ and a single binary relation symbol
        $\leq$, representing the partial order relation $\leq^\mathfrak{P}$ on
        $P$.
        Let $\mathfrak{M}=\mathfrak{P}^m$. Then
        $\dom\mathfrak{M}$ consists of $m$-tuples from $\dom\mathfrak{P}$,
        and
        $\tup{x_1,\dots,x_m}\leq^\mathfrak{M}\tup{y_1,\dots,y_m}$
        holds
        iff
        $x_i\leq^\mathfrak{P} y_i$ for $i=1,\dots,m$.
    \item
        \label{subex:vecspace}
        A \emph{vector space $V$} over a field $F$ of scalars may be seen as
        a structure $\mathfrak{V}$ with domain $V$ and signature
        $\set{\mkern1mu+\mkern1mu,0}\cup F$, where the binary function symbol
        $+$ represents the vector addition $+^\mathfrak{V}$, the constant
        symbol~$0$ represents the vector space origin~$0^\mathfrak{V}$, and
        each $a\in F$ is used as a unary function symbol representing the
        operation $a^\mathfrak{V}:V\to V$ of scalar multiplication by $a$.
    \end{paras}
\end{examples}

\begin{definition}
    A string constructed using $\Sigma$ along with variables, quantifiers and
    Boolean operators, respecting arities and the syntax of first-order logic,
    is called a \emph{formula}. A formula whose variables are all bound by
    quantifiers (i.e., one without \emph{free variables}) is a
    \emph{sentence}. Such a sentence $\varphi$ is interpreted recursively in a
    structure $\A$ of signature $\Sigma$, by treating each non-logical symbol
    as its $\A$-interpretation, and giving the logical symbols their canonical
    meaning. If the resulting statement is true, we write $\A\models\varphi$,
    or say that $\varphi$ is \emph{true in \A}.
\end{definition}

\begin{example}
    If $\mathfrak{P}$ represents any poset with order symbol $R$, then:
    \[
        \mathfrak{P}\models
        \forall xyz\enskip
        R(x,y)\land R(y,z)\rightarrow R(x,z)
    \]
    That is, for all $x,y,z\in\dom\A$, if $\tup{x,y},\tup{y,z}\in
    R^\mathfrak{P}$ then $\tup{x,z}\in R^\mathfrak{P}$.
\end{example}

\begin{remark}
    \label{rem:manysortedreduction}
    Let $\A$ be a many-sorted structure, and $\varphi$ a sentence in the
    corresponding many-sorted language; let $\A^*$ and $\varphi^*$ be the
    corresponding translations, as described in
    \cref{rem:structvariations}\ref{rem:it:multisorted}. Then
    $\A\models\varphi$ if and only if $\A^*\models\varphi^*$.
\end{remark}

\begin{definition}
    \label{def:homomorphism}
    For any two structures $\A$ and $\B$ with the same signature $\Sigma$, a
    map $f:\A\to\B$ is a function from $\dom\A$ to $\dom\B$. A
    \emph{homomorphism} is a map $f:\A\to\B$ that satisfies the following
    axioms for all function and relation symbols $F$ and $R$ in $\Sigma$, of
    any arity $n\geq 0$, and all $a_1,\dots,a_n\in\dom\A$:
    \begin{axioms}{H}
    \item $R^\A(a_1,\dots,a_n)\Rightarrow R^\B(f(a_1),\dots,f(a_n))$;
    \label{it:impl}
    \item $f(F^\A(a_1,\dots,a_n)) = F^\B(f(a_1),\dots,f(a_n))$.
    \label{ax:homomophism:func}
    \end{axioms}
    An \emph{embedding}
    is an injective homomorphism for which~\ref{it:impl} is biconditional, and
    an \emph{isomorphism} is a surjective (and, consequently, bijective)
    embedding.
\end{definition}

\begin{examples}
    \label{ex:morphisms}
\begin{paras}
\item For posets $\mathfrak{P}_0$ and $\mathfrak{P}_1$, the homomorphisms
    $f:\mathfrak{P}_0\to\mathfrak{P}_1$ are exactly the nondecreasing maps from
    $\mathfrak{P}_0$ to $\mathfrak{P}_1$.
\item For vector spaces $\mathfrak{V}_0$ and $\mathfrak{V}_1$, represented as
    in \ref{ex:struct}\ref{subex:vecspace}, the homomorphisms
    $f:\mathfrak{V}_0\to\mathfrak{V}_1$ are exactly the linear maps from
    $\mathfrak{V}_0$ to $\mathfrak{V}_1$.
\item For any structure $\A$, the diagonal embedding $\Delta_m:\A\to\A^m$
    is, as its name implies, an embedding.
    \label{it:diagonal}
\end{paras}
\end{examples}

\begin{definition}
A formula $\varphi$ with free variables $x_1,\dots,x_n$ is \emph{true of
$a_1,\dots,a_n$ in} $\A$, written $\A\models\varphi(a_1,\dots,a_n)$, if
$\varphi$ is true in $\A$ when each $x_i$ is seen as a name for $a_i\in\A$.
For two structures $\A$ and $\B$ with the same signature, a map
$f:\A\to\B$ \emph{preserves} a formula $\varphi$ with $n$ free variables if
\[
    \A\models\varphi(a_1,\dots,a_n)
    \implies
    \B\models\varphi(f(a_1),\dots,f(a_n))\eqcomma
\]
for all $a_1,\dots,a_n\in\A$.
If $\varphi$ has $n+1$ free variables, then, given $n$ parameters $b_i\in\A$,
we write $\varphi(\A,b_1,\dots,b_n)$ for the set
\set{a\in\dom\A:\A\models\varphi(a,b_1,\dots,b_n)}.
\end{definition}

\begin{definition}
    Formulas are \emph{existential} if they combine quantifier-free formulas
    using $\land$, $\lor$ and $\exists$. They are \emph{positive} if they do
    not contain negation (or, by extention, implication). A \emph{basic Horn
    formula} has the form $\varphi_1\land\cdots\land\varphi_n\rightarrow\psi$,
    where each $\varphi_i$ is an atomic formula (i.e., without logical
    operators) and $\psi$ is either an atomic formula or $\bot$ (false).
    A \emph{Horn formula} is built from basic Horn formulas using $\land$,
    $\exists$ and $\forall$.
\end{definition}

\printendnotes[mynotes]

\let\L\oldL

\printbibliography

@book{Korte:1991,
    author = {Bernhard Korte and László Lovász and Rainer Schrader},
    title = {Greedoids},
    publisher = {Springer},
    year = {1991},
    volume = {4},
    series = {Algorithms and Combinatorics},
}

@article{Goguen:1967,
    title = "{$L$}-fuzzy sets",
    journal = "Journal of Mathematical Analysis and Applications",
    volume = "18",
    number = "1",
    pages = "145--174",
    year = "1967",
    author = "J. A. Goguen"
}

@article{Aho:1974,
    title={The Design and Analysis of Computer Algorithms},
    author={Aho, Alfred V and Hopcroft, John E and Ullman, Jeffrey D},
    journal={Addison-Wesley},
    year={1974}
}

@article{Ruiz:1986,
    title={An algorithm for finding nearest neighbours in (approximately)
           constant average time},
    author={Ruiz, Enrique Vidal},
    journal={Pattern Recognition Letters},
    volume={4},
    number={3},
    pages={145--157},
    year={1986},
    publisher={Elsevier}
}

@inproceedings{Skopal:2004a,
    author = {Tomá{\v{s}} Skopal},
    booktitle = {Proceedings of the Annual International Workshop on
                 Databases, Texts, Specifications and Objects},
    editor = {V. Sná{\v{s}}el and J. Pokorn{\'{y}} and K. Richta},
    publisher = {Technical University of Aachen},
    title = {Pivoting {M}-tree},
    subtitle = {A Metric Access Method for Efficient Similarity Search},
    volume = {98},
    year = {2004}
}

@article{Sahadevana:1987,
    author = {P. V. Sahadevana},
    title = {The theory of the egglipse},
    subtitle = {A new curve with three focal points},
    journal = {International Journal of Mathematical Education in Science and
               Technology},
    year = {1987},
    volume = {18},
    number = {1},
    pages = {29--39},
}

@article{Cantor:1884,
    title={De la puissance des ensembles parfaits de points},
    author={Cantor, Georg},
    journal={Acta Mathematica},
    volume={4},
    number={1},
    pages={381--392},
    year={1884},
    publisher={Springer}
}

@book{Zimmermann:2001,
    title={Fuzzy set theory and its applications},
    author={Zimmermann, Hans-Jürgen},
    year={2001},
    publisher={Springer}
}

@book{Garey:1979,
    Author = {Michael R. Garey and David S. Johnson},
    Publisher = {W. H. Freeman and Company},
    Title = {Computers and Intractability},
    subtitle = {A Guide to the Theory of {NP}-Completeness},
    Year = {1979}
}

@incollection{Ausiello:2001,
    title={Directed hypergraphs},
    subtitle={Problems, Algorithmic Results, and a Novel Decremental Approach},
    author={Ausiello, Giorgio and Franciosa, Paolo G and Frigioni, Daniele},
    booktitle={Theoretical Computer Science},
    volume={2202},
    series={Lecture Notes in Computer Science},
    pages={312--328},
    year={2001},
    publisher={Springer}
}

@article{Kalantari:1983,
    Author = {Iraj Kalantari and Gerard McDonald},
    Journal = {{IEEE} Transactions on Software Engineering},
    Number = {5},
    Pages = {631--634},
    Title = {A Data Structure and an Algorithm for the Nearest Point Problem},
    Volume = {9},
    Year = {1983}
}

@article{Uhlmann:1991,
    Author = {Jeffrey K. Uhlmann},
    Journal = {Applied Mathematics Letters},
    Number = {5},
    Pages = {61--62},
    Title = {Metric trees},
    Volume = {4},
    Year = {1991}
}

@article{Mico:1994,
    author = {Mar{\'{i}}a Luisa Micó and José Oncina and Enrique Vidal},
    journal = {Pattern Recognition Letters},
    number = {1},
    pages = {9--17},
    title = {A new version of the nearest-neighbour approximating and
             eliminating search algorithm ({AESA}) with linear preprocessing
             time and memory requirements},
    volume = {15},
    year = {1994}
}

@inproceedings{Brin:1995,
    author = {Sergey Brin},
    booktitle = {Proceedings of 21th International Conference on Very Large
                 Data Bases},
    editor = {Umeshwar Dayal and Peter M. D. Gray and Shojiro Nishio},
    pages = {574--584},
    publisher = {Morgan Kaufmann},
    title = {Near Neighbor Search in Large Metric Spaces},
    year = {1995}
}

@inproceedings{Dohnal:2001,
    author = {Vlastislav Dohnal and Claudio Gennaro and Pasquale Savino and
              Pavel Zezula},
    booktitle = {Proceedings of the Nono Convegno Nazionale Sistemi Evoluti
                 per Basi di Dati},
    title = {Separable Splits of Metric Data Sets},
    year = {2001}
}

@article{Lokoc:2014,
    title = "On indexing metric spaces using cut-regions",
    journal = "Information Systems",
    volume = "43",
    pages = "1--19",
    year = "2014",
    author = "Jakub Lokoč and Juraj Moško and Přemysl Čech and Tomáš Skopal",
}

@inproceedings{Navarro:2001,
    Author = {Gonzalo Navarro and Nora Reyes},
    Booktitle = {Proceedings of the {XXI} Conference of the Chilean Computer
                 Science Society},
    Pages = {213--222},
    Title = {Dynamic Spatial Approximation Trees},
    Year = {2001}
}

@article{Navarro:2002,
    Author = {Gonzalo Navarro},
    Journal = {The {VLDB} Journal},
    Number = {1},
    Pages = {28--46},
    Title = {Searching in metric spaces by spatial approximation},
    Volume = {11},
    Year = {2002}
}

@article{Uhlmann:1991a,
    Author = {Jeffrey K. Uhlmann},
    Journal = {Information Processing Letters},
    Number = {4},
    Pages = {175--179},
    Title = {Satisfying general proximity/similarity queries with metric trees},
    Volume = {40},
    Year = {1991}
}

@inproceedings{Lokoc:2010,
    title={On applications of parameterized hyperplane partitioning},
    author={Loko{\v{c}}, Jakub and Skopal, Tomá{\v{s}}},
    booktitle={Proceedings of the Third International Conference on Similarity
               Search and Applications},
    pages={131--132},
    year={2010},
    organization={ACM}
}

@inproceedings{Bugnion:1993,
    title={Approximate multiple string matching using spatial indexes},
    author={Bugnion, Edouard and Roos, Thomas and Shi, Fei and Widmayer, Peter
            and Widmer, Felizitas},
    booktitle={Proceedings of the 1st South American Workshop on String
               Processing},
    pages={43--54},
    year={1993}
}

@article{Chazelle:1993,
    title={An optimal convex hull algorithm in any fixed dimension},
    author={Chazelle, Bernard},
    journal={Discrete \& Computational Geometry},
    volume={10},
    number={4},
    pages={377--409},
    year={1993},
    publisher={Springer}
}

@book{Cygan:2015,
    title={Parameterized Algorithms},
    author={Cygan, Marek and Fomin, Fedor V and Kowalik, {\L}ukasz and
            Lokshtanov, Daniel and Marx, Dániel and Pilipczuk, Marcin and
            Pilipczuk, Micha{\l} and Saurabh, Saket},
    year={2015},
    publisher={Springer}
}

@article{Chazelle:2000,
    title={The soft heap},
    subtitle={An approximate priority queue with optimal error rate},
    author={Chazelle, Bernard},
    journal={Journal of the ACM (JACM)},
    volume={47},
    number={6},
    pages={1012--1027},
    year={2000},
    publisher={ACM}
}

@article{Blum:1973,
    title={Time bounds for selection},
    author={Blum, Manuel and Floyd, Robert W and Pratt, Vaughan and Rivest,
            Ronald L and Tarjan, Robert E},
    journal={Journal of Computer and System Sciences},
    volume={7},
    number={4},
    pages={448--461},
    year={1973},
    publisher={Elsevier}
}

@book{Hodges:1993,
    title = {Model Theory},
    author = {Hodges, Wilfrid},
    series  = {Encyclopedia of Mathematics and its Applications},
    volume = {42},
    year = {1993},
    publisher = {Cambridge University Press}
}

@book{Poizat:2000,
    author={Bruno Poizat},
    title={A Course in Model Theory},
    subtitle={An Introduction to Contemporary Mathematical Logic},
    year={2000},
    publisher={Springer},
}

@book{Tent:2012,
    title={A Course in Model Theory},
    author={Tent, Katrin and Ziegler, Martin},
    year={2012},
    publisher={Cambridge University Press}
}

@book{Crama:2011,
    author = {Yves Crama and Peter L. Hammer},
    title = {Boolean Functions},
    subtitle = {Theory, Algorithms, and Applications},
    publisher = {Cambridge University Press},
    year = {2011},
    volume = {142},
    series = {Encyclopedia of Mathematics and its Applications},
}

@article{Gallo:1993,
    title = "Directed hypergraphs and applications",
    journal = "Discrete Applied Mathematics",
    volume = "42",
    number = "2",
    pages = "177--201",
    year = "1993",
    author = "Giorgio Gallo and Giustino Longo and Stefano Pallottino",
}

@inproceedings{Navarro:2012,
    title={Constant-time array initialization in little space},
    author={Navarro, Gonzalo},
    booktitle={Proceedings of the 31st International Conference of the Chilean
               Computer Science Society},
    publisher={IEEE CS Press},
    year={2012}
}

@article{Marek:1999,
    author = {V. W. Marek and A. Nerode and J. B. Remmel},
    title = {Logic programs, well-orderings, and forward chaining},
    journal = {Annals of Pure and Applied Logic},
    year = {1999},
    volume = {96},
    pages = {231--276},
}

@inproceedings{Hellerstein:1995,
    title = "Generalized Search Trees for Database Systems",
    author = "Joseph M. Hellerstein and Jeffrey F. Naughton and Avi Pfeffer",
    booktitle = "Proceedings of the 21st International Conference on Very
        Large Data Bases",
    year = "1995",
    pages = "562--573"
}

@book{Deza:2013,
    author = {Michel Marie Deza and Elena Deza},
    title = {Encyclopedia of Distances},
    publisher = {Springer},
    year = {2013},
}

@article{Mennucci:2013,
    title={On asymmetric distances},
    author={Mennucci, Andrea C. G.},
    journal={Analysis and Geometry in Metric Spaces},
    volume={1},
    pages={200--231},
    year={2013}
}

@book{Zezula:2006,
    Author = {Pavel Zezula and Giuseppe Amato and Vlastislav Dohnal and Michal Batko},
    Publisher = {Springer},
    Title = {Similarity Search},
    subtitle = {The Metric Space Approach},
    Year = {2006}
}

@article{Pestov:2006,
    title={Indexing schemes for similarity search},
    subtitle={An illustrated paradigm},
    author={Pestov, Vladimir and Stojmirovi{\'{c}}, Aleksandar},
    journal={Fundamenta Informaticae},
    volume={70},
    number={4},
    pages={367--385},
    year={2006},
    publisher={IOS Press}
}

@article{Filip:2010,
    title={Fixed point theorems on spaces endowed with vector-valued metrics},
    author={Alexandru-Darius Filip and Adrian Petrusel},
    journal={Fixed Point Theory and Applications},
    volume={2010},
    number={1},
    year={2010},
    publisher={Springer}
}

@BOOK{Schweizer:1983,
    title={Probabilistic Metric Spaces},
    author={B. Schweizer and A. Sklar},
    publisher={Elsevier},
    year={1983},
}

@misc{Conant:2016,
    title={Extending partial isometries of generalized metric spaces},
    author={Conant, Gabriel},
    eprint={1509.04950v3},
    eprinttype={arxiv},
    eprintclass={math.LO},
    year={2016}
}

@book{Enderton:2001,
    author={Herbert B. Enderton},
    title={A Mathematical Introduction to Logic},
    year={2001},
    publisher={Harcourt Academic Press}
}

@article{Bayer:1972,
    title={Organization and maintenance of large ordered indexes},
    author={Bayer, R and McCreight, EM},
    journal={Acta Informatica},
    volume={1},
    number={3},
    pages={173--189},
    year={1972},
    publisher={Springer}
}

@inproceedings{Chavez:1999,
    Author = {Edgar Chávez and José Luis Marroqu{\'{i}}n and Ricardo
              Baeza-Yates},
    Booktitle = {Proceedings of the String Processing and Information
                 Retrieval Symposium \& International Workshop on Groupware
                 ({SPIRE})},
    Pages = {38--46},
    Publisher = {{IEEE} Computer Society},
    Title = {Spaghettis},
    subtitle = {An Array Based Algorithm for Similarity Queries in
             Metric Spaces},
    Year = {1999}
}

@techreport{Yianilos:1999,
    Author = {Peter N. Yianilos},
    Institution = {{NEC} Research Institute},
    Title = {Excluded Middle Vantage Point Forests for Nearest Neighbor Search},
    Year = {1999}
}

@inproceedings{Dohnal:2004,
    Author = {Vlastislav Dohnal},
    Booktitle = {{EDBT} Workshops},
    Editor = {Wolfgang Lindner and Marco Mesiti and Can T{\"u}rker and Yannis Tzitzikas and Athena Vakali},
    Pages = {133--143},
    Publisher = {Springer},
    Series = {Lecture Notes In Computer Science},
    Title = {An Access Structure for Similarity Search in Metric Spaces},
    Volume = {3268},
    Year = {2004}
}

@article{Koufogiannakis:2014,
    title={A nearly linear-time {PTAS} for explicit fractional packing and covering linear programs},
    author={Koufogiannakis, Christos and Young, Neal E},
    journal={Algorithmica},
    volume={70},
    number={4},
    pages={648--674},
    year={2014},
    publisher={Springer}
}

@article{Xu:2015,
    title={A Comprehensive Survey of Clustering Algorithms},
    author={Xu, Dongkuan and Tian, Yingjie},
    journal={Annals of Data Science},
    volume={2},
    number={2},
    pages={165--193},
    year={2015},
    publisher={Springer}
}

@incollection{Berkhin:2006,
    title={A survey of clustering data mining techniques},
    author={Berkhin, Pavel},
    booktitle={Grouping Multidimensional Data},
    pages={25--71},
    year={2006},
    publisher={Springer}
}

@book{Williamson:2011,
    title={The Design of Approximation Algorithms},
    author={Williamson, David P and Shmoys, David B},
    year={2011},
    publisher={Cambridge University Press}
}

@article{Lawvere:1973,
    title={Metric spaces, generalized logic, and closed categories},
    author={Lawvere, F. William},
    journal={Rendiconti del seminario matématico e fisico di Milano},
    volume={43},
    number={1},
    pages={135--166},
    year={1973},
    publisher={Springer}
}

@article{Chavez:2001,
    Author = {Edgar Chávez and Gonzalo Navarro and Ricardo Baeza-Yates and
              José Luis Marroqu{\'{i}}n},
    Journal = {{ACM} Computing Surveys},
    Number = {3},
    Pages = {273--321},
    Title = {Searching in metric spaces},
    Volume = {33},
    Year = {2001}}

@article{Hjaltason:2003,
    Author = {Gisli R. Hjaltason and Hanan Samet},
    Journal = {{ACM} Transactions on Database Systems, {TODS}},
    Number = {4},
    Pages = {517--580},
    Title = {Index-driven similarity search in metric spaces},
    Volume = {28},
    Year = {2003}}

@article{Novak:2011,
    title={Metric index},
    subtitle={An efficient and scalable solution for precise and
           approximate similarity search},
    author={Novak, David and Batko, Michal and Zezula, Pavel},
    journal={Information Systems},
    volume={36},
    number={4},
    pages={721--733},
    year={2011},
    publisher={Elsevier}
}

@article{Fukunaga:1975,
    title={A branch and bound algorithm for computing $k$-nearest neighbors},
    author={Fukunaga, Keinosuke and Narendra, Patrenahalli M},
    journal={{IEEE} Transactions on Computers},
    volume={100},
    number={7},
    pages={750--753},
    year={1975},
    publisher={IEEE}
}

@techreport{Zezula:1996,
    Author = {Pavel Zezula and Paolo Ciaccia and Fausto Rabitti},
    Institution = {\textsc{Hermes} \textsc{Esprit} \textsc{Ltr} Project},
    Number = {7},
    Title = {{M}-tree},
    subtitle = {A Dynamic Index for Similarity Queries in Multimedia Databases},
    Type = {Technical Report},
    Year = {1996}
}

@article{Mico:1996,
    title={A fast branch \& bound nearest neighbour classifier in metric spaces},
    author={Micó, Luisa and Oncina, Jose and Carrasco, Rafael C},
    journal={Pattern Recognition Letters},
    volume={17},
    number={7},
    pages={731--739},
    year={1996},
    publisher={Elsevier}
}

@inproceedings{Edsberg:2010,
    title={Indexing inexact proximity search with distance regression in pivot space},
    author={Edsberg, Ole and Hetland, Magnus Lie},
    booktitle={Proceedings of the Third International Conference on Similarity Search and Applications},
    pages={51--58},
    year={2010},
    organization={ACM}
}

@book{Ziegler:1995,
    title={Lectures on Polytopes},
    author={Ziegler, Günter M},
    series={Graduate Texts in Mathematics},
    volume={152},
    year={1995},
    publisher={Springer}
}

@article{Bustos:2003,
    Author = {Benjamin Bustos and Gonzalo Navarro and Edgar Chávez},
    Journal = {Pattern Recognition Letters},
    Number = {14},
    Pages = {2357--2366},
    Title = {Pivot Selection Techniques for Proximity Searching in Metric Spaces},
    Volume = {24},
    Year = {2003}}

@article{Schneider:2011,
    title={Weak homomorphisms between functorial algebras},
    author={Schneider, Friedrich Martin},
    journal={Demonstratio Mathematica},
    publisher={De Gruyter},
    volume={44},
    number={4},
    pages={801--818},
    year={2011},
}

@Article{Fujiwara:2016,
    author="Fujiwara, Koji and Kapovich, Michael",
    title="On quasihomomorphisms with noncommutative targets",
    journal="Geometric and Functional Analysis",
    year="2016",
    pages="1--42",
    publisher={Springer},
}

@article{Dobos:1995,
    title={A survey of metric preserving functions},
    author={Dobo{\v{s}}, Jozef},
    journal={Questions and Answers in General Topology},
    volume={13},
    number={2},
    pages={129--134},
    year={1995}
}

@BOOK{Dobos:1998,
    author = {Jozef Dobo{\v s}},
    title = {Metric Preserving Functions},
    publisher = {\v Stroffek},
    year = {1998},
}

@ARTICLE{Borsik:1981,
    author = {J{\'{a}}n Bors{\'{i}}k and Jozef Dobo{\v{s}}},
    title = {On a Product of Metric Spaces},
    journal = {Mathematica Slovaca},
    year = {1981},
    volume = {31},
    number = {2},
    pages = {193--205},
}

@article{Corazza:1999,
    title={Introduction to Metric-Preserving Functions},
    author={Corazza, Paul},
    journal={The American Mathematical Monthly},
    volume={106},
    number={4},
    pages={309--323},
    year={1999},
    publisher={Mathematical Association of America}
}

@article{Lowe:1988,
    title={Multivariable functional interpolation and adaptive networks},
    author={Lowe, David and Broomhead, D},
    journal={Complex Syst},
    volume={2},
    pages={321--355},
    year={1988}
}

@article{Blinn:1982,
    author = {Blinn, James F.},
    title = {A Generalization of Algebraic Surface Drawing},
    journal = {ACM Trans.\ Graph.},
    volume = {1},
    number = {3},
    year = {1982},
    pages = {235--256},
    publisher = {ACM},
    address = {New York, NY, USA},
}

@article{Skopal:2007,
    title={Unified framework for fast exact and approximate search in
           dissimilarity spaces},
    author={Skopal, Tom{\'{a}}{\v{s}}},
    journal={ACM Transactions on Database Systems (TODS)},
    volume={32},
    number={4},
    pages={29},
    year={2007},
    publisher={ACM}
}

@book{Kuczma:2009,
    author = {Marek Kuczma},
    title = {An Introduction to the Theory of Functional Equations and
             Inequalities~:\ Cauchy's Equation and Jensen's Inequality},
    publisher = {Birkh\"auser},
    year = {2009},
    edition = {Second},
}

@article{Maxwell:1851,
    title={On the Description of Oval Curves, and those having a plurality of
           Foci},
    subtitle={with remarks by Professor Forbes},
    author={James Clerk Maxwell},
    journal={Proceedings of the Royal Society of Edinburgh},
    volume={2},
    pages={89--91},
    year={1851},
    addendum={Comm.\ by J. Forbes},
    publisher={Cambridge Univ Press}
}

@article{Arora:2012,
    author={Sanjeev Arora and Elad Hazan and Satyen Kale},
    title={The Multiplicative Weights Update Method},
    subtitle={A Meta-Algorithm and Applications},
    journaltitle={Theory of Computing},
    volume={8},
    year={2012}
}

\end{document}